%% file: main.tex
\documentclass[sigplan,10pt,screen,usenames,dvipsnames,svgnames,table]{acmart}

\newif\ifspacehacks
\spacehacksfalse

\usepackage[normalem]{ulem}

\usepackage{epsfig,endnotes}
\usepackage{format}
\usepackage{pierre_macros}
\usepackage{macros}

\usepackage{multicol}
\usepackage{pifont}
\usepackage{array}
\DeclareMathAlphabet{\mathcal}{OMS}{cmsy}{m}{n}

\ifspacehacks
\fi

\newif\iflong
\longtrue %

\newif\ifauth
\authtrue 

\copyrightyear{2020}
\acmYear{2020}
\setcopyright{acmcopyright}
\acmConference[EuroSys '20]{Fifteenth European Conference on Computer Systems}{April 27--30, 2020}{Heraklion, Greece}
\acmBooktitle{Fifteenth European Conference on Computer Systems (EuroSys '20), April 27--30, 2020, Heraklion, Greece}
\acmPrice{15.00}
\acmDOI{10.1145/3342195.3387543}
\acmISBN{978-1-4503-6882-7/20/04}

\begin{document}

\title{State-Machine Replication for Planet-Scale Systems}

\author{Vitor Enes}
\affiliation{%
  \institution{INESC TEC and University of Minho}
}

\author{Carlos Baquero}
\affiliation{%
  \institution{INESC TEC and University of Minho}
}

\author{Tuanir Fran\c{c}a Rezende}
\affiliation{%
  \institution{T\'el\'ecom SudParis}
}

\author{Alexey Gotsman}
\affiliation{%
  \institution{IMDEA Software Institute}
}

\author{Matthieu Perrin}
\affiliation{%
  \institution{University of Nantes}
}

\author{Pierre Sutra}
\affiliation{%
  \institution{T\'el\'ecom SudParis}
}




\input{abstract}

\maketitle

\input{introduction}

\input{smr}

\input{overview}

\input{new_algo}

\input{properties}

\input{epaxos}

\input{proof}

\input{new_opt}

\input{evaluation}

\input{related}

\input{conclusion}

\input{thanks}

\bibliographystyle{plain}
\bibliography{biblio}

\iflong
\clearpage
\appendix
\onecolumn

\input{atlas_full_spec}

\clearpage
\input{app-proof}

\clearpage
\input{app_smr}

\fi

\end{document}

%% file: abstract.tex
\begin{abstract} 
  Online applications now routinely replicate their data at multiple sites
  around the world. In this paper we present \SYS, the first state-machine
  replication protocol tailored for such planet-scale systems. \SYS does not
  rely on a distinguished leader, so clients enjoy the same quality of
  service independently of their geographical locations. Furthermore,
  client-perceived latency improves as we add sites closer to clients. To
  achieve this, \SYS minimizes the size of its quorums using an observation that
  concurrent data center failures are rare. It also processes a high 
  percentage of accesses in a single round trip, even when these 
  conflict. We experimentally demonstrate that \SYS consistently
  outperforms state-of-the-art protocols in planet-scale scenarios. In
  particular, \SYS is up to two times faster than Flexible Paxos with identical
  failure assumptions, and more than doubles the performance of
  Egalitarian Paxos in the YCSB benchmark.
\end{abstract}

%% file: introduction.tex
\section{Introduction}
\labsection{introduction}

Modern online applications run at multiple sites scattered across the
globe: they are now \emph{planet-scale}. Deploying applications in this way
enables high availability and low latency, by allowing clients to access the
closest responsive site. A major challenge in developing planet-scale
applications is that many of their underlying components, such as coordination
kernels~\cite{zookeeper,chubby} and critical databases~\cite{spanner}, require
strong guarantees about the consistency of replicated data.

The classical way of maintaining strong consistency in a distributed service is
{\em state-machine replication (SMR)}~\cite{smr}.  In SMR, a service is defined
by a deterministic state machine, and each site maintains its own local replica
of the machine.  An {\em SMR protocol} coordinates the execution of commands at
the sites, ensuring that they stay in sync. The resulting system is {\em
  linearizable}~\cite{linearizability}, providing an illusion that each command
executes instantaneously throughout the system.

Unfortunately, existing SMR protocols are poorly suited to planet-scale systems.
Common SMR protocols, such as Paxos~\cite{paxos} and Raft~\cite{raft}, are
rooted in cluster computing where a \emph{leader} site determines the ordering
of commands.  This is unfair to clients far away from the leader.  It impairs
scalability, since the leader cannot be easily parallelized and thus becomes a
bottleneck when the load increases.  It also harms availability as, if the
leader fails, the system cannot serve requests until a new one is elected.
Moreover, adding more sites to the system does not help, but on the contrary,
hinders performance, requiring the leader to replicate commands to more sites on
the critical path. This is a pity, as geo-replication has a lot of potential for
improving performance, since adding sites brings the service closer to clients.

To fully exploit the potential of geo-replication, we propose \SYS, a new SMR
protocol tailored to planet-scale systems with many sites spread 
across the world. In particular, \SYS improves client-perceived latency as we
add sites closer to clients. The key to the \SYS design is an observation that
common SMR protocols provide a level of fault-tolerance that is unnecessarily
high in a geo-distributed setting. These protocols allow any minority of sites to
fail simultaneously: e.g., running a typical protocol over $13$ data centers
would tolerate $6$ of them failing. However, natural disasters leading to the loss of a data center are
rare, and planned downtime can be handled by reconfiguring the unavailable site
out of the system~\cite{paxos,zab-reconfiguration}. Furthermore, temporary data
center outages (e.g., due to connectivity issues) typically have a short
duration~\cite{xpaxos}, and, as we confirm experimentally in
\refsection{evaluation}, rarely happen concurrently. For this reason, industry
practitioners assume that the number of concurrent site failures in a
geo-distributed system is low, e.g. $1$ or $2$~\cite{spanner}. Motivated by
this, our SMR protocol allows choosing the maximum number of sites that can fail
($f$) independently of the overall number of sites ($n$), and is optimized for
small values of the former. Our protocol thus trades off higher fault tolerance
for higher scalability\footnote{Apart from data centers being down,
  geo-distributed systems may also exhibit network partitionings, which
  partition off several data centers from the rest of the system. Our protocol
  may block for the duration of the partitioning, which is unavoidable due to
  the CAP theorem~\cite{cap}.}.

In more detail, like previously proposed protocols such as Egalitarian Paxos
(EPaxos)~\cite{epaxos} and Mencius~\cite{mencius}, our protocol is {\em
  leaderless}, i.e., it orders commands in a decentralized way, without relying
on a distinguished leader site. This improves availability and allows serving
clients with the same quality of service independently of their geographical
locations. As is common, our protocol also exploits the fact that commands in
SMR applications frequently commute~\cite{chubby,spanner}, and for the
replicated machine to be linearizable, it is enough that replicas only agree on the
order of non-commuting commands~\cite{gb,gpaxos}. This permits processing a
command in one round trip from the closest replica using a {\em fast path}, e.g., when the command
commutes with all commands concurrently submitted for execution. In the presence
of concurrent non-commuting commands, the protocol may sometimes have to take a
{\em slow path}, which requires two round trips.

Making our protocol offer better latency for larger-scale deployments required
two key innovations in the baseline scheme of a leaderless SMR protocol. First,
the lower latency of the fast path in existing protocols comes with a downside:
the fast path must involve a {\em fast quorum} of replicas bigger than a
majority, which increases latency due to accesses to far-away replicas. For
example, in Generalized Paxos~\cite{gpaxos} the fast quorum consists of at least
$\frac{2n}{3}$ replicas, and in EPaxos of at least $\frac{3n}{4}$
replicas.  To solve this problem, in \SYS the size of the fast quorum is a
function of the number of allowed failures $f$ --  namely,
$\left\lfloor \frac{n}{2} \right\rfloor + f$. Smaller values of $f$ result in
smaller fast quorums, thereby decreasing latency. Furthermore, violating the
assumption the protocol makes about the number of failures may only compromise
liveness, but never safety. In particular, if more than $f$ transient outages
occur, due to, e.g., connectivity problems, \SYS will just block until enough sites are
reachable.

A second novel feature of \SYS is that it can take the fast path
even when non-commuting commands are submitted concurrently, something that is
not allowed by existing SMR protocols~\cite{gpaxos,epaxos}. This permits
processing most commands via the fast path when the conflict rate is
low-to-moderate, as is typical for SMR
applications~\cite{chubby,spanner}. Moreover, when $f = 1$ our protocol {\em
  always} takes the fast path and its fast quorum is a plain majority.

The biggest challenge we faced in achieving the above features --~smaller fast
quorums and a flexible fast-path condition~-- was in designing a correct failure
recovery mechanism for \SYS. Failure recovery is the most subtle part of a SMR
protocol with a fast path because the protocol needs to recover the decisions reached by
the failed replicas while they were short-cutting some of the protocols steps in the fast path. This
is only made more difficult with smaller fast quorums, as a failed
process leaves information about its computations at fewer replicas. \SYS
achieves its performant fast path while having a recovery protocol that is
significantly simpler than that of previous leaderless
protocols~\cite{epaxos,caesar} and has been rigorously proved correct.

As an additional optimization, \SYS also includes a novel mechanism to accelerate the
execution of linearizable reads and reduce their impact on the protocol
stack. This improves performance in read-dominated workloads.

We experimentally evaluate \SYS on Google Cloud Platform using 3 to 13 sites spread around the world.
As new replicas are added closer to clients, \SYS gets faster:
going from 3 to 13 sites, the client-perceived latency is almost cut by half.
We also experimentally compare \SYS with Flexible Paxos~\cite{flexible-paxos} (a variant of Paxos that also allows
selecting $f$ independently of $n$), EPaxos and Mencius.
\SYS consistently outperforms these protocols in planet-scale scenarios.
In particular, our protocol is up to two times faster than Flexible Paxos with identical failure assumptions ($f=1,2$),
and more than doubles the performance of EPaxos in mixed YCSB workloads \cite{ycsb}.

%% file: smr.tex
\section{State-Machine Replication}
\labsection{smr}

We consider an asynchronous distributed system consisting of $n$ processes
$\Proc = \{1, \dots, n\}$. At most $f$ processes may fail by crashing (where
$1 \le f \le \lfloor \frac{n-1}{2} \rfloor$), but processes do not behave
maliciously. In a geo-distributed deployment, each process represents a data
center, so that a failure corresponds to the outage of a whole data center. Failures
of single machines are orthogonal to our concerns and can be masked by
replicating a process within a data center using standard
techniques~\cite{paxos,raft}. We call a majority of processes a {\em (majority)
  quorum}. We assume that the set of processes is static. Classical approaches
can be used to add reconfiguration to our
protocol~\cite{paxos,epaxos}. Reconfiguration can also be used in practice to
allow processes that crash and recover to rejoin the system.

\input{overview-figure}

State-machine replication (SMR) is a common way of implementing strongly consistent replicated
services~\cite{smr}. A service is defined by a deterministic state machine with
an appropriate set of {\em commands}, denoted by $\CMD$. Processes maintain
their own local copy of the state machine, and proxy the access
to the replicated service by client applications (not modeled). 
An {\em SMR protocol} coordinates the
execution of commands at the processes, ensuring that service replicas stay in
sync. The protocol provides a command $\submit(c)$, which allows a process to
submit a command $c \in \CMD$ for execution on behalf of a client. The protocol
may also trigger an event $\execute(c)$ at a process, asking it to apply $c$ to
the local service replica; after execution, the process that submitted the
command may return the outcome of $c$ to the client. Without loss of generality,
we assume that each submitted command is unique.

The strongest property a replicated service implemented using SMR may satisfy is
\emph{linearizability}~\cite{linearizability}. Informally, this means that
commands appear as if executed sequentially on a single copy of the state
machine in an order consistent with the {\em real-time order}, i.e., the order
of non-overlapping command invocations. As observed in \cite{gpaxos,gb}, for the
replicated service to be linearizable, the SMR protocol does not need to ensure that
commands are executed at processes in the exact same order: it is enough to
agree on the order of non-commuting commands.

We now give the specification of the SMR protocol.
We say that commands $c$ and $d$ \emph{commute} if in every state $s$ of the
state machine:
\begin{inparaenum}
\item executing $c$ followed by $d$ or $d$ followed by $c$ in $s$ leads to the same state; and
\item $c$ returns the same response in $s$ as in the state obtained by executing
  $d$ in $s$, and vice versa.
\end{inparaenum}
If commands do not commute, we say that they {\em conflict}%
\footnote{
  Detecting if two commands conflict must be possible without executing them.
    In practice, this information can often be extracted from the API provided by the replicated service.
    In cases when such inference is infeasible, it is always safe to consider that a pair of commands conflict.
}.
We write $c \mapsto_i d$ when $c$ and $d$ conflict and process $i \in \Proc$ executes $c$
before executing $d$. We also define the following {\em real-time order}:
$c \leadsto d$ if $c$ was executed at some process before $d$ was submitted.
Let ${\mapsto} = {\leadsto} \,\cup\, (\bigcup_{i \in \Proc} {\mapsto_i})$. Then, the
specification of the SMR protocol is given by the following properties:
\begin{itemize}
\item [\em \bf Validity.] If a process executes a command $c$, then some process submitted $c$ before.
\item [\em \bf Integrity.] A process executes each command at most once.
\item [\em \bf Ordering.] The relation $\mapsto$ is acyclic.
\end{itemize}
Note that the Ordering property enforces that conflicting commands are executed
in a consistent manner across the system. In particular, it prevents two
conflicting commands from being executed in contradictory orders by different
processes.  If the SMR protocol satisfies the above properties, then the
replicated service implemented using it is linearizable (we prove this in \tra{\ref{appendix:smr}}{B}).
In the following sections we present \SYS, which satisfies the above specification.

%% file: overview-figure.tex
\begin{figure*}[t]
  \input{execution_conf_2.tex}
  \caption{ \labfigure{overview} Example of processing two conflicting commands
    $\msgA$ and $\msgB$ in \SYS with $n=5$ processes and up to $f = 2$ failures.
    We omit the messages implementing consensus and depict this step abstractly by
    the consensus box. }
\end{figure*}
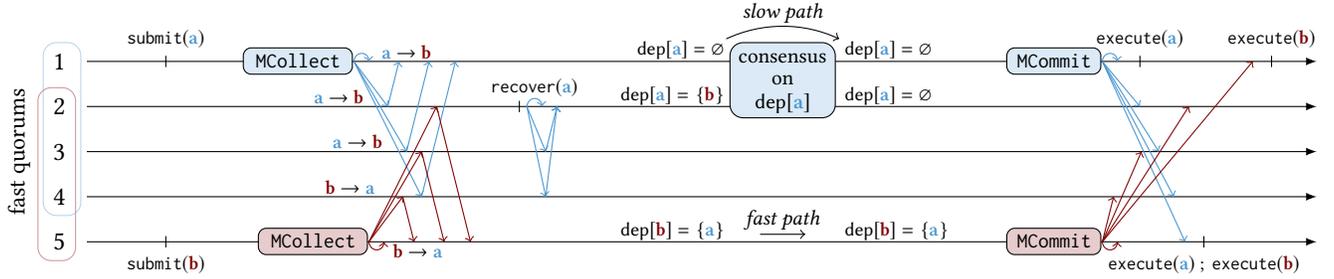

%% file: execution_conf_2.tex
\tikzset{
    ncbar angle/.initial=90,
    ncbar/.style={
        to path=(\tikztostart)
        -- ($(\tikztostart)!#1!\pgfkeysvalueof{/tikz/ncbar angle}:(\tikztotarget)$)
        -- ($(\tikztotarget)!($(\tikztostart)!#1!\pgfkeysvalueof{/tikz/ncbar angle}:(\tikztotarget)$)!\pgfkeysvalueof{/tikz/ncbar angle}:(\tikztostart)$)
        -- (\tikztotarget)
    },
    ncbar/.default=0.5cm,
}
\tikzset{square left brace/.style={ncbar=0.15cm}}
\tikzset{square right brace/.style={ncbar=-0.15cm}}
\tikzset{round left paren/.style={ncbar=0.5cm,out=120,in=-120}}
\tikzset{round right paren/.style={ncbar=0.5cm,out=60,in=-60}}

\begin{tikzpicture}%

  \draw (-4.15,1.7) node [rotate=90] {\small{fast quorums}};
  \draw[\msgAColor!40,rounded corners=5] (-3.83,2) +(0, -1.15) rectangle +(0.5, 1.15) ;
  \draw[\msgBColor!40,rounded corners=5] (-3.9,1.4) +(0, -1.15) rectangle +(0.5, 1.15) ;
  
  \draw[-latex] (-3.25,2.9) node[left,xshift=-4.5]{1} -- (13.1, 2.9) ;
  \draw[-latex] (-3.25,2.3) node[left,xshift=-4.5]{2} -- (13.1, 2.3) ;
  \draw[-latex] (-3.25,1.7) node[left,xshift=-4.5]{3} -- (13.1, 1.7) ;
  \draw[-latex] (-3.25,1.1) node[left,xshift=-4.5]{4} -- (13.1, 1.1) ;
  \draw[-latex] (-3.25,0.5) node[left,xshift=-4.5]{5} -- (13.1, 0.5) ;

  \draw (-2.2,2.9) node{$\rule{.4pt}{1ex}$};
  \draw (-2.2,3.2) node{\scriptsize{$\pcdbroadcast(\msgA)$}};

  \draw (-2.2,0.5) node{$\rule{.4pt}{1ex}$};
  \draw (-2.2,0.2) node{\scriptsize{$\pcdbroadcast(\msgB)$}};

  \draw (2.5,2.3) node{$\rule{.4pt}{1ex}$};
  \draw (2.7,2.55) node{\scriptsize{$\recover(\msgA)$}};

  \draw[fill=\msgAColor!20,rounded corners=3] (-0.42,2.78) +(-0.75, -0.05) rectangle +(0.7, 0.3) +(-0.02,0.13) node{\footnotesize $\KwMCol$};

  \draw[\msgAColor, ->] (0.3,2.9) to[out=85,in=95,distance=4] (0.5, 2.9);
  \draw[\msgAColor, ->] (0.3,2.9) -- (0.75, 2.3);
  \draw[\msgAColor, ->] (0.3,2.9) -- (1, 1.7);
  \draw[\msgAColor, ->] (0.3,2.9) -- (1.2, 1.1);

  \draw[\msgAColor, ->] (0.75, 2.3) -- (0.9,2.9);
  \draw[\msgAColor, ->] (1, 1.7) -- (1.3,2.9);
  \draw[\msgAColor, ->] (1.2, 1.1) -- (1.65,2.9);

  \draw[fill=\msgBColor!20,rounded corners=3] (-0.22,0.38) +(-0.75, -0.05) rectangle +(0.7, 0.3) +(-0.02,0.13) node{\footnotesize $\KwMCol$};
  \draw[\msgBColor, ->] (0.5,0.5) to[out=-85,in=-95,distance=4] (0.7, 0.5);
  \draw[\msgBColor, ->] (0.5,0.5) -- (1.4, 2.3);
  \draw[\msgBColor, ->] (0.5,0.5) -- (1.2, 1.7);
  \draw[\msgBColor, ->] (0.5,0.5) -- (0.95, 1.1);

  \draw[\msgBColor, ->] (1.4, 2.3) -- (1.85,0.5);
  \draw[\msgBColor, ->] (1.2, 1.7) -- (1.5,0.5);
  \draw[\msgBColor, ->] (0.95, 1.1) -- (1.1,0.5);

  \draw (1,3.03) node{\scriptsize $\msgA \to \msgB$};
  \draw (0.1,2.43) node{\scriptsize $\msgA \to \msgB$};
  \draw (0.35,1.83) node{\scriptsize $\msgA \to \msgB$};
  \draw (0.25,1.23) node{\scriptsize $\msgB \to \msgA$};
  \draw (1.15,0.38) node{\scriptsize $\msgB \to \msgA$};

  \draw[\msgAColor, ->] (2.6,2.3) to[out=85,in=95,distance=4] (2.8, 2.3);
  \draw[\msgAColor, ->]  (2.6,2.3) -- (2.85, 1.7);
  \draw[\msgAColor, ->]  (2.6,2.3) -- (2.85, 1.1);
  \draw[\msgAColor, ->]  (2.85, 1.7) -- (3, 2.3);
  \draw[\msgAColor, ->]  (2.85, 1.1) -- (3, 2.3);

  \draw[fill=\msgAColor!20,rounded corners=5] (6,2.65) +(-0.7, -0.5) rectangle +(0.7, 0.5) +(0,0.3) node[black]{\footnotesize consensus} +(0,0) node[black]{\footnotesize on}+(0,-0.3) node[black]{\footnotesize $\KwMsgCcl[\msgA]$};
  \draw[Black] (6,0.8) node{\footnotesize \emph{fast path}};
  \draw[Black, ->]  (5.7,0.6) -- (6.3,0.6);

  \draw[Black] (6,3.55) node{\footnotesize \emph{slow path}};
  \draw[Black, ->]  (5.25,3.2) to[out=25,in=155,distance=15] (6.75,3.2);

  \draw (5.35,3.05) node[left]{\scriptsize $\KwMsgCcl[\msgA] = \emptyset$};
  \draw (5.35,2.45) node[left]{\scriptsize $\KwMsgCcl[\msgA] = \{\msgB\}$};
  \draw (5.35,0.65) node[left]{\scriptsize $\KwMsgCcl[\msgB] = \{\msgA\}$};

  \draw (6.7,3.05) node[right]{\scriptsize $\KwMsgCcl[\msgA] = \emptyset$};
  \draw (6.7,2.45) node[right]{\scriptsize $\KwMsgCcl[\msgA] = \emptyset$};
  \draw (6.7,0.65)  node[right]{\scriptsize $\KwMsgCcl[\msgB] = \{\msgA\}$};

  \draw[fill=\msgAColor!20,rounded corners=3] (9.63,2.78) +(-0.65, -0.05) rectangle +(0.6, 0.3) +(-0.02,0.13) node{\footnotesize $\KwMCom$};
  \draw[\msgAColor, ->]  (10.25,2.9) to[out=85,in=95,distance=4] (10.45, 2.9);
  \draw[\msgAColor, ->]  (10.25,2.9) -- (10.75, 2.3);
  \draw[\msgAColor, ->]  (10.25,2.9) -- (11.0, 1.7);
  \draw[\msgAColor, ->]  (10.25,2.9) -- (11.2, 1.1);
  \draw[\msgAColor, ->]  (10.25,2.9) -- (11.35, 0.5);
  
  \draw[fill=\msgBColor!20,rounded corners=3] (9.63,0.38) +(-0.65, -0.05) rectangle +(0.6, 0.3) +(-0.02,0.13) node{\footnotesize $\KwMCom$};
  \draw[\msgBColor, ->]  (10.25,0.5) to[out=-85,in=-95,distance=4] (10.45, 0.5);
  \draw[\msgBColor, ->]  (10.25,0.5) -- (10.4, 1.1);
  \draw[\msgBColor, ->]  (10.25,0.5) -- (10.77, 1.7);
  \draw[\msgBColor, ->]  (10.25,0.5) -- (11.4, 2.3);
  \draw[\msgBColor, ->]  (10.25,0.5) -- (12.25, 2.9);

  \draw (10.75,2.9) node{$\rule{.4pt}{1ex}$};
  \draw (10.75,3.2) node{\scriptsize{$\execute(\msgA)$}};
  \draw (12.5,2.9) node{$\rule{.4pt}{1ex}$};
  \draw (12.5,3.2) node{\scriptsize{$\execute(\msgB)$}};

  \draw (11.6,0.5) node{$\rule{.4pt}{1ex}$};
  \draw (11.6,0.2) node{\scriptsize{$\execute(\msgA)\ ;~\execute(\msgB)$}};

\end{tikzpicture}

%% file: overview.tex
\section{The \SYS Protocol}
\labsection{algo}

To aid understanding, we first illustrate by example the message flow of the
\SYS protocol (\refsection{algo:bigpicture}), which corresponds to a common
structure of leaderless SMR protocols~\cite{epaxos}. We then describe the
protocol in detail (\refsection{algo:details}).

\subsection{Overview}
\labsection{algo:bigpicture}

\reffigure{overview} illustrates how \SYS processes two conflicting commands,
$\msgA$ and $\msgB$, with $n = 5$ processes and at most $f = 2$ failures. At a
given process, a command usually goes through several 
{\em phases}: the initial phase $\KwPhA$, then $\KwPhB$, $\KwPhD$ and $\KwPhE$
(an additional phase $\KwPhF$ is used when handling failures).

Command $\msgA$ starts its journey when $\pcdbroadcast(\msgA)$ is invoked at process $1$.
We call process $1$ the initial \emph{coordinator} of $\msgA$.
This coordinator is initial because, if it fails or is slow, another process may take over.
Command $\msgA$ then enters the $\KwPhB$ phase at process $1$, whose goal is to compute the set of commands
that are \emph{dependencies} of $\msgA$, denoted by $\KwMsgCcl[\msgA]$.
These dependencies are later used to determine the order of execution of conflicting commands.
To compute dependencies, process $1$ sends an $\KwMCol$ message containing command $\msgA$
to a {\em fast quorum} of processes, which is at least a majority but may 
be bigger. In our example the fast quorum picked by $1$ is $\{1, 2, 3, 4\}$.

Each process in the fast quorum returns the set of commands conflicting with $\msgA$ that it received before $\msgA$.
In \reffigure{overview}, $\to$ indicates the order in which processes receive commands.
For instance, process $4$ receives $\msgB$ first, whereas the other fast-quorum processes do not receive any command before $\msgA$.
Based on the replies, process $1$ computes the value of $\KwMsgCcl[\msgA]$
(as described in the next section); in our example this happens to be $\emptyset$.

If a coordinator of a command is suspected to have failed, another process may
try to take over.  In \reffigure{overview}, process $2$ suspects $1$ and becomes
another coordinator of $\msgA$, denoted by $\mathtt{recover}(\msgA)$.
Process $2$ contacts a majority quorum of processes $\{2, 3, 4\}$ and computes
its own version of the dependencies of $\msgA$: $\KwMsgCcl[\msgA] = \{\msgB\}$.

Dependencies are used to determine the order in which conflicting commands are
executed, and all processes have to execute conflicting commands in the same
order. To ensure this, the coordinators of command $\msgA$ need to reach a
consensus on the value of $\KwMsgCcl[\msgA]$. This is implemented using an
optimized variant of single-decree Paxos~\cite{paxos}, with all $n$ processes
acting as acceptors. In our example, this makes the processes agree on
$\KwMsgCcl[\msgA] = \emptyset$. The use of consensus represents the {\em slow
  path} of the protocol.

If a coordinator can ensure that all the values that can possibly be proposed to consensus are the same,
then it can take the {\em fast path} of the protocol, avoiding the use of consensus.
In \reffigure{overview}, this is the case for process $5$ coordinating command $\msgB$.
For a process to take the fast path, we require it to receive a response
from every process in the fast quorum, motivating the name of the latter.

After consensus or the shortcut via the fast path, a coordinator of a command
sends its final dependencies to other processes in an $\KwMCom$ message.
A process stores these dependencies and marks the command as having entered the $\KwPhD$ phase.
A command can be executed (and thereby transition to the $\KwPhE$ phase) only after all
its dependencies are in the $\KwPhD$ or $\KwPhE$ phases. Since in our example
$\KwMsgCcl[\msgA] = \emptyset$, processes can
execute command $\msgA$ right after receiving its final dependencies ($\emptyset$).
This is exploited by processes $1$ and $2$ in \reffigure{overview}. 
However, as $\KwMsgCcl[\msgB] = \{\msgA\}$, processes
must delay the execution of $\msgB$ until $\msgA$ is
executed. This is the case for processes $3$, $4$
and $5$ in \reffigure{overview}. Such an execution mechanism guarantees that
the conflicting commands $\msgA$ and $\msgB$ are executed in the 
same order at all processes.

%% file: new_algo.tex
\subsection{Protocol in Detail}
\labsection{algo:details}

\refalg{atlas} specifies the \SYS protocol at process $i \in \Proc$ in the
failure-free case. We assume that self-addressed protocol messages are delivered
immediately.

\subsubsection{Start phase}
A client submits a command $c \in \CMD$ by invoking $\pcdbroadcast(c)$ at one of
the processes running \SYS, which will serve as the initial command
coordinator. When $\pcdbroadcast(c)$ is invoked at a process $i$
(line~\ref{algo:new:request}), this coordinator first assigns
to command $c$ a unique identifier -- a pair $\tup{i, l}$ where $l - 1$ is the number of
commands submitted at process $i$ before $c$. In the following we denote
the set of all identifiers by $\ID$. At the bottom of \refalg{atlas}, we
summarize the data maintained by each process for a command with identifier
$\id \in \ID$. In particular, the mapping $\KwMsgM$ stores the payload of the
command, and the mapping $\KwMsgPh$ tracks the progress of the command through
phases. For brevity, the name of the phase written in lower case also denotes
all the identifiers in that phase, e.g.,
$\KwPhAset = \{\id \in \ID \mid \KwMsgPh[\id] = \KwPhA\}$.

Once the coordinator assigns an identifier to $c$, the command starts its $\KwPhB$
phase, whose goal is to compute a set of identifiers that are the
\emph{dependencies} of $c$. At the end of this phase, the coordinator sends an
$\KwMCom(\id, c, D)$ message including the computed dependencies $D$. Before
this, it agrees with other possible coordinators on the same final value of $D$,
resulting in the following invariant.
\begin{invariant}\label{inv:cons}
  For any two messages $\KwMCom(\id, c, D)$ and
  $\KwMCom(\id, c', D')$ sent, $c = c'$ and $D = D'$.
\end{invariant}
Hence, each identifier is associated with a unique command and final set of dependencies. 
The key property of dependencies is that, for any two distinct conflicting
commands, one has to be a dependency of the other.
This is stated by the following invariant.
\begin{invariant}\label{inv:conf}
  Assume that messages $\KwMCom(\id, c, D)$ and $\KwMCom(\id', c', D')$ have
  been sent. If $\id\not=\id'$ and $\conflict(c, c')$ then either $\id' \in D$
  or $\id \in D'$, or both.
\end{invariant}
This invariant is key to ensure that conflicting commands are executed in the
same order at all processes, since we allow processes to execute commands that
are not a dependency of each other in any order. We next explain how \SYS
ensures the above invariants.

\input{atlas_algo_spec}

\subsubsection{Collect phase}
To compute the dependencies of a command $c$, its coordinator first computes the
set of commands it knows about that conflict with $c$ (denoted by $\past$,
line~\ref{algo:new:past}) using a function
$\conflicts(c) = \{ \id \not \in \KwPhAset \mid \conflict(c, \KwMsgM[\id]) \}$.
The coordinator then picks a fast quorum $Q$ of size $\floor{\frac{n}{2}} + f$
that includes itself (line~\ref{algo:new:fast-quorum})
and sends an $\KwMCol$ message with the information it computed to all processes in $Q$.

Upon receiving an $\KwMCol$ message from the coordinator, a process in the fast
quorum computes its contribution to $c$'s dependencies
as the set of commands that conflict with $c$, combined with $\past$ (line~\ref{algo:new:conflicts}).
The process stores the computed dependencies, command $c$ and the fast quorum
$Q$ in mappings $\KwMsgCcl$, $\KwMsgM$ and $\KwMsgQ$, respectively, and sets the
command's phase to $\KwPhB$. The process then replies to the coordinator with an
$\KwMColA$ message, containing the computed dependencies
(line~\ref{algo:new:collect-ack}).

Once the coordinator receives an $\KwMColA$ message
from all processes in the fast quorum
(line~\ref{algo:new:collect-ack-pre}),
it computes the dependencies for the command as the union of all reported dependencies $D = \bigcup_Q \dep = \bigcup \{ \dep_j \mid j \in Q \}$ (line~\ref{algo:new:fp-proposal}).
Since a fast quorum contains at least a majority of processes,
the following property implies that this computation maintains Invariant~\ref{inv:conf}.
\begin{property}\label{lemma:ordering}
  Assume two conflicting commands with identifiers $id$ and $id'$ and
  dependencies $D$ and $D'$ computed as in
  line~\ref{algo:new:fp-proposal} over majority quorums. Then either $\id' \in D$ or $\id \in D'$, or
  both.
\end{property}
\paragraph{\it Proof.}
Assume that the property does not hold: there are two conflicting commands with
distinct identifiers $id$ and $id'$ and dependencies $D$ and $D'$ such that
$id' \not \in D$ and $id \not \in D'$. We know that $D$ was computed over some majority
$Q$ and $D'$ over some majority $Q'$.  Since $id' \not \in D$, we have: {\em
  (i)} the majority $Q$ observed $id$ before $id'$. Similarly, since
$id \not \in D'$: {\em (ii)} the majority $Q'$ observed $id'$ before $id$.
However, as majorities $Q$ and $Q'$ must intersect, we cannot have both {\em
  (i)} and {\em (ii)}. This contradiction shows the required.\qed

\smallskip
\smallskip

For example, in \reffigure{overview} coordinator $5$ determines the dependencies
for \msgB using the computation at line~\ref{algo:new:fp-proposal} (coordinator
$1$ uses an optimized version of this computation presented
in~\refsection{execution:improvements}).

After computing the command's dependencies, its coordinator decides to either take the fast path
(line~\ref{algo:new:fp-condition}) or the slow path (line~\ref{algo:new:sp-else}).
Both fast and slow paths end with the coordinator sending an $\KwMCom$ message containing
the command and its final dependencies.

\addtolength{\belowcaptionskip}{0.2cm}%
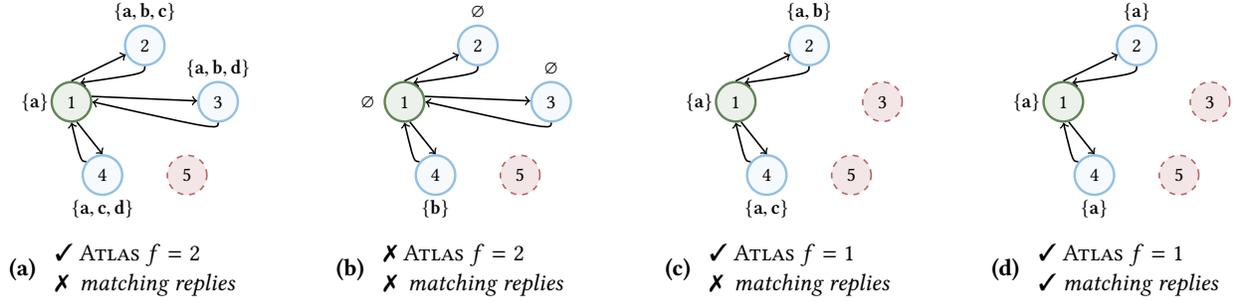
\begin{figure*}[t]
  \input{fast-path-examples}
  \caption{
    \labfigure{fp-ex}
    Examples in which the fast path is taken \rmark{} or not \wmark{},
    for both \SYS and protocols that
    require \emph{matching replies} from fast-quorum processes, such as EPaxos~\cite{epaxos}.
    All examples consider $n = 5$ processes while tolerating $f$ faults.
    The coordinator is always process $1$, and
    circles with a solid line represent the processes that are part of the fast quorum.
    Next to each process we depict the set of dependencies
    sent to the coordinator (e.g. \tikznodelabel{a,b}).
  }
\end{figure*}
\addtolength{\belowcaptionskip}{-0.2cm}%

\subsubsection{Slow path}
If the coordinator of a command is suspected to have failed,
another process may try to take over its job and compute a different set of
dependencies. Hence, before an $\KwMCom$ message is sent, processes must reach an
agreement on its contents to satisfy Invariant~\ref{inv:cons}. They can always
achieve this by running a consensus protocol -- this is the slow path of
\SYS. Consensus is implemented using single-decree (Flexible)
Paxos~\cite{flexible-paxos}. For each identifier we allocate ballot numbers to
processes round-robin, with ballot $i$ reserved for the initial coordinator $i$
and ballots higher than $n$ for processes that try to take over. Every process
stores for each identifier $\id$ the ballot number $\KwMsgR[\id]$ it is
currently participating in and the last ballot $\KwMsgW[\id]$ in which it
accepted a proposal (if any).

When the initial coordinator $i$ decides to go onto the slow path, it performs
an analog of Paxos Phase 2: it sends an $\KwMCons$ message with its proposal and
ballot $i$ to a {\em slow quorum} that includes itself
(line~\ref{algo:new:slow-quorum})\footnote{%
  The initial coordinator $i$ can safely skip Paxos Phase 1: since processes
  perform recovery with ballots higher than $n$, no proposal with a ballot lower
  than $i$ can ever be accepted.}.
Following Flexible Paxos~\cite{flexible-paxos}, the size of the slow quorum is
only $f+1$, rather than a majority like in classical Paxos. This minimizes the
additional latency incurred on the slow path in exchange for using larger
quorums in recovery (as described below). As usual in Paxos, a process accepts
an $\KwMCons$ message only if its $\KwMsgR[\id]$ is not greater than the ballot
in the message (line~\ref{algo:new:consensus-pre}). Then it stores the proposal,
sets $\KwMsgR[\id]$ and $\KwMsgW[\id]$ to the ballot in the message, and replies
to the coordinator with $\KwMConsA$. Once the coordinator gathers $f + 1$ such
replies (line~\ref{algo:new:consensus-ack-pre}), it is sure that its proposal
will survive the allowed number of failures $f$, and it thus broadcasts the
proposal in an $\KwMCom$ message (line~\ref{algo:new:consensus-end}).

\subsubsection{Fast path}
The initial coordinator of a command can avoid consensus when it can ensure that
any process performing recovery will propose the same set of dependencies to
consensus~\cite{fast-paxos} -- this is the fast path of \SYS, in which a command
is committed after a single round trip to the closest fast quorum
(line~\ref{algo:new:fp}). In order to take the fast path, previous SMR
protocols, such as Generalized Paxos~\cite{gpaxos} and EPaxos~\cite{epaxos},
require fast-quorum replies to match exactly. One of the key innovations of
\SYS is that it is able to take the fast path even if this is not the case,
e.g., when conflicting commands are submitted concurrently. 
This feature significantly improves performance
in practice
(\refsection{evaluation}).

In more detail, the coordinator takes the fast path if every dependency reported
by some fast-quorum process is actually reported by at least $f$ such
processes. This is expressed by the condition
$\bigcup_{Q} \dep = \quorumcup{f}{Q} \dep$ in line~\ref{algo:new:fp-condition},
where
\begin{align*}
\quorumcup{f}{Q} \dep & = \{ \id \mid \multiplicity(\id) \geq f \};
\\
\multiplicity(\id) & = |\{ j \in Q \mid \id \in \dep_j \}|.
\end{align*}

\reffigure{fp-ex} contains several examples that illustrate the flexibility of
the above fast-path condition. All examples consider $n = 5$ processes while
tolerating varying numbers of faults $f$. The example in \reffigure{fp-ex1}
considers \SYS $f=2$. The coordinator of some command, process $1$, picks a
fast quorum $Q = \{1,2,3,4\}$ of size $\floor{\frac{n}{2}} + f = 4$. It receives
replies $\dep_1 = \tikznodelabel{a}$, $\dep_2 = \tikznodelabel{a,b,c}$,
$\dep_3 = \tikznodelabel{a,b,d}$, $\dep_4 = \tikznodelabel{a,c,d}$. The
coordinator then computes
$\quorumcup{\hspace{-1pt}2}{Q} \dep = \tikznodelabel{a, b, c, d}$,
i.e., all the dependencies reported at least twice.
Since $\bigcup_Q \dep = \quorumcup{\hspace{-1pt}2}{Q} \dep$,
the coordinator takes the fast path. This is not the case for
the example in \reffigure{fp-ex2} where
$\bigcup_Q \dep = \tikznodelabel{b} \not= \varnothing = \quorumcup{\hspace{-1pt}2}{Q} \dep$
($\mathbf{b}$ is excluded from
$\quorumcup{\hspace{-1pt}2}{Q} \dep$
because $\multiplicity(\mathbf{b}) = 1$). In this case the coordinator has to take the slow path.
Back in \reffigure{overview}
we had the same situation: coordinator $1$ had to take the slow path
because dependency \msgB was declared solely by process $4$.
On the other hand, coordinator $5$ was able to take the fast path, because
dependency \msgA was declared by $f=2$ processes: $2$ and $3$.

Notice that in \reffigure{fp-ex1}, the coordinator takes the fast path even
though dependencies reported by processes do not match, a situation which may
arise when conflicting commands are submitted concurrently. 
Furthermore, when $f=1$ we have
$\{ \id \mid \multiplicity(\id) < f \} = \varnothing$, so that the fast-path
condition in line~\ref{algo:new:fp-condition} always holds. Hence, \SYS $f=1$
\emph{always} takes the fast path, as is the case in
\reffiguretwo{fp-ex3}{fp-ex4}. In contrast, EPaxos is able to take the fast path
only in \reffigure{fp-ex4}, since it is the only example in which fast-quorum
replies match.

\subsubsection{Recovery idea}
The initial coordinator of a command may fail or be slow to respond,
in which case
\SYS allows a process to take over its role and recover the command and its
dependencies.  We start by describing the idea of the most subtle part of this
mechanism -- recovering decisions reached by failed coordinators via the fast
path.

Let $D = \bigcup_{Q} \dep = \quorumcup{f}{Q} \dep$ be some fast-path
proposal (line~\ref{algo:new:fp}). By definition of
$\quorumcup{f}{Q} \dep$, each $\id \in D$ was reported in the $\KwMColA$
message of at least $f$ fast-quorum processes. It follows that $D$ can be
obtained without $f - 1$ of those processes by taking the union of the
dependencies reported by the remaining processes. Moreover, as the initial
coordinator is always part of the fast quorum and each process in the quorum
combines its dependencies with the ones declared by the coordinator (i.e.,
$\past$ in line~\ref{algo:new:conflicts}), the latter is also not necessary to
obtain $D$. Thus, the proposal $D$ can be obtained without $f$ fast-quorum
processes including the initial coordinator (e.g., if the processes fail),
by combining the dependencies
reported by the remaining $\floor{\frac{n}{2}} + f - f = \floor{\frac{n}{2}}$
processes. The following property captures this observation.
\begin{property}\label{property:set-union}
  Any fast-path proposal can be obtained by taking the union of the dependencies
  sent in $\KwMColA$ by at least $\floor{\frac{n}{2}}$ fast-quorum processes
  that are not the initial coordinator.
\end{property}

As an example, assume that after the fast path is taken in \reffigure{fp-ex1},
$f = 2$ processes inside the fast quorum fail, one of them being the coordinator, process $1$.
Independently of which $\floor{\frac{n}{2}} = 2$
fast-quorum processes survive,
the proposal is always recovered by set union:
$\bigcup_{\{2,3\}} \dep = \bigcup_{\{2,4\}} \dep = \bigcup_{\{3,4\}} \dep = \tikznodelabel{a, b, c, d}$.

In the case of \reffigure{fp-ex2} it is unsafe to take the fast path since the
proposal may not be recoverable: the failure of process $4$ would lead to losing the
dependency $\mathbf{b}$, since this dependency was reported exclusively by
this process.

\input{atlas_recovery_spec}

\subsubsection{Recovery in detail}
A process takes over as the coordinator for some command with identifier $\id$
by calling $\KwRec(\id)$ (line~\ref{algo:new:rec} in \refalg{recovery}).  In
order to find out if a decision on the dependencies of $\id$ has been reached in
consensus, the new coordinator first performs an analog of Paxos Phase 1. It
picks a ballot number it owns higher than any it participated in so far
(line~\ref{algo:new:my-ballot}) and sends an $\KwMR$ message with this ballot to
all processes.

Upon the receipt of such a message, in case $\id$ is already committed or
executed (line~\ref{algo:new:rec-committed}), the process notifies the new
coordinator with an $\KwMCom$ message. Otherwise, as is standard in Paxos, the
process accepts the $\KwMR$ message only if the ballot in the message is greater
than its $\KwMsgR[\id]$ (line~\ref{algo:new:rec-pre}). In this case, if the
process is seeing $\id$ for the first time (line~\ref{algo:new:rec-first-see}),
it computes its contribution to $\id$'s dependencies as the set of conflicting
commands (line~\ref{algo:new:conflicts-rec}). Then, the process sets $\KwMsgR[\id]$ to the new ballot
and $\KwMsgPh[\id]$ to $\KwPhF$.
Finally, the process replies with an $\KwMRA$ message
containing all the information it has regarding $\id$: the corresponding command
($\KwMsgM$), its current set of dependencies ($\KwMsgCl$),
the ballot at which these were previously accepted ($\KwMsgW$), and the fast quorum ($\KwMsgQ$).
Note that $\KwMsgQ[\id] = \emptyset$ if the process did not see the initial
$\KwMCol$ message, and $\KwMsgW[\id] = 0$ if the process has not yet accepted any
consensus proposal.

In the $\KwMRA$ handler (line~\ref{algo:new:mrec-ack-pre}), the new coordinator
computes its proposal given the information provided by processes and sends this
proposal in an $\KwMCons$ message to all processes. As in Flexible Paxos, the
new coordinator waits for $n - f$ $\KwMRA$ messages. This guarantees that, if a
quorum of $f + 1$ processes accepted an $\KwMCons$ message with a proposal
(which could have thus been sent in an $\KwMCom$ message), the new coordinator
will find out about this proposal. To maintain Invariant~\ref{inv:cons}, if any
process previously accepted a consensus proposal
(line~\ref{algo:new:previous-proposal}), by the standard Paxos
rules~\cite{paxos,flexible-paxos}, the coordinator selects the proposal accepted
at the highest ballot (line~\ref{algo:new:highest-proposal}).

If no consensus proposal has been accepted before, the new coordinator checks
whether any of the processes that replied has seen the initial $\KwMCol$
message, by looking for any non-empty fast quorum
(line~\ref{algo:new:qk-check}). If the fast quorum is known, depending on
whether the initial coordinator replied or not, there are two possible cases
that we describe next.

\emph{1) The initial coordinator replies to the new one ($\id.1 \in Q$,
  line~\ref{algo:new:coordinator-check}).}  In this case the initial coordinator
has not taken the fast path before receiving the $\KwMR$ message from the new
one, as it would have replied with $\KwMCom$ instead of $\KwMRA$
(line~\ref{algo:new:rec-committed-send}). It will also not take the fast path in
the future, since when processing the $\KwMR$ message it sets the command phase
to $\KwPhF$ (line~\ref{algo:new:recovery-phase}), which invalidates the
$\KwMColA$ precondition (line~\ref{algo:new:collect-ack-pre}). Since the initial
coordinator never takes the fast path, the new coordinator can choose the
command's dependencies in any way, as long as it maintains
Invariant~\ref{inv:conf}. By Property~\ref{lemma:ordering}, this is satisfied
if the coordinator chooses the set union of the dependencies declared by at
least a majority of processes. Hence, the new coordinator takes the union of the
dependencies reported by the
$n - f \geq n - \floor{\frac{n-1}{2}} \geq \floor{\frac{n}{2}} + 1$ processes in
$Q$ (line~\ref{algo:new:rec-union}).

\emph{2) The initial coordinator does not reply to the new one ($\id.1 \not\in Q$,
  line~\ref{algo:new:coordinator-check}).}  In this case the initial coordinator could have taken the
fast path and, if it did, the new coordinator must propose the same dependencies.
Given that the recovery quorum $Q$ has size $n - f$ 
and the fast quorum $Q^0_k$ has size $\floor{\frac{n}{2}} + f$, 
the set of processes $Q' = Q \cap Q^0_k$ (line~\ref{algo:new:rec-union})
contains at least $\floor{\frac{n}{2}}$ fast-quorum processes (distinct from the
initial coordinator, as it did not reply).
Furthermore, recall that when a process from $Q'$ replies to the new coordinator, it sets 
the command phase to $\KwPhF$ (line~\ref{algo:new:recovery-phase}),
which invalidates the $\KwMCol$ precondition (line~\ref{algo:new:collect-pre}).
Hence, if the initial coordinator took the fast path, then each process
in $Q'$ must have processed its $\KwMCol$ before the $\KwMR$ of the new
coordinator, and reported in the latter the dependencies from the former.
Then using Property~\ref{property:set-union}, the new coordinator recovers the
fast-path proposal by taking the union of the dependencies from the processes in
$Q'$ (line~\ref{algo:new:rec-union}). It can be shown that, even if the initial
coordinator did not take the fast path, this computation maintains
Invariant~\ref{inv:conf}, despite $Q'$ containing only $\floor{\frac{n}{2}}$
processes and Property~\ref{lemma:ordering} requiring a majority of them. This
is for the same reason this number of processes is sufficient in
Property~\ref{property:set-union}: dependencies declared by the initial
coordinator are included into those declared by other fast-quorum processes
(line~\ref{algo:new:conflicts}).

It remains to address the case in which the process performing the recovery
observes that no process saw the initial fast quorum, and consequently the
submitted command (line~\ref{algo:new:noop}). For instance, suppose that
process $i$ sends an $\KwMCol(\id, c, \_, \_)$ only to process $j$ and then fails.
Further, assume that $j$ receives another $\KwMCol(\_, c', \_, \_)$ from process
$k$, replies with a dependency set that includes the identifier $\id$ of $c$,
and also fails.  Now, process $k$ cannot execute $c'$ without executing $c$
(since $c$ is a dependency of $c'$), and it cannot execute $c$ because its
payload has been lost.  We solve this issue similarly to EPaxos:
if a process takes over as the new coordinator and cannot find the associated
payload, it may replace it by a special $\KwNoop$ command
(line~\ref{algo:new:noop}) that is not executed by the protocol and conflicts
with all commands.  With this, the final command for some identifier can take
two possible values: the one submitted (line~\ref{algo:new:request}) or
$\KwNoop$. It is due to this that we include the command payload in addition to
its dependencies into consensus messages associated with a given identifier
(e.g., line~\ref{algo:new:sp}), thus ensuring that a unique payload will be
chosen (Invariant~\ref{inv:cons}). Due to the possible replacement of a command
by a $\KwNoop$, the protocol actually ensures the following weakening of 
Invariant~\ref{inv:conf}, which is still sufficient to establish its correctness.
\begin{custominvariant}{\textsc{2$'$}}\label{inv:conf-a}
  Assume that messages $\KwMCom(\id, c, D)$ and $\KwMCom(\id', c', D')$ have
  been sent. If $\id\not=\id'$, $\conflict(c, c')$, $c \not= \KwNoop$ and
  $c' \not= \KwNoop$, then either $\id' \in D$ or $\id \in D'$, or both.
\end{custominvariant}

\subsubsection{Command execution}
\refalg{execution} describes a background task employed by \SYS that is
responsible for executing commands after they are committed.  This task runs in
an infinite loop trying to execute a {\em batch} of commands.
We define a batch as the smallest set of
committed identifiers $S \subseteq \KwPhDset$ such that, for each identifier
$\id\in S$,
its dependencies are in the batch or already executed: $\KwMsgCl[\id] \subseteq S \cup \KwPhEset$ (line~\ref{algo:new:exec-start}).
This ensures that a command can only be executed after its dependencies or
in the same batch with them, which yields the following invariant.
\begin{invariant}\label{inv:del}
Assume two commands $c$ and $c'$ with identifiers $\id$ and $\id'$, respectively.
  If a process executes a batch of commands containing $c$ before executing a batch containing $c'$,
  then $\id' \notin \KwMsgCl[\id]$.
\end{invariant}

As processes agree on the dependencies of each command
(Invariant~\ref{inv:cons}), the batch in which a
command is executed is equal in every process, as reflected in following invariant.

\input{atlas_execution_spec}

\begin{invariant}\label{inv:del2}
  If a process executes command $c$ in batch $S$
  and another process executes the same command $c$ in batch $S'$,
  then $S = S'$.
\end{invariant}

Inside a batch, commands are ordered according to some fixed total order $<$ 
on identifiers (line~\ref{algo:new:exec-sort}).
This guarantees that conflicting commands are executed in a consistent
order across all processes.

Consider again the example in \reffigure{overview}, where the final dependencies
are $\KwMsgCcl[\msgA] = \varnothing$ and $\KwMsgCcl[\msgB] = \{\msgA\}$. 
There are two cases, depending on the order in which processes commit the
commands \msgA and \msgB: 
{\setlength{\leftmargini}{10pt}
\begin{itemize}
\item \msgA then \msgB: at processes $1$ and $2$.  When the command \msgA is
  committed, the processes execute it in a singleton batch, as it has
  no dependencies. When later the command \msgB is committed, the processes
  execute it in a singleton batch too, since its only dependency \msgA has
  already been executed.
\item \msgB then \msgA: at processes $3$, $4$ and $5$.  When the command \msgB
  is committed, the processes cannot execute it, as its dependency \msgA has not
  yet been committed. When later the command \msgA is committed, the processes
  execute two singleton batches: first \msgA, then \msgB.
\end{itemize}} 
Note that \msgA is executed before \msgB in both cases, thus
ensuring a consistent execution order across processes.

Assume now we had final dependencies $\KwMsgCcl[\msgA] = \{\msgB\}$ and
$\KwMsgCcl[\msgB] = \{\msgA\}$. In this case, independently of the order in
which processes commit the commands, a batch will only be formed when both are
committed.
Since all processes will form the same batch containing both \msgA and \msgB,
these commands will be executed in a predefined order on their identifiers, again
ensuring a consistent execution order.

%% file: atlas_algo_spec.tex
\IncMargin{0.2em}
\begin{algorithm}[!t]
    \algosetup

    \setcounter{AlgoLine}{0}
    \SubAlgo{\Fun $\KwPCD(c)$\label{algo:new:request}}{
      $\id \leftarrow \tup{i, \mathsf{min}\{l \mid \tup{i, l} \in \KwPhAset \}}$ \\
      $\past \leftarrow \conflicts(c)$ \label{algo:new:past}\;
      $Q \leftarrow \fastq(i)$ \label{algo:new:fast-quorum}\;
      \Send $\KwMCol(\id, c, \past, Q)$ \To $Q$ \label{algo:new:collect-past} \;
    }

    \SubAlgo{\RecA $\KwMCol(\id, c, \past, Q)$ \RecB $j$}{
      \Pre $\id \in \KwPhAset$ \label{algo:new:collect-pre}\;
      $\KwMsgCcl[\id] \leftarrow \conflicts(c) \union \past$ \label{algo:new:conflicts}\;
      $\KwMsgM[\id] \leftarrow c$;
      $\KwMsgQ[\id] \leftarrow Q$\;
      $\KwMsgPh[\id] \leftarrow \KwPhB$\;
      \Send $\KwMColA(\id, \KwMsgCcl[\id])$ \To $j$ \label{algo:new:collect-ack}
    }

    \SubAlgo{\RecA $\KwMColA(\id, \dep_j)$ \RecBAll $j \in Q$\label{algo:new:collect-ack-receive}}{
      \Pre $\id \in \KwPhBset \land Q = \KwMsgQ[\id]$ \label{algo:new:collect-ack-pre}\;
      $D \leftarrow \bigcup_{Q} \dep$ \label{algo:new:fp-proposal} \;
      \If{$\bigcup_{Q} \dep = \quorumcup{f}{Q} \dep$\label{algo:new:fp-condition}}{
        \Send $\KwMCom(\id, \KwMsgM[\id], D)$ \ToAll \label{algo:new:fp}
      }
      \Else{\label{algo:new:sp-else}
        $Q' \leftarrow \slowq(i)$ \label{algo:new:slow-quorum}\;
        \Send $\KwMCons(\id, \KwMsgM[\id], D, i)$ \To $Q'$ \label{algo:new:sp}
      }
    }

    \SubAlgo{\RecA $\KwMCons(\id, c, D, b)$ \RecB $j$} {
      \Pre $\KwMsgR[\id] \leq b$ \label{algo:new:consensus-pre} \;
      $\KwMsgM[\id] \leftarrow c$;
      $\KwMsgCcl[\id] \leftarrow D$ \;
      $\KwMsgR[\id] \leftarrow b$;
      $\KwMsgW[\id] \leftarrow b$ \;
      \Send $\KwMConsA(\id, b)$ \To $j$ \;
    }

    \SubAlgo{\RecA $\KwMConsA(\id, b)$ \RecB $Q$} {
      \Pre $\KwMsgR[\id] = b \land |Q| = f + 1$ \label{algo:new:consensus-ack-pre} \;
      \Send $\KwMCom(\id, \KwMsgM[\id], \KwMsgCcl[\id])$ \ToAll \label{algo:new:consensus-end}\;
    }  

    \SubAlgo{\RecA $\KwMCom(\id, c, D)$} {
      \Pre $\id \not \in \KwPhDset \union \KwPhEset$\;
      $\KwMsgM[\id] \leftarrow c$;
      $\KwMsgCcl[\id] \leftarrow D$;
      $\KwMsgPh[\id] \leftarrow \KwPhD$ \;
    }

    \newcounter{CoreLastLine}
    \setcounter{CoreLastLine}{\value{AlgoLine}}

    \ifspacehacks
    \vspace{-.6em}
    \fi
    \algrule[0.8pt]
    \begin{center}
        \begin{tabular}{r@{}c@{}l@{}ll}
          $\KwMsgM[\id]$ & ${} \assign {}$ & $\KwNoop$ & ${} \in \CMD$ & Command
          \\
          $\KwMsgPh[\id]$ & ${} \assign {}$ & $\KwPhA$ & & Phase
          \\
          $\KwMsgCcl[\id]$ & ${} \assign {}$ & $\emptySet$ & ${} \subseteq \ID$ & Dependency set
          \\
          $\KwMsgQ[\id]$ & ${} \assign {}$ & $\emptySet$ & ${} \subseteq \Proc$ & Fast quorum
          \\
          $\KwMsgR[\id]$ & ${} \assign {}$ & $0$ & ${} \in \mathds{N}$ & Current ballot
          \\
          $\KwMsgW[\id]$ & ${} \assign {}$ & $0$ & ${} \in \mathds{N}$ & Last accepted ballot
          \\
        \end{tabular}
    \end{center}

    \caption{\SYS protocol at process $i$: failure-free case.}
    \labalg{atlas}
  \end{algorithm}

%% file: fast-path-examples.tex
\subcaptionbox{
  \begin{tabular}{l}
    \rmark \ \SYS $f=2$\\
    \wmark \hspace{0.058cm} \emph{matching replies}
  \end{tabular}
  \label{fig:fp-ex1}}[0.24\textwidth]{
    \centering
    \scalebox{.75}{
    \begin{tikzpicture}
      \node[coordi] (1) at (-1.3, -1) {1};
      \node (1-text) at (-1.95, -1) {\tikznodelabel{a}};
      \node[quorum] (2) at (0, 0) {2};
      \node (2-text) at (0, 0.6) {\tikznodelabel{a,b,c}};
      \node[quorum] (3) at (1.3, -1) {3};
      \node (3-text) at (1.3, -0.4) {\tikznodelabel{a,b,d}};
      \node[restof] (4) at (.75, -2.3) {5};
      \node (4-text) at (.75, -2.9) {};
      \node[quorum] (5) at (-.75, -2.3) {4};
      \node (5-text) at (-.75, -2.9) {\tikznodelabel{a,c,d}};
      \draw[->,thick] (1.north) -- (2);
      \draw[->,thick] (2) .. controls +(down:5mm) .. (1.65);
      \draw[->,thick] (1.15) -- (3);
      \draw[->,thick] (3) .. controls +(down:5mm) .. (1.0);
      \draw[->,thick] (1.285) -- (5.north);
      \draw[->,thick] (5.140) .. controls +(left:1mm) .. (1.270);
    \end{tikzpicture}
  }
}
\subcaptionbox{
  \begin{tabular}{l}
    \wmark \ \SYS $f=2$\\
    \wmark \hspace{0.058cm} \emph{matching replies}
  \end{tabular}
  \label{fig:fp-ex2}}[0.24\textwidth]{
    \centering
    \scalebox{.75}{
    \begin{tikzpicture}
      \node[coordi] (1) at (-1.3, -1) {1};
      \node (1-text) at (-1.95, -1) {$\varnothing$};
      \node[quorum] (2) at (0, 0) {2};
      \node (2-text) at (0, 0.6) {$\varnothing$};
      \node[quorum] (3) at (1.3, -1) {3};
      \node (3-text) at (1.3, -0.4) {$\varnothing$};
      \node[restof] (4) at (.75, -2.3) {5};
      \node (4-text) at (.75, -2.9) {};
      \node[quorum] (5) at (-.75, -2.3) {4};
      \node (5-text) at (-.75, -2.9) {\tikznodelabel{b}};
      \draw[->,thick] (1.north) -- (2);
      \draw[->,thick] (2) .. controls +(down:5mm) .. (1.65);
      \draw[->,thick] (1.15) -- (3);
      \draw[->,thick] (3) .. controls +(down:5mm) .. (1.0);
      \draw[->,thick] (1.285) -- (5.north);
      \draw[->,thick] (5.140) .. controls +(left:1mm) .. (1.270);
    \end{tikzpicture}
  }
}
\subcaptionbox{
  \begin{tabular}{l}
    \rmark \ \SYS $f=1$\\
    \wmark \hspace{0.058cm} \emph{matching replies}
  \end{tabular}
  \label{fig:fp-ex3}}[0.24\textwidth]{
    \centering
    \scalebox{.75}{
      \begin{tikzpicture}
        \node[coordi] (1) at (-1.3, -1) {1};
        \node (1-text) at (-1.95, -1) {\tikznodelabel{a}};
        \node[quorum] (2) at (0, 0) {2};
        \node (2-text) at (0, 0.6) {\tikznodelabel{a,b}};
        \node[restof] (3) at (1.3, -1) {3};
        \node[restof] (4) at (.75, -2.3) {5};
        \node[quorum] (5) at (-.75, -2.3) {4};
        \node (5-text) at (-.75, -2.9) {\tikznodelabel{a,c}};
        \draw[->,thick] (1.north) -- (2);
        \draw[->,thick] (2) .. controls +(down:5mm) .. (1.65);
        \draw[->,thick] (1.285) -- (5.north);
        \draw[->,thick] (5.140) .. controls +(left:1mm) .. (1.270);
    \end{tikzpicture}
  }
}
\subcaptionbox{
  \begin{tabular}{l}
    \rmark \ \SYS $f=1$\\
    \rmark \ \emph{matching replies}
  \end{tabular}
  \label{fig:fp-ex4}}[0.24\textwidth]{
  \centering
  \scalebox{.75}{
    \begin{tikzpicture}
      \node[coordi] (1) at (-1.3, -1) {1};
      \node (1-text) at (-1.95, -1) {\tikznodelabel{a}};
      \node[quorum] (2) at (0, 0) {2};
      \node (2-text) at (0, 0.6) {\tikznodelabel{a}};
      \node[restof] (3) at (1.3, -1) {3};
      \node[restof] (4) at (.75, -2.3) {5};
      \node[quorum] (5) at (-.75, -2.3) {4};
      \node (5-text) at (-.75, -2.9) {\tikznodelabel{a}};
      \draw[->,thick] (1.north) -- (2);
      \draw[->,thick] (2) .. controls +(down:5mm) .. (1.65);
      \draw[->,thick] (1.285) -- (5.north);
      \draw[->,thick] (5.140) .. controls +(left:1mm) .. (1.270);
    \end{tikzpicture}
  }
}

%% file: atlas_recovery_spec.tex
\begin{algorithm}[!t]
  \addtolength{\linewidth}{1em}
    \algosetup
    \setcounter{AlgoLine}{\value{CoreLastLine}}
    \SubAlgo{\Fun $\KwRec(\id)$\label{algo:new:rec}}{
      $b \leftarrow i + n (\floor{\frac{\KwMsgR[\id]}{n}} + 1)$ \label{algo:new:my-ballot} \;
      \Send $\KwMR (\id, \KwMsgM[\id], b)$ \ToAll \;
    }

    \SubAlgo{\RecA $\KwMR(\id, \_, \_)$ \RecB $j$}{
      \Pre $\id \in \KwPhDset \union \KwPhEset$\label{algo:new:rec-committed} \;
      \Send $\KwMCom(\id, \KwMsgM[\id], \KwMsgCcl[\id])$ \To $j$\label{algo:new:rec-committed-send}
    }

    \SubAlgo{\RecA $\KwMR(\id, c, b)$ \RecB $j$}{
      \Pre $\KwMsgR[\id] < b \land \id \not \in \KwPhDset \union \KwPhEset$ \label{algo:new:rec-pre} \;
      \If{$\KwMsgR[\id] = 0 \land \id \in \KwPhAset$\label{algo:new:rec-first-see}}{
        $\KwMsgCcl[\id] \leftarrow \conflicts(c)$; \label{algo:new:conflicts-rec}
        $\KwMsgM[\id] \leftarrow c$
      }
      $\KwMsgR[\id] \leftarrow b$ \label{algo:new:recovery-ballot} \;
      $\KwMsgPh[\id] \leftarrow \KwPhF$ \label{algo:new:recovery-phase}\;
      \Send $\KwMRA(\id, \KwMsgM[\id], \KwMsgCcl[\id], \KwMsgQ[\id],$
      \nonl\hspace*{1.9cm}
      $\KwMsgW[\id], b)$ \To $j$ \label{algo:new:rec-ack}
  }

    \SubAlgo{\mbox{\RecA $\KwMRA(\id, \msg_j, \dep_j, Q^0_j, ab_j, b)$ \RecBAll
      $j\hspace{1pt}{\in}\hspace{1pt} Q$}\label{algo:new:mrec-ack-pre}}{ 
      \Pre $\KwMsgR[\id] = b \land |Q| = n - f$\;
      \If{$\exists k \in Q.\, ab_k \not = 0$\label{algo:new:previous-proposal}}{
        \Let $k$ be such that $ab_k$ is maximal \label{algo:new:highest-proposal}\;
        \Send $\KwMCons(\id, \msg_k, \dep_k, b)$ \ToAll
      }
    \ElseIf{$\exists k \in Q.\, Q^0_k \not= \emptyset$\label{algo:new:qk-check}}{
      $Q' \leftarrow$ \leIf{$id.1 \in Q$\label{algo:new:coordinator-check}}{$Q$}{$Q \cap Q^0_k$}
      \Send $\KwMCons(\id, \msg_k, \bigcup_{Q'} \dep, b)$ \ToAll \label{algo:new:rec-union}
    }
    \lElse{\Send $\KwMCons(\id, \KwNoop, \emptyset, b)$ \ToAll \label{algo:new:noop}}

    \newcounter{RecoveryLastLine}
    \setcounter{RecoveryLastLine}{\value{AlgoLine}}
  }
    \caption{\SYS protocol at process $i$: recovery.}
    \labalg{recovery}
  \end{algorithm}

%% file: atlas_execution_spec.tex
\begin{algorithm}[!t]
    \algosetup
    \setcounter{AlgoLine}{\value{RecoveryLastLine}}
    \SubAlgo{\Loop}{
      \Let $S$ be the smallest subset of $\KwPhDset$ such that\;
      {
        \nonl
        \hspace{0.65em}
          $\forall \id \in S.\, (\KwMsgCcl[\id] \subseteq S \cup \KwPhEset)$
          \label{algo:new:exec-start} \;
      }
      \For{$\id \in S$ \OrderedBy $<$\label{algo:new:exec-sort}}{
        $\pcdeliver(\KwMsgM[\id])$ \;
        $\KwMsgPh[\id] \leftarrow \KwPhE$
        \label{algo:new:exec-end}
      }
    }
    \caption{\SYS protocol: command execution.}
    \labalg{execution}
  \end{algorithm}

%% file: properties.tex
\subsection{\SYS Properties and Comparison with EPaxos}
\labsection{properties}

\paragraph{Complexity.}
\SYS commits a command after two communication delays when taking the fast path, and four otherwise.
As pointed out in \refsection{algo:details}, when $f = 1$, a fast quorum contains exactly a majority of processes and \SYS \emph{always} takes the fast path.
This is optimal for leaderless protocols~\cite{cap,lower-bounds} and results in a significant
performance pay-off (\refsection{evaluation}).

\paragraph{Fault tolerance.}
\SYS is parameterized by the number of tolerated concurrent faults $f$: smaller
values of $f$ yield smaller fast and slow quorums, thus reducing latency. As
observed in the literature~\cite{spanner,xpaxos} and as we experimentally
confirm in \S\ref{sect:evaluation:synchrony}, assuming small values of $f$ is
acceptable for geo-distribution. Furthermore, violating our assumption that the
number of failures is bounded by $f$ may only compromise the liveness of the
protocol, and never its safety: if more than $f$ transient outages occur, due
to, e.g., connectivity problems, \SYS will just block until enough sites are
reachable.

%% file: epaxos.tex
\paragraph{Comparison with EPaxos.}
\SYS belongs to the family of leaderless SMR protocols.  We now provide a
concise comparison with the most prominent protocol in this family,
EPaxos~\cite{epaxos}. The two protocols share the message flow, including the
splitting into fast and slow paths. However, as we demonstrate experimentally
in \refsection{evaluation}, \SYS significantly outperforms EPaxos, which is due
to a number of novel design decisions that we took.

First, EPaxos requires the conflicts reported by the fast quorum processes to the coordinator
to match exactly, whereas \SYS allows processes to report different
dependencies, as long as each dependency can be recovered after $f$ failures.
This allows \SYS to take the fast path even when non-commuting commands are
submitted concurrently.

Second, \SYS allows choosing the number of failures $f$ independently of the
size of the system $n$, which yields fast quorums of size
$\left\lfloor \frac{n}{2} \right\rfloor + f$. EPaxos assumes up to
$\left\lfloor \frac{n}{2} \right\rfloor$ failures and sets the fast quorum size
to $\left\lfloor \frac{3n}{4} \right\rfloor$. Our decision results in smaller
quorums for small values of $f$, which are appropriate in planet-scale
systems~\cite{xpaxos,spanner,blotter}; smaller quorums then result in lower latency.
Note that EPaxos cannot be straightforwardly modified to exploit the independent
bound on failures $f$ due to its very complex recovery
mechanism~\cite{epaxos-thesis,pando} (which, in fact, has been recently shown to
contain a bug~\cite{epaxos-bug}). In contrast, \SYS achieves its
smaller fast quorums with a significantly simpler recovery protocol that is able
to recover fast-path decisions using Property~\ref{property:set-union}.

%% file: proof.tex
\subsection{Correctness}
\label{sect:correctness}

We have rigorously proved Invariants~\ref{inv:cons} and \ref{inv:conf} (see
\tra{\ref{app:On}}{A}; we omit the easy proofs of Invariants~\ref{inv:del}
and~\ref{inv:del2}). We now prove that the protocol invariants imply the
correctness of \SYS, i.e., that it satisfies the SMR specification. The only
nontrivial property is Ordering, which we prove next.

\begin{lemma}\label{thm:assym}
The relation $\bigcup_{i=1}^n {\mapsto_i}$ is asymmetric.
\end{lemma}
\paragraph{\it Proof.} By contradiction, assume that for some processes $i$ and $j$
and conflicting commands $c$ and $c'$ with identifiers $\id$ and $\id'$, we have
$c \mapsto_i c'$ and $c' \mapsto_j c$; then $c \not= \KwNoop$ and 
$c' \not= \KwNoop$. By Integrity we must have $i\not=j$ and
$c \not= c'$.

Assume first that $c$ and $c'$ are executed at process $i$ in the same batch
$S$. Then by Invariant~\ref{inv:del2} they also have to be executed at process
$j$ in the batch $S$. Since inside a batch commands are ordered using the fixed
order $<$ on their identifiers, $c$ and $c'$ have to be executed in the same
order at the two processes: a contradiction.

Assume now that $c$ and $c'$ are not executed at process $i$ in the same
batch. Then by Invariant~\ref{inv:del2} this also must be the case at process
$j$. Hence, Invariant~\ref{inv:del} implies that $\id' \not\in \KwMsgCl[\id]$ at
process $i$, and $\id \not\in \KwMsgCl[\id']$ at process $j$. Then process $i$
received $\KwMCom(\id, c, D)$ with $\id' \not\in D$, and process $j$ received
$\KwMCom(\id', c', D')$ with $\id \not\in D'$, which contradicts
Invariant~\ref{inv:conf-a}.\qed

\begin{lemma}\label{thm:path}
  Assume $c_1 \mapsto \ldots \mapsto c_n$ for $n \ge 2$. Whenever a process $i$
  executes $c_n$, some process has already executed $c_1$.
\end{lemma}
\paragraph{\it Proof.}
We prove the lemma by induction on $n$. The base case of $n=2$ directly follows
from the definition of $\mapsto$. Take $n > 3$ and assume
$c_1 \mapsto \ldots \mapsto c_{n-1} \mapsto c_n$. Consider the moment when a
process $i$ executes $c_n$. We want to show that by this moment some process has
already executed $c_1$. Since $c_{n-1} \mapsto c_n$, either
$c_{n-1} \leadsto c_n$ or $c_{n-1} \mapsto_j c_n$ for some process $j$. Consider
first the case when $c_{n-1} \leadsto c_n$. Then $c_{n-1}$ is executed at some
process $k$ before $c_n$ is submitted and, hence, before $c_n$ is executed at
process $i$. By induction hypothesis, $c_1$ is executed at some process before
$c_{n-1}$ is executed at process $k$ and, hence, before $c_n$ is executed at
process $i$, as required. We now consider the case when $c_{n-1} \mapsto_j c_n$
for some process $j$. Since process $i$ executes $c_n$, we must have either
$c_{n-1} \mapsto_i c_n$ or $c_n \mapsto_i c_{n-1}$. The latter case would
contradict Lemma~\ref{thm:assym}, so that $c_{n-1} \mapsto_i c_n$. By induction
hypothesis, $c_1$ is executed at some process before $c_{n-1}$ is executed at
process $i$ and, hence, before $c_n$ is executed at process $i$, as
required.\qed

\paragraph{\it Proof of Ordering.}
By contradiction, assume that $c_1 \mapsto \ldots \mapsto c_n = c_1$ for
$n \ge 2$. Then some process executed $c_1$. Consider the moment when the first
process did so. By Lemma~\ref{thm:path} some process has already executed $c_1$
before this, which yields a contradiction.\qed

%% file: new_opt.tex
\section{Optimizations}
\labsection{execution:improvements}

This section presents two mechanisms employed by the \SYS protocol to accelerate command execution.

\paragraph{Reducing dependencies in the slow path.}
When computing dependencies in the slow path, instead of proposing $\bigcup_{Q} \dep$ to consensus (line~\ref{algo:new:sp}),
the coordinator can propose $\quorumcup{f}{Q} \dep$.
This allows \SYS to prune from dependencies those commands that been reported by less than $f$ fast-quorum processes ($\{ \id \mid \multiplicity(\id) < f \}$)
without breaking Invariant~\ref{inv:conf-a}.
Smaller dependency sets allow batches to form more quickly in execution (\refalg{execution}),
thus reducing the delay between a command being committed and executed.

Back in \reffigure{overview}, command \msgB was reported to coordinator $1$
solely by process $4$. Since \msgB was reported by less than $f=2$ processes,
the above optimization allows coordinator $1$ to prune \msgB from the dependencies
of \msgA, thus proposing $\KwMsgCcl[\msgA] = \varnothing$ to consensus. This
maintains Invariant~\ref{inv:conf-a}, as \msgA is still a dependency of \msgB:
$\KwMsgCcl[\msgB] = \{\msgA\}$.
In \tra{\ref{app:slow-path-opt}}{A.2.1} we prove that this optimization always
maintains Invariant~\ref{inv:conf-a}.

\paragraph{Non-fault-tolerant reads.}
We observe that reads can be excluded from dependencies at lines~\ref{algo:new:past} and \ref{algo:new:conflicts} when the conflict relation between commands is transitive.
In this case, a read is never a dependency and thus it will never block a later command,
even if it is not fully executed, e.g., when its coordinator fails (or hangs).
For this reason, reads can be executed in a non-fault-tolerant manner.
More precisely, for some read with identifier $\id$, the coordinator selects a plain majority as a fast quorum (line~\ref{algo:new:fast-quorum}), independently of the value of $f$.
Then, at the end of the $\KwPhB$ phase, it immediately commits $\id$, setting $\KwMsgCcl[\id]$ to the union of all dependencies returned by this quorum
(line~\ref{algo:new:fp}).
This optimization, that we denote by \readopt, accelerates the execution of
linearizable reads and reduces their impact in the protocol stack.
We show its correctness in \tra{\ref{appendix:smr}}{B}.
The transitivity requirement on conflicts is satisfied by many common applications.
We experimentally evaluate the case of a key-value store in \refsection{evaluation}.

%% file: evaluation.tex
\section{Performance Evaluation}
\labsection{evaluation}

In this section we experimentally compare \SYS with Flexible Paxos (FPaxos) \cite{flexible-paxos}
and two leaderless protocols, EPaxos~\cite{epaxos} and Mencius~\cite{mencius}.
Mencius distributes the leader responsibilities round-robin among replicas; because of this, executing a command in Mencius requires contacting all replicas.
As discussed previously, FPaxos uses a quorum of $f+1$ replicas in the failure-free case
in exchange for a bigger quorum of $n-f$ replicas on recovery.

To improve the fairness of our comparison, \SYS and EPaxos use the same codebase.
This codebase consists of a server component, written in Erlang (3.7K~SLOC), and a client component, written in Java (3.1K~SLOC).
The former commits commands, while the latter executes them.
Thus, the implementation of two protocols differs only in the logic of the commit component.
For Mencius and Paxos we use the Golang implementation provided by the authors of EPaxos \cite{epaxos},
which we extended to support FPaxos.

Our evaluation takes place on Google Cloud Platform (GCP), in a federation of Kubernetes clusters \cite{k8s}.
The federation spans from 3 to 13 geographical regions spread across the world, which we call \emph{sites}.
When protocols are deployed in all 13 sites, we have 4 sites in Asia, 1 in Australia, 4 in Europe, 3 in North America, and 1 in South America. 
A site consists of a set of virtualized Linux nodes, %
each an 8-core Intel Xeon machine with 30~GB of memory (\texttt{n1-standard-8}).
At a site, the SMR protocol and its clients execute on distinct machines.
When benchmarking FPaxos, we take as leader the site that minimizes the standard deviation of clients-perceived latency.
This site corresponds to the fairest location in the system, trying to satisfy uniformly all the clients.

\subsection{Bounds on Failures}
\labsection{evaluation:synchrony}

In a practical deployment of \SYS, a critical parameter is the number of concurrent site failures $f$ the protocol can tolerate.  
It has been reported that concurrent site failures are rare in geo-distributed systems~\cite{spanner}. %
However, the value of $f$ should also account for asynchrony periods during which sites cannot communicate due to link failures: if more than $f$ sites are unreachable in this way, our protocol may block for the duration of the outage.
We have thus conducted an experiment to check that assuming small values of $f$ is still appropriate when this is taken into account.

\begin{figure}[!t]
  \centering
  \includegraphics[scale=0.6,height=4.1cm,keepaspectratio]{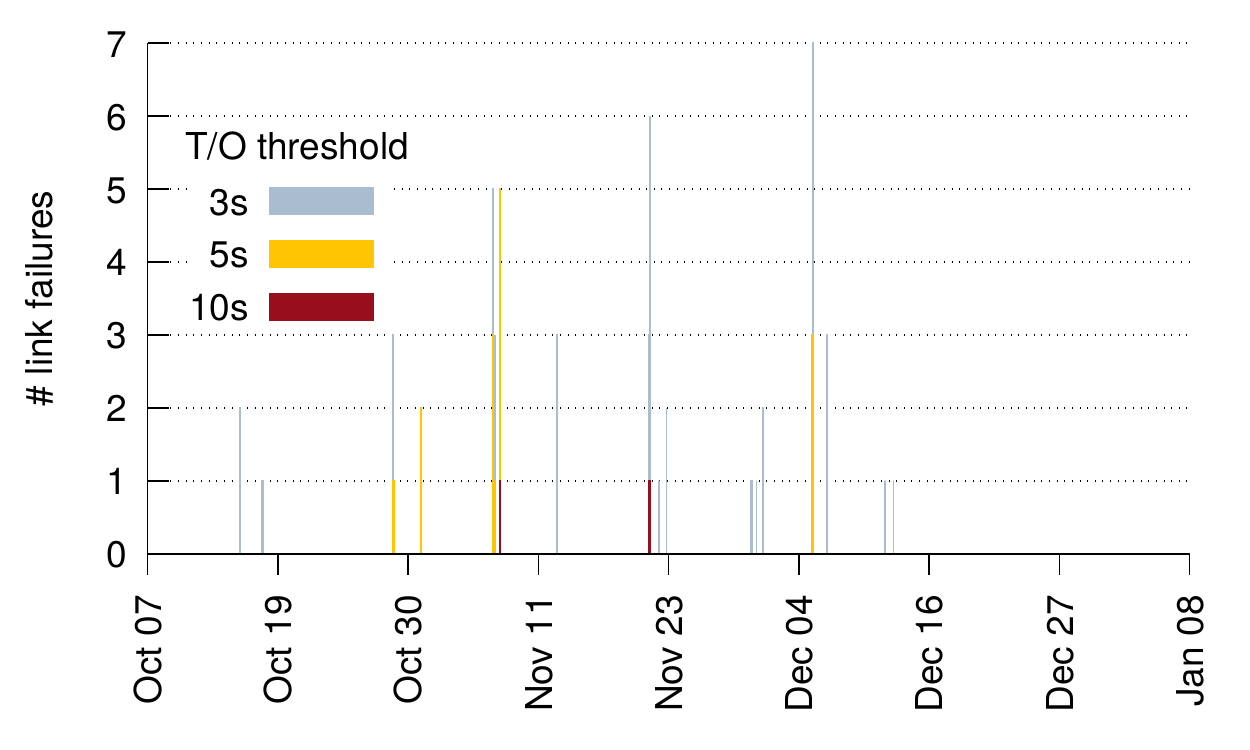}
  \caption{%
    \labfigure{ping}
    The number of simultaneous link failures among 17 sites in GCP when varying the timeout threshold.
  }
\end{figure}

Our experiment ran for 3 months (October 2018 -- January 2019) among 17 sites, the maximal number of sites available in GCP at the time.
During the experiment, sites ping each other every second (in the spirit of \cite{xpaxos} but on a much larger scale).
A link failure occurs between two sites when one of them does not receive a reply after a (tunable) amount of time.
\reffigure{ping} reports the number of simultaneous link failures for various timeout thresholds.
Note that no actual machine crash occurred during the campaign of measurements.

When the timeout threshold is set to 10s, only two events occur, each with a single link failure.
Fixing the threshold to either 3s or 5s leads to two events of noticeable length.
During the first event, occurring on November 7, the links between the Canadian site (QC) and five others are slow for a couple of hours.
During the second event, on December 8, the links between Taiwan (TW) and seven other sites are slow for around two minutes.

From the data collected, we compute the value of $f$ as the smallest number
of sites $k$ such that, at any point in the experiment, crashing $k$ sites would
cover all the slow links.
During our experiment, timeouts were reported on the links incident to at most a single site (e.g., the Canadian site on November 7).
Thus, we may conclude that $f \leq 1$ held during the whole experiment, even with the smallest timeout threshold.
In other words, \SYS with $f \geq 1$ would have been always responsive during this 3-month experiment.
In light of these results, we evaluate deployments of \SYS in which $f$ is set to $1$, $2$ or $3$.

\subsection{Benchmarks}
\labsection{evaluation:micro}

Our first set of experiments uses a microbenchmark -- a stub application that
executes dummy commands
(\refsection{evaluation:micro}-\refsection{evaluation:faults}). We then evaluate
\SYS with a geo-replicated key-value store under the YCSB workload~\cite{ycsb} (\refsection{evaluation:ycsb}).
In our microbenchmark a varying number of closed-loop clients access the service
at the closest (non-failed) site.
Clients measure latency as the time between submitting a command and the system executing it.
Each command carries a key of 8 bytes and
(unless specified otherwise) a payload of 100 bytes.  We assume that commands
conflict when they carry the same key. To measure performance under a rate
$\rho$ of conflicting commands, a client chooses key $0$ with a probability
$\rho$, and some unique key otherwise.

\subsection{Fast-Path Likelihood}
\labsection{evaluation:fast}

\reffigure{fastpath} evaluates the benefits of our new fast-path condition.
To this end, it compares the fast-path ratio of \SYS and EPaxos for different conflict rates and values of $f$.
The system consists of 3 sites when $f=1$, 5 sites when $f=2$, and 7 sites when $f=3$.
There is a single client per site (the results with more clients are almost identical).

As noted before, \SYS always commits a command on the fast path when $f=1$.
For higher values of $f$, our condition for taking the fast-path significantly improves its likelihood in comparison to EPaxos.
With 5 sites and $f=2$, when the conflict rate increases by 20\%, the ratio of fast paths in EPaxos decreases on average by 20\%.
In contrast, the fast-path ratio in \SYS only decreases by 10\%.
When all commands conflict, EPaxos rarely takes the fast path, while \SYS does so for 50\% of  commands.
Similar conclusions can be drawn from \reffigure{fastpath} when the two protocols are deployed with $f=3$.

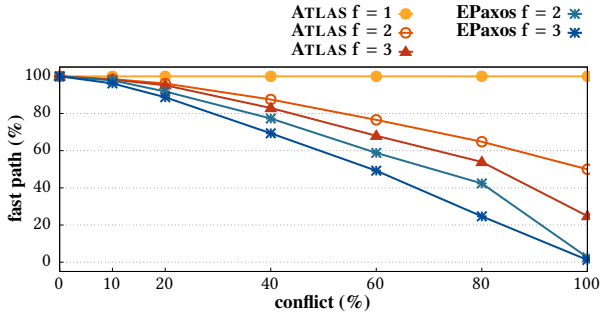
\begin{figure}[t]
  \centering
  \fontsize{12}{14}\selectfont
  \scalebox{.6}{\input{fast_path_plot_source}}
  \caption{
    \labfigure{fastpath}
    Ratio of fast paths for varying conflict rates.}
\end{figure}

\subsection{Planet-Scale Performance}
\labsection{evaluation:scalability}

We now consider two planet-scale scenarios that motivate the design of \SYS.
In these experiments we measure how the performance of \SYS evolves as the system scales up progressively from 3 to 13 sites.
In the first experiment, the load on \SYS is constant, using a fixed number of clients spread across all 13 sites.  
We demonstrate how bringing the service closer to already existing clients, by adding new replicas, improves the latency these clients perceive.
In the second experiment, each \SYS site hosts a fixed number of clients, so that the growth in the number of sites translates into increased load. This models a scenario where the service expands to new locations around the globe in order to serve new clients in these locations.
In this case we demonstrate that \SYS gracefully copes with this growth, maintaining its performance in contrast to state-of-the-art SMR protocols.

\begin{figure}[t]
  \centering
  \fontsize{12}{14}\selectfont
  \scalebox{.6}{\input{new_latency_plot_source}}
  \caption{
    \labfigure{latency}
    Latency when scaling-out from 3 to 13 sites
    with 1000 clients spread across 13 sites
    and $2\%$ conflict rate. Percentages indicate the overhead wrt the optimal performance.
  }
\end{figure}
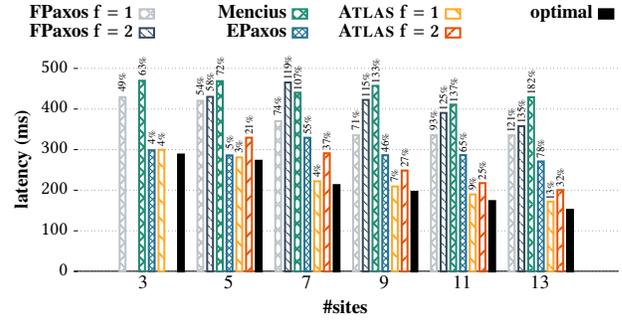

\paragraph{Bringing the service closer to clients.}
We deploy 1000 clients equally spread over 13 sites, which issue commands at a fixed 2\% conflict rate.
\reffigure{latency} reports how the average latency changes as \SYS is gradually deployed closer to client locations.
  The black bar in \reffigure{latency} gives the average of the sum of round-trip
  latencies from a client to the closest coordinator, and from the coordinator
  to its closest majority.
  As clients execute no protocol logic and may not be co-located with sites, 
  this gives the optimal latency for leaderless protocols (\refsection{properties}).
The percentages on bars indicate the overhead of the different protocols with respect to this theoretical value.

As shown in \reffigure{latency}, \SYS improves its performance when deployed closer to clients:
these can access a closer coordinator site, which in its turn accesses the closest fast quorum of sites.
In \reffigure{latency}, the latency of \SYS $f=1$ improves on average by 25ms whenever two new sites are added; for $f=2$ this improvement is 33ms.
For 13 sites, the optimal latency is 151ms, and \SYS $f=1$ is only 13\% above this value, with an average latency of 172ms; \SYS $f=2$ is 32\% above the optimum, with an average latency of 200ms.
Overall, going from 3 to 13 sites with \SYS ($f=1,2$) cuts down client-perceived latency by 39\%-42\%.

As seen in \reffigure{latency}, the performance of \SYS greatly contrasts with that of the three other SMR protocols.
With 13 sites, FPaxos executes commands after 336ms when $f=1$ and after 358ms when $f=2$, which is almost twice as slow as \SYS with identical failure assumptions.
This large gap comes from the fact that, for a command to execute in a
leader-based protocol, clients wait for four message delays on the critical path: a round trip from the client to the leader, and a round trip from the leader to a quorum.

The performance of EPaxos remains almost constant, within 10\% of 300ms.
With 13 sites, EPaxos is 78\% slower than the optimum, and 57\% slower than \SYS $f = 1$.
This penalty is due to the large fast quorums it employs.

Finally, Mencius exhibits a high latency -- above 400ms -- in every configuration.
This is because a replica needs to contact all the other replicas to execute a command,
and thus, the performance of Mencius is bounded by the speed of its slowest replica.

\begin{figure}[t]
  \centering
  \fontsize{12}{14}\selectfont
  \scalebox{.6}{\input{new_relative_latency_plot_source}}
  \caption{
    \labfigure{tput-relative}
    Latency penalty (with respect to the optimal performance) when scaling-out from 3 to 13 sites,
    with 128 clients deployed on each site,
    and $1\%$ of conflict rate.
  }
\end{figure}
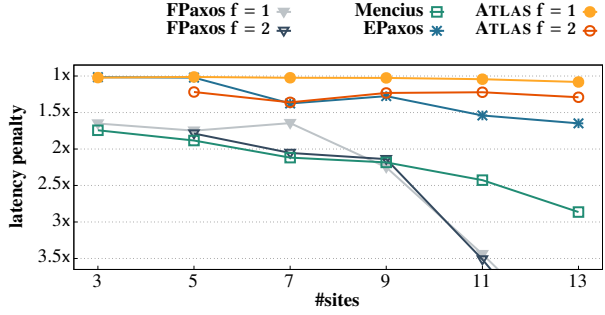

\paragraph{Expanding the service.}
We now consider another planet-scale experiment that models a situation in which the service expands to new locations to serve new clients. The experiment considers 3 to 13 sites, with 128 clients per site, and each clients submits commands with a payload of 3KB. 
\reffigure{tput-relative} reports the latency penalty with respect to the optimal.

FPaxos $f = 1$ exhibits a latency penalty ranging from 
1.7\textsf{x} to 4.7\textsf{x} (the last value is not shown in \reffigure{tput-relative} for readability).
In particular, starting from 9 sites its performance degrades sharply with the
increase in the number of sites and, hence, the number of clients.
This happens because the leader cannot cope with the load, having to broadcast each command to all replicas.
FPaxos $f=2$ follows a similar trend.

EPaxos behaves better than FPaxos, hitting the optimal performance with 3 and 5 sites.
However, starting from 11 sites, the latency of EPaxos becomes at best 50\% of the optimum.
Overall, due to its large fast quorums, the performance of EPaxos lowers as the number of sites increases.

In contrast to the prior two protocols, \SYS distributes the cost of broadcasting command payloads among replicas and uses small fast quorums.
This allows the protocol to be within 4\% of the optimum when $f=1$, and within 26\% when $f=2$.
\SYS is thus able to cope with the system growth without degrading performance.

\subsection{Varying Load and Conflict Rate}
\labsection{evaluation:varying}

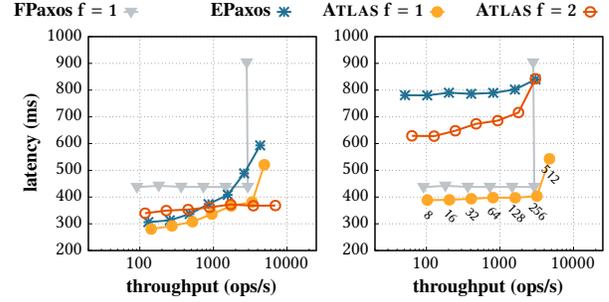
\begin{figure}[t]
  \centering
  \fontsize{12}{14}\selectfont
  \scalebox{.6}{\input{new_tput_latency_plot_source}}
  \caption{
    \labfigure{tput-latency}
    Throughput and latency with 5 sites when the load (number of clients) increases
    under moderate (left, 10\%) and high (right, 100\%) conflict rates.
  }
\end{figure}

To further understand the performance of \SYS, we conduct an experiment in which the load and the conflict rate vary.
The protocol is deployed at 5 sites, and the load increases from 8 to 512 clients per site, under a moderate (10\%) to high (100\%) conflict rate.
As before, messages carry a payload of 3KB.
The results are presented in \reffigure{tput-latency}, where we also compare with FPaxos $f=1$ and EPaxos.

Under a 10\% conflict rate (left-hand side of \reffigure{tput-latency}) and with up to 64 clients per site, \SYS $f=1$ executes commands with an average latency below 336ms.
When this load doubles and then quadruples, the latency increases to respectively 366ms and 381ms.
Compared to \SYS, the performance of EPaxos degrades faster with increasing load,
yielding latencies of 368ms, 404ms and 484ms for 64, 128 and 256 clients per site, respectively.
FPaxos performance is stable at 437ms for up to 256 clients per site, as the leader is capable of handling such a moderate load.

At a high load, with 512 clients per site, all the protocols but \SYS $f=2$
saturate.  In particular, FPaxos saturates because the leader is no longer
capable of coping with the load.  Although \SYS $f=1$ is the most efficient
protocol until saturation, its performance degrades sharply at 512 clients per
site due to large batches formed during command execution. Interestingly, \SYS $f=2$ behaves 
better due to the slow-path optimization in \refsection{execution:improvements}.
Since this protocol uses larger fast quorums, the optimization allows it to
prune dependencies that \SYS $f=1$ cannot: while coordinators in \SYS $f=1$ must
include every dependency reported by a fast-quorum process for a given command,
coordinators in \SYS $f=2$ only include the dependencies reported by at least
$2$ fast-quorum processes. This reduces the size of batches in execution,
improving the overall protocol latency.

With a 100\% conflict rate (right-hand side of \reffigure{tput-latency}),
EPaxos performs worse than the remaining protocols.
It executes commands with an average latency of at least 780ms, making the protocol unpractical in this context.
As pointed out in \refsection{evaluation:fast}, this is explained by its fast-path condition which rarely triggers when the conflict rate is high.
In contrast, \SYS $f=1$ is consistently the fastest protocol.
\SYS is slower than FPaxos $f=1$
only when providing a higher fault-tolerance level ($f=2$).

\subsection{Availability under Failures}
\labsection{evaluation:faults}

\begin{figure}[t]
  \centering
  \fontsize{12}{14}\selectfont
  \scalebox{.6}{\input{recovery_multi_plot_source}}
  \caption{
    \labfigure{recovery-multi}
    The impact of a failure on the throughput of Paxos and \SYS (3 sites, $f=1$).
  }
\end{figure}
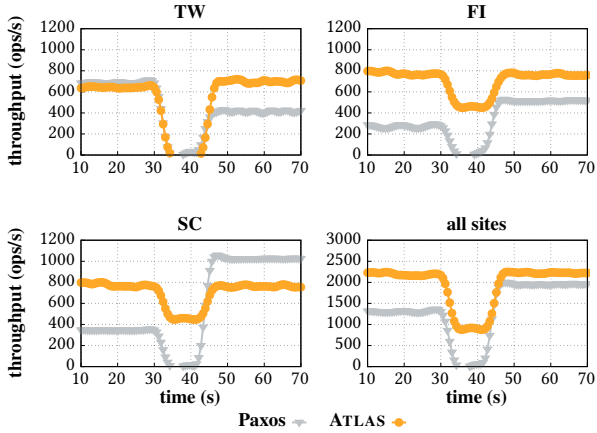

\reffigure{recovery-multi} depicts an experiment demonstrating that
\SYS is inherently more available than a leader-driven protocol.
The experiment runs across 3 sites: Taiwan (TW), Finland (FI) and South Carolina (SC). Such configuration tolerates a single site failure, so FPaxos is the same as Paxos.
We do not evaluate EPaxos, as its availability guarantees are similar to those of \SYS in this configuration.
Each site hosts 128 closed-loop clients.
Half of the clients issue commands targeting key $0$ and the other half
issue commands targeting a unique key per client. Hence, commands by clients in
the first half conflict with each other, while commands by clients in the second
half commute with all commands by a different client.

After 30s of execution, the SMR service is abruptly halted at the TW site,
where the Paxos leader is located.
Based on the measurements reported in \refsection{evaluation:synchrony},
we set the timeout after which a failure is suspected to 10s for both protocols.
Upon detecting the failure, the clients located at the failed site (TW) reconnect to the closest alive site, SC. In the case of Paxos, the surviving sites initiate recovery and elect SC as the new leader. 
In the case of Atlas, the surviving sites recover the commands that were initially coordinated by TW.

As shown in \reffigure{recovery-multi}, Paxos blocks during the recovery time.
In contrast, \SYS keeps executing commands, albeit at a reduced throughput.
The drop in throughput happens largely because the clients issuing commands on key $0$
(50\% of all clients) collect as dependencies some of the commands being
recovered (those that also access  key $0$).
The execution of the former commands then blocks until the latter are recovered.
In contrast, the clients at non-failed sites issuing commands with per-client keys continue to operate as normal.
Since commands by these clients commute with those by other clients, their execution never blocks on the commands being recovered.
This means that these clients operate without disruption during the whole experiment.

The bottom right plot contains the aggregate throughput of the system.
Before failure, \SYS is almost two times faster than Paxos, and operates consistently better during the whole experiment.
Note, however, that Paxos has a slightly higher throughput at the leader (TW) before the crash,
and at the new leader (SC) after recovery. This is due to the delay
between committing and executing commands in \SYS.

\subsection{Key-Value Store Service}
\labsection{evaluation:ycsb}

\begin{figure}[t]
  \centering
  \fontsize{12}{14}\selectfont
  \scalebox{.6}{\input{new_ycsb_reads_commute_plot_source}}
  \caption{
    \labfigure{ycsb}
    YCSB performance for update-heavy (20\%-80\%), balanced (50\%-50\%), read-heavy (80\%-20\%) and read-only (100\%-0\%) workloads, with 7 (top) and 13 sites (bottom).
    A * before the protocol name indicates that the \readopt optimization (\refsection{execution:improvements}) is enabled.
    The number at the top of each bar indicates the speed-up over (vanilla) EPaxos.
  }
\end{figure}
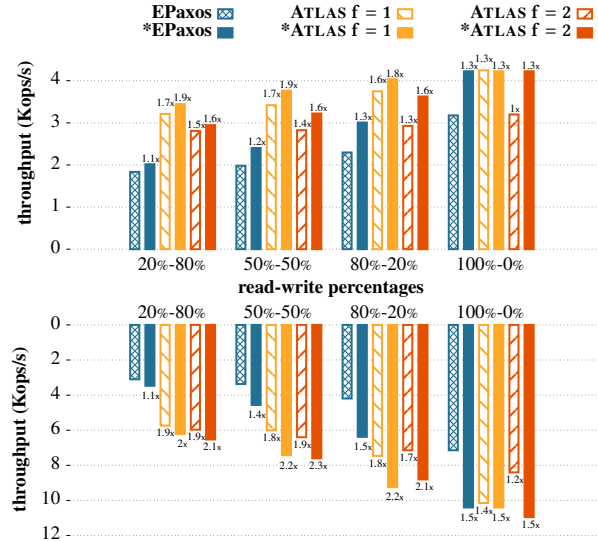

Our final experiment compares \SYS and EPaxos when the protocols are applied to a replicated key-value store (KVS) service.
When accessing a KVS record stored under key $k$, a client executes either
\begin{inparaenumorig}[]
\item command $\get(k)$ to fetch its content, or
\item $\aput(k,v)$ to update it to value $v$.
\end{inparaenumorig}
To benchmark the performance of the replicated KVS we use the Yahoo! Cloud Serving Benchmark (YCSB)~\cite{ycsb}.
We apply four types of workloads, each with a different mix of $\get$/$\aput$ operations.
The KVS contains $10^6$ records and all workloads select records following a Zipfian distribution with the default YCSB skew.

In this experiment, \SYS ($f=1,2$) and EPaxos are deployed over 7 and 13 sites (respectively, top and bottom of \reffigure{ycsb}).
At each site running the benchmark we execute 128 YCSB client threads.
The protocol name is preceded with the * symbol if the \readopt optimization is enabled.
As pointed out in \refsection{execution:improvements}, this optimization accelerates the execution of $\get$ commands.
The number at the top of each bar indicates the speed-up over (vanilla) EPaxos.

  With 7 sites, EPaxos executes 1.8K ops/s in the update-heavy workload, whereas \SYS executes 3.2K ops/s when $f=1$, and 2.8K ops/s when $f=2$.
  Although EPaxos and \SYS $f=2$ have the same fast-quorum size with $n=7$, the performance gap between the protocols is large for two reasons.
  First, the key-access distribution in YCSB does not allow EPaxos to take the fast path frequently, since the first 12 records have a 20\% chance of getting picked.
  Due to this, \SYS $f=2$ takes the fast path for 88\% of commands, while EPaxos does so in 70\% of cases.
  This makes the average commit latency of \SYS $f=2$ lower by 50ms.
   Second, batches formed in execution with EPaxos are larger than with \SYS $f=2$ because \SYS prunes unnecessary dependencies (\refsection{execution:improvements}):
  once commands are committed, EPaxos takes on average 147ms to execute them, while \SYS $f=2$ needs only 30ms.
  With $f=1$, \SYS has a longer execution delay of 73ms (this difference between $f=1$ and $f=2$ is explained in \refsection{evaluation:varying}).
  Nevertheless, \SYS $f=1$ beats \SYS $f=2$, since it always commits commands after contacting the closest majority, and this compensates for its higher execution delay.

  Increasing the percentage of read operations improves the performance of all the protocols because reads do not conflict with other reads.
  In the read-only workload the performance is simply determined by the quorum size, since all the protocols take the fast path.
  In this case, both EPaxos and \SYS $f=2$ execute 3.2K ops/s, while \SYS $f=1$, which has a smaller fast quorum, executes 4.2K ops/s.

  With the \readopt optimization and $n=7$, the protocols execute up to 33\% more operations.
  The highest speedup occurs in the read-only workload, where the protocols execute all commands after a single round trip to the closest majority. In this case, \readopt allows EPaxos and \SYS $f=2$ to match the performance of vanilla \SYS $f=1$ while maintaining their higher fault-tolerance level.
  Compared to vanilla EPaxos, \SYS with \readopt is up to 1.9\textsf{x} faster with $f=1$, and 1.6\textsf{x} with $f=2$.
  Similar conclusions can be drawn from \reffigure{ycsb} when the protocols are deployed over 13 sites.
  Overall, \SYS with \readopt outperforms EPaxos by 1.5-2.3\textsf{x}.

%% file: fast_path_plot_source.tex
\begingroup
  \fontfamily{ptm}%
  \selectfont
  \makeatletter
  \providecommand\color[2][]{%
    \GenericError{(gnuplot) \space\space\space\@spaces}{%
      Package color not loaded in conjunction with
      terminal option `colourtext'%
    }{See the gnuplot documentation for explanation.%
    }{Either use 'blacktext' in gnuplot or load the package
      color.sty in LaTeX.}%
    \renewcommand\color[2][]{}%
  }%
  \providecommand\includegraphics[2][]{%
    \GenericError{(gnuplot) \space\space\space\@spaces}{%
      Package graphicx or graphics not loaded%
    }{See the gnuplot documentation for explanation.%
    }{The gnuplot epslatex terminal needs graphicx.sty or graphics.sty.}%
    \renewcommand\includegraphics[2][]{}%
  }%
  \providecommand\rotatebox[2]{#2}%
  \@ifundefined{ifGPcolor}{%
    \newif\ifGPcolor
    \GPcolortrue
  }{}%
  \@ifundefined{ifGPblacktext}{%
    \newif\ifGPblacktext
    \GPblacktexttrue
  }{}%
  \let\gplgaddtomacro\g@addto@macro
  \gdef\gplbacktext{}%
  \gdef\gplfronttext{}%
  \makeatother
  \ifGPblacktext
    \def\colorrgb#1{}%
    \def\colorgray#1{}%
  \else
    \ifGPcolor
      \def\colorrgb#1{\color[rgb]{#1}}%
      \def\colorgray#1{\color[gray]{#1}}%
      \expandafter\def\csname LTw\endcsname{\color{white}}%
      \expandafter\def\csname LTb\endcsname{\color{black}}%
      \expandafter\def\csname LTa\endcsname{\color{black}}%
      \expandafter\def\csname LT0\endcsname{\color[rgb]{1,0,0}}%
      \expandafter\def\csname LT1\endcsname{\color[rgb]{0,1,0}}%
      \expandafter\def\csname LT2\endcsname{\color[rgb]{0,0,1}}%
      \expandafter\def\csname LT3\endcsname{\color[rgb]{1,0,1}}%
      \expandafter\def\csname LT4\endcsname{\color[rgb]{0,1,1}}%
      \expandafter\def\csname LT5\endcsname{\color[rgb]{1,1,0}}%
      \expandafter\def\csname LT6\endcsname{\color[rgb]{0,0,0}}%
      \expandafter\def\csname LT7\endcsname{\color[rgb]{1,0.3,0}}%
      \expandafter\def\csname LT8\endcsname{\color[rgb]{0.5,0.5,0.5}}%
    \else
      \def\colorrgb#1{\color{black}}%
      \def\colorgray#1{\color[gray]{#1}}%
      \expandafter\def\csname LTw\endcsname{\color{white}}%
      \expandafter\def\csname LTb\endcsname{\color{black}}%
      \expandafter\def\csname LTa\endcsname{\color{black}}%
      \expandafter\def\csname LT0\endcsname{\color{black}}%
      \expandafter\def\csname LT1\endcsname{\color{black}}%
      \expandafter\def\csname LT2\endcsname{\color{black}}%
      \expandafter\def\csname LT3\endcsname{\color{black}}%
      \expandafter\def\csname LT4\endcsname{\color{black}}%
      \expandafter\def\csname LT5\endcsname{\color{black}}%
      \expandafter\def\csname LT6\endcsname{\color{black}}%
      \expandafter\def\csname LT7\endcsname{\color{black}}%
      \expandafter\def\csname LT8\endcsname{\color{black}}%
    \fi
  \fi
    \setlength{\unitlength}{0.0500bp}%
    \ifx\gptboxheight\undefined%
      \newlength{\gptboxheight}%
      \newlength{\gptboxwidth}%
      \newsavebox{\gptboxtext}%
    \fi%
    \setlength{\fboxrule}{0.5pt}%
    \setlength{\fboxsep}{1pt}%
\begin{picture}(7200.00,3880.00)%
    \gplgaddtomacro\gplbacktext{%
      \csname LTb\endcsname%
      \put(436,593){\makebox(0,0)[r]{\strut{}$0$}}%
      \csname LTb\endcsname%
      \put(436,1060){\makebox(0,0)[r]{\strut{}$20$}}%
      \csname LTb\endcsname%
      \put(436,1527){\makebox(0,0)[r]{\strut{}$40$}}%
      \csname LTb\endcsname%
      \put(436,1995){\makebox(0,0)[r]{\strut{}$60$}}%
      \csname LTb\endcsname%
      \put(436,2462){\makebox(0,0)[r]{\strut{}$80$}}%
      \csname LTb\endcsname%
      \put(436,2929){\makebox(0,0)[r]{\strut{}$100$}}%
      \csname LTb\endcsname%
      \put(536,298){\makebox(0,0){\strut{}$0$}}%
      \csname LTb\endcsname%
      \put(1199,298){\makebox(0,0){\strut{}$10$}}%
      \csname LTb\endcsname%
      \put(1862,298){\makebox(0,0){\strut{}$20$}}%
      \csname LTb\endcsname%
      \put(3188,298){\makebox(0,0){\strut{}$40$}}%
      \csname LTb\endcsname%
      \put(4513,298){\makebox(0,0){\strut{}$60$}}%
      \csname LTb\endcsname%
      \put(5839,298){\makebox(0,0){\strut{}$80$}}%
      \csname LTb\endcsname%
      \put(7165,298){\makebox(0,0){\strut{}$100$}}%
    }%
    \gplgaddtomacro\gplfronttext{%
      \csname LTb\endcsname%
      \put(7,1761){\rotatebox{-270}{\makebox(0,0){\strut{}\textbf{fast path (\%)}}}}%
      \csname LTb\endcsname%
      \put(3850,60){\makebox(0,0){\strut{}\textbf{conflict (\%)}}}%
      \csname LTb\endcsname%
      \put(4726,3714){\makebox(0,0)[r]{\strut{}\textbf{\SYS $\mathbf{f=1}$}}}%
      \csname LTb\endcsname%
      \put(4726,3476){\makebox(0,0)[r]{\strut{}\textbf{\SYS $\mathbf{f=2}$}}}%
      \csname LTb\endcsname%
      \put(4726,3238){\makebox(0,0)[r]{\strut{}\textbf{\SYS $\mathbf{f=3}$}}}%
      \csname LTb\endcsname%
      \put(6850,3714){\makebox(0,0)[r]{\strut{}\textbf{EPaxos $\mathbf{f=2}$}}}%
      \csname LTb\endcsname%
      \put(6850,3476){\makebox(0,0)[r]{\strut{}\textbf{EPaxos $\mathbf{f=3}$}}}%
    }%
    \gplbacktext
    \put(0,0){\includegraphics{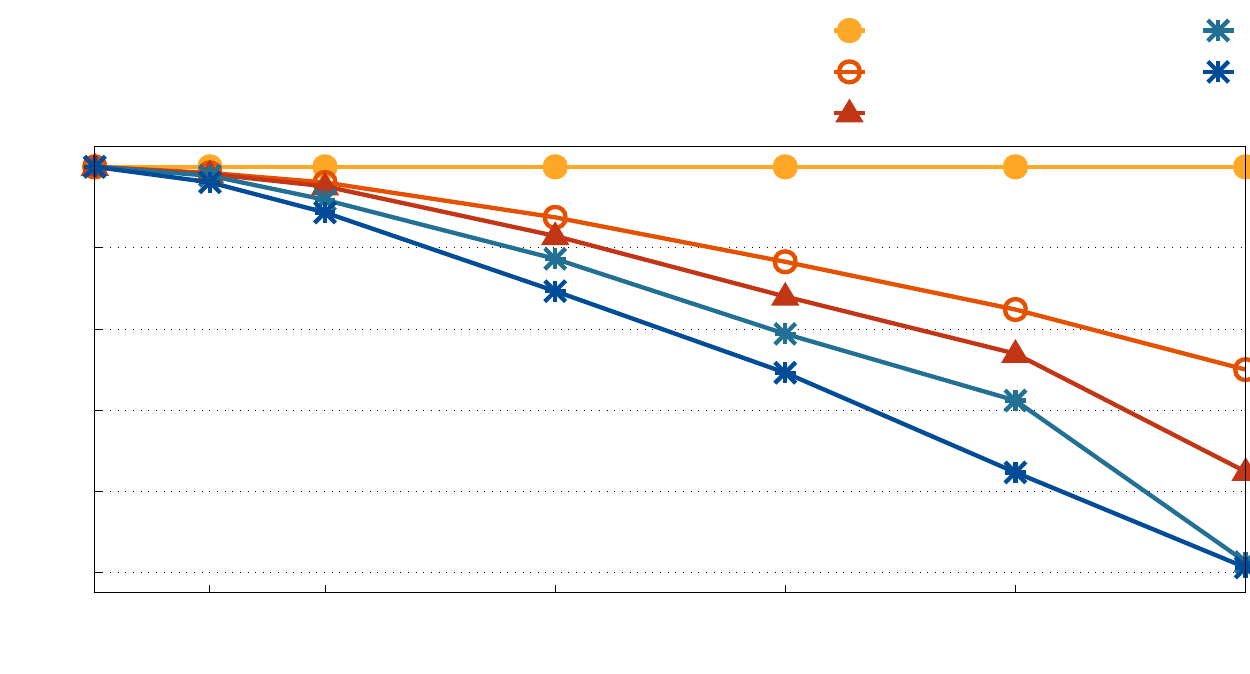}}%
    \gplfronttext
  \end{picture}%
\endgroup

%% file: new_latency_plot_source.tex
\begingroup
  \fontfamily{ptm}%
  \selectfont
  \makeatletter
  \providecommand\color[2][]{%
    \GenericError{(gnuplot) \space\space\space\@spaces}{%
      Package color not loaded in conjunction with
      terminal option `colourtext'%
    }{See the gnuplot documentation for explanation.%
    }{Either use 'blacktext' in gnuplot or load the package
      color.sty in LaTeX.}%
    \renewcommand\color[2][]{}%
  }%
  \providecommand\includegraphics[2][]{%
    \GenericError{(gnuplot) \space\space\space\@spaces}{%
      Package graphicx or graphics not loaded%
    }{See the gnuplot documentation for explanation.%
    }{The gnuplot epslatex terminal needs graphicx.sty or graphics.sty.}%
    \renewcommand\includegraphics[2][]{}%
  }%
  \providecommand\rotatebox[2]{#2}%
  \@ifundefined{ifGPcolor}{%
    \newif\ifGPcolor
    \GPcolortrue
  }{}%
  \@ifundefined{ifGPblacktext}{%
    \newif\ifGPblacktext
    \GPblacktexttrue
  }{}%
  \let\gplgaddtomacro\g@addto@macro
  \gdef\gplbacktext{}%
  \gdef\gplfronttext{}%
  \makeatother
  \ifGPblacktext
    \def\colorrgb#1{}%
    \def\colorgray#1{}%
  \else
    \ifGPcolor
      \def\colorrgb#1{\color[rgb]{#1}}%
      \def\colorgray#1{\color[gray]{#1}}%
      \expandafter\def\csname LTw\endcsname{\color{white}}%
      \expandafter\def\csname LTb\endcsname{\color{black}}%
      \expandafter\def\csname LTa\endcsname{\color{black}}%
      \expandafter\def\csname LT0\endcsname{\color[rgb]{1,0,0}}%
      \expandafter\def\csname LT1\endcsname{\color[rgb]{0,1,0}}%
      \expandafter\def\csname LT2\endcsname{\color[rgb]{0,0,1}}%
      \expandafter\def\csname LT3\endcsname{\color[rgb]{1,0,1}}%
      \expandafter\def\csname LT4\endcsname{\color[rgb]{0,1,1}}%
      \expandafter\def\csname LT5\endcsname{\color[rgb]{1,1,0}}%
      \expandafter\def\csname LT6\endcsname{\color[rgb]{0,0,0}}%
      \expandafter\def\csname LT7\endcsname{\color[rgb]{1,0.3,0}}%
      \expandafter\def\csname LT8\endcsname{\color[rgb]{0.5,0.5,0.5}}%
    \else
      \def\colorrgb#1{\color{black}}%
      \def\colorgray#1{\color[gray]{#1}}%
      \expandafter\def\csname LTw\endcsname{\color{white}}%
      \expandafter\def\csname LTb\endcsname{\color{black}}%
      \expandafter\def\csname LTa\endcsname{\color{black}}%
      \expandafter\def\csname LT0\endcsname{\color{black}}%
      \expandafter\def\csname LT1\endcsname{\color{black}}%
      \expandafter\def\csname LT2\endcsname{\color{black}}%
      \expandafter\def\csname LT3\endcsname{\color{black}}%
      \expandafter\def\csname LT4\endcsname{\color{black}}%
      \expandafter\def\csname LT5\endcsname{\color{black}}%
      \expandafter\def\csname LT6\endcsname{\color{black}}%
      \expandafter\def\csname LT7\endcsname{\color{black}}%
      \expandafter\def\csname LT8\endcsname{\color{black}}%
    \fi
  \fi
    \setlength{\unitlength}{0.0500bp}%
    \ifx\gptboxheight\undefined%
      \newlength{\gptboxheight}%
      \newlength{\gptboxwidth}%
      \newsavebox{\gptboxtext}%
    \fi%
    \setlength{\fboxrule}{0.5pt}%
    \setlength{\fboxsep}{1pt}%
\begin{picture}(7480.00,3880.00)%
    \gplgaddtomacro\gplbacktext{%
      \csname LTb\endcsname%
      \put(569,476){\makebox(0,0)[r]{\strut{}$0$}}%
      \csname LTb\endcsname%
      \put(569,987){\makebox(0,0)[r]{\strut{}$100$}}%
      \csname LTb\endcsname%
      \put(569,1497){\makebox(0,0)[r]{\strut{}$200$}}%
      \csname LTb\endcsname%
      \put(569,2008){\makebox(0,0)[r]{\strut{}$300$}}%
      \csname LTb\endcsname%
      \put(569,2518){\makebox(0,0)[r]{\strut{}$400$}}%
      \csname LTb\endcsname%
      \put(569,3029){\makebox(0,0)[r]{\strut{}$500$}}%
      \csname LTb\endcsname%
      \put(1516,313){\makebox(0,0){\strut{}3}}%
      \csname LTb\endcsname%
      \put(2591,313){\makebox(0,0){\strut{}5}}%
      \csname LTb\endcsname%
      \put(3569,313){\makebox(0,0){\strut{}7}}%
      \csname LTb\endcsname%
      \put(4546,313){\makebox(0,0){\strut{}9}}%
      \csname LTb\endcsname%
      \put(5524,313){\makebox(0,0){\strut{}11}}%
      \csname LTb\endcsname%
      \put(6501,313){\makebox(0,0){\strut{}13}}%
    }%
    \gplgaddtomacro\gplfronttext{%
      \csname LTb\endcsname%
      \put(7,1880){\rotatebox{-270}{\makebox(0,0){\strut{}\textbf{latency (ms)}}}}%
      \csname LTb\endcsname%
      \put(4057,1){\makebox(0,0){\strut{}\textbf{\#sites}}}%
      \csname LTb\endcsname%
      \put(1395,3714){\makebox(0,0)[r]{\strut{}\textbf{FPaxos $\mathbf{f=1}$}}}%
      \csname LTb\endcsname%
      \put(1242,2846){\rotatebox{90}{\makebox(0,0){\strut{}\scriptsize{49\tiny{\%}}}}}%
      \csname LTb\endcsname%
      \put(2220,2801){\rotatebox{90}{\makebox(0,0){\strut{}\scriptsize{54\tiny{\%}}}}}%
      \csname LTb\endcsname%
      \put(3197,2546){\rotatebox{90}{\makebox(0,0){\strut{}\scriptsize{74\tiny{\%}}}}}%
      \csname LTb\endcsname%
      \put(4175,2367){\rotatebox{90}{\makebox(0,0){\strut{}\scriptsize{71\tiny{\%}}}}}%
      \csname LTb\endcsname%
      \put(5152,2368){\rotatebox{90}{\makebox(0,0){\strut{}\scriptsize{93\tiny{\%}}}}}%
      \csname LTb\endcsname%
      \put(6130,2368){\rotatebox{90}{\makebox(0,0){\strut{}\scriptsize{121\tiny{\%}}}}}%
      \csname LTb\endcsname%
      \put(1395,3476){\makebox(0,0)[r]{\strut{}\textbf{FPaxos $\mathbf{f=2}$}}}%
      \csname LTb\endcsname%
      \put(2347,2850){\rotatebox{90}{\makebox(0,0){\strut{}\scriptsize{58\tiny{\%}}}}}%
      \csname LTb\endcsname%
      \put(3324,3031){\rotatebox{90}{\makebox(0,0){\strut{}\scriptsize{119\tiny{\%}}}}}%
      \csname LTb\endcsname%
      \put(4302,2812){\rotatebox{90}{\makebox(0,0){\strut{}\scriptsize{115\tiny{\%}}}}}%
      \csname LTb\endcsname%
      \put(5279,2648){\rotatebox{90}{\makebox(0,0){\strut{}\scriptsize{125\tiny{\%}}}}}%
      \csname LTb\endcsname%
      \put(6257,2481){\rotatebox{90}{\makebox(0,0){\strut{}\scriptsize{135\tiny{\%}}}}}%
      \csname LTb\endcsname%
      \put(3318,3714){\makebox(0,0)[r]{\strut{}\textbf{Mencius}}}%
      \csname LTb\endcsname%
      \put(1496,3052){\rotatebox{90}{\makebox(0,0){\strut{}\scriptsize{63\tiny{\%}}}}}%
      \csname LTb\endcsname%
      \put(2474,3048){\rotatebox{90}{\makebox(0,0){\strut{}\scriptsize{72\tiny{\%}}}}}%
      \csname LTb\endcsname%
      \put(3451,2903){\rotatebox{90}{\makebox(0,0){\strut{}\scriptsize{107\tiny{\%}}}}}%
      \csname LTb\endcsname%
      \put(4429,2989){\rotatebox{90}{\makebox(0,0){\strut{}\scriptsize{133\tiny{\%}}}}}%
      \csname LTb\endcsname%
      \put(5407,2753){\rotatebox{90}{\makebox(0,0){\strut{}\scriptsize{137\tiny{\%}}}}}%
      \csname LTb\endcsname%
      \put(6384,2845){\rotatebox{90}{\makebox(0,0){\strut{}\scriptsize{182\tiny{\%}}}}}%
      \csname LTb\endcsname%
      \put(3318,3476){\makebox(0,0)[r]{\strut{}\textbf{EPaxos}}}%
      \csname LTb\endcsname%
      \put(1614,2154){\rotatebox{90}{\makebox(0,0){\strut{}\scriptsize{4\tiny{\%}}}}}%
      \csname LTb\endcsname%
      \put(2591,2089){\rotatebox{90}{\makebox(0,0){\strut{}\scriptsize{5\tiny{\%}}}}}%
      \csname LTb\endcsname%
      \put(3569,2310){\rotatebox{90}{\makebox(0,0){\strut{}\scriptsize{55\tiny{\%}}}}}%
      \csname LTb\endcsname%
      \put(4546,2091){\rotatebox{90}{\makebox(0,0){\strut{}\scriptsize{46\tiny{\%}}}}}%
      \csname LTb\endcsname%
      \put(5524,2092){\rotatebox{90}{\makebox(0,0){\strut{}\scriptsize{65\tiny{\%}}}}}%
      \csname LTb\endcsname%
      \put(6501,2012){\rotatebox{90}{\makebox(0,0){\strut{}\scriptsize{78\tiny{\%}}}}}%
      \csname LTb\endcsname%
      \put(5241,3714){\makebox(0,0)[r]{\strut{}\textbf{\SYS $\mathbf{f=1}$}}}%
      \csname LTb\endcsname%
      \put(1731,2132){\rotatebox{90}{\makebox(0,0){\strut{}\scriptsize{4\tiny{\%}}}}}%
      \csname LTb\endcsname%
      \put(2708,2038){\rotatebox{90}{\makebox(0,0){\strut{}\scriptsize{3\tiny{\%}}}}}%
      \csname LTb\endcsname%
      \put(3686,1737){\rotatebox{90}{\makebox(0,0){\strut{}\scriptsize{4\tiny{\%}}}}}%
      \csname LTb\endcsname%
      \put(4664,1673){\rotatebox{90}{\makebox(0,0){\strut{}\scriptsize{7\tiny{\%}}}}}%
      \csname LTb\endcsname%
      \put(5641,1572){\rotatebox{90}{\makebox(0,0){\strut{}\scriptsize{9\tiny{\%}}}}}%
      \csname LTb\endcsname%
      \put(6619,1483){\rotatebox{90}{\makebox(0,0){\strut{}\scriptsize{13\tiny{\%}}}}}%
      \csname LTb\endcsname%
      \put(5241,3476){\makebox(0,0)[r]{\strut{}\textbf{\SYS $\mathbf{f=2}$}}}%
      \csname LTb\endcsname%
      \put(2836,2311){\rotatebox{90}{\makebox(0,0){\strut{}\scriptsize{21\tiny{\%}}}}}%
      \csname LTb\endcsname%
      \put(3813,2115){\rotatebox{90}{\makebox(0,0){\strut{}\scriptsize{37\tiny{\%}}}}}%
      \csname LTb\endcsname%
      \put(4791,1898){\rotatebox{90}{\makebox(0,0){\strut{}\scriptsize{27\tiny{\%}}}}}%
      \csname LTb\endcsname%
      \put(5768,1740){\rotatebox{90}{\makebox(0,0){\strut{}\scriptsize{25\tiny{\%}}}}}%
      \csname LTb\endcsname%
      \put(6746,1652){\rotatebox{90}{\makebox(0,0){\strut{}\scriptsize{32\tiny{\%}}}}}%
      \csname LTb\endcsname%
      \put(7164,3714){\makebox(0,0)[r]{\strut{}\textbf{optimal}}}%
    }%
    \gplbacktext
    \put(0,0){\includegraphics{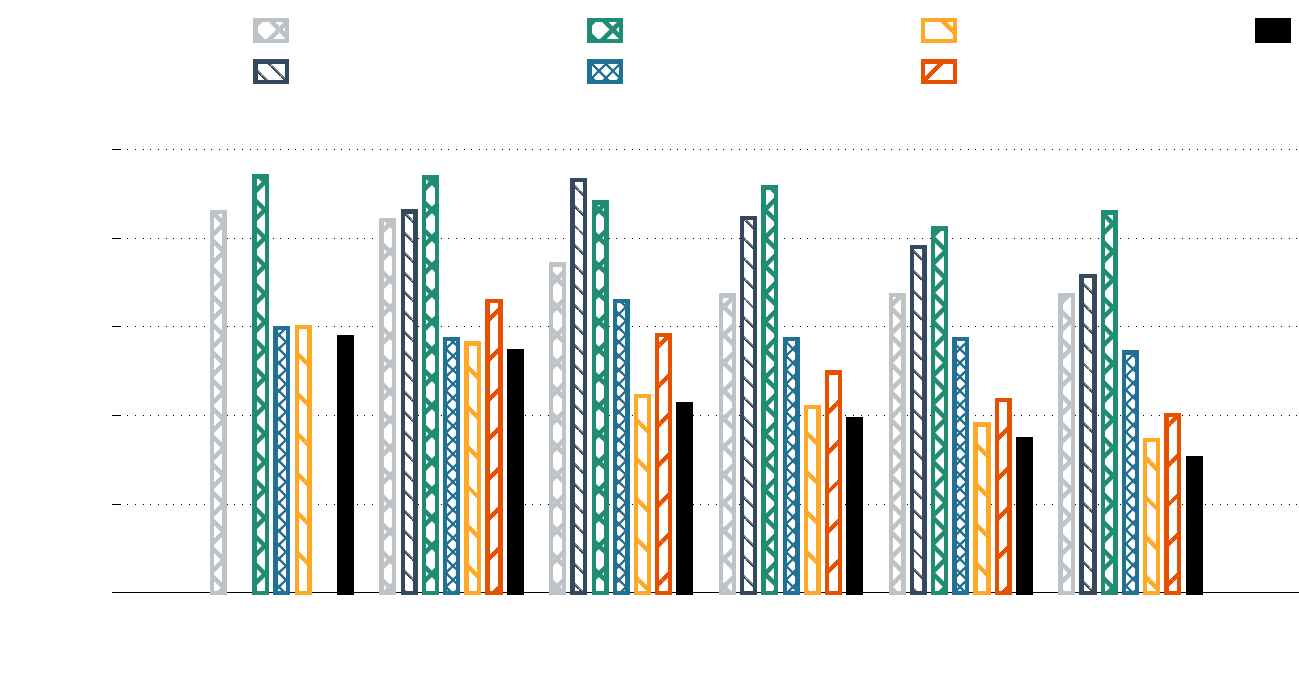}}%
    \gplfronttext
  \end{picture}%
\endgroup

%% file: new_relative_latency_plot_source.tex
\begingroup
  \fontfamily{ptm}%
  \selectfont
  \makeatletter
  \providecommand\color[2][]{%
    \GenericError{(gnuplot) \space\space\space\@spaces}{%
      Package color not loaded in conjunction with
      terminal option `colourtext'%
    }{See the gnuplot documentation for explanation.%
    }{Either use 'blacktext' in gnuplot or load the package
      color.sty in LaTeX.}%
    \renewcommand\color[2][]{}%
  }%
  \providecommand\includegraphics[2][]{%
    \GenericError{(gnuplot) \space\space\space\@spaces}{%
      Package graphicx or graphics not loaded%
    }{See the gnuplot documentation for explanation.%
    }{The gnuplot epslatex terminal needs graphicx.sty or graphics.sty.}%
    \renewcommand\includegraphics[2][]{}%
  }%
  \providecommand\rotatebox[2]{#2}%
  \@ifundefined{ifGPcolor}{%
    \newif\ifGPcolor
    \GPcolortrue
  }{}%
  \@ifundefined{ifGPblacktext}{%
    \newif\ifGPblacktext
    \GPblacktexttrue
  }{}%
  \let\gplgaddtomacro\g@addto@macro
  \gdef\gplbacktext{}%
  \gdef\gplfronttext{}%
  \makeatother
  \ifGPblacktext
    \def\colorrgb#1{}%
    \def\colorgray#1{}%
  \else
    \ifGPcolor
      \def\colorrgb#1{\color[rgb]{#1}}%
      \def\colorgray#1{\color[gray]{#1}}%
      \expandafter\def\csname LTw\endcsname{\color{white}}%
      \expandafter\def\csname LTb\endcsname{\color{black}}%
      \expandafter\def\csname LTa\endcsname{\color{black}}%
      \expandafter\def\csname LT0\endcsname{\color[rgb]{1,0,0}}%
      \expandafter\def\csname LT1\endcsname{\color[rgb]{0,1,0}}%
      \expandafter\def\csname LT2\endcsname{\color[rgb]{0,0,1}}%
      \expandafter\def\csname LT3\endcsname{\color[rgb]{1,0,1}}%
      \expandafter\def\csname LT4\endcsname{\color[rgb]{0,1,1}}%
      \expandafter\def\csname LT5\endcsname{\color[rgb]{1,1,0}}%
      \expandafter\def\csname LT6\endcsname{\color[rgb]{0,0,0}}%
      \expandafter\def\csname LT7\endcsname{\color[rgb]{1,0.3,0}}%
      \expandafter\def\csname LT8\endcsname{\color[rgb]{0.5,0.5,0.5}}%
    \else
      \def\colorrgb#1{\color{black}}%
      \def\colorgray#1{\color[gray]{#1}}%
      \expandafter\def\csname LTw\endcsname{\color{white}}%
      \expandafter\def\csname LTb\endcsname{\color{black}}%
      \expandafter\def\csname LTa\endcsname{\color{black}}%
      \expandafter\def\csname LT0\endcsname{\color{black}}%
      \expandafter\def\csname LT1\endcsname{\color{black}}%
      \expandafter\def\csname LT2\endcsname{\color{black}}%
      \expandafter\def\csname LT3\endcsname{\color{black}}%
      \expandafter\def\csname LT4\endcsname{\color{black}}%
      \expandafter\def\csname LT5\endcsname{\color{black}}%
      \expandafter\def\csname LT6\endcsname{\color{black}}%
      \expandafter\def\csname LT7\endcsname{\color{black}}%
      \expandafter\def\csname LT8\endcsname{\color{black}}%
    \fi
  \fi
    \setlength{\unitlength}{0.0500bp}%
    \ifx\gptboxheight\undefined%
      \newlength{\gptboxheight}%
      \newlength{\gptboxwidth}%
      \newsavebox{\gptboxtext}%
    \fi%
    \setlength{\fboxrule}{0.5pt}%
    \setlength{\fboxsep}{1pt}%
\begin{picture}(7480.00,3880.00)%
    \gplgaddtomacro\gplbacktext{%
      \csname LTb\endcsname%
      \put(770,2908){\makebox(0,0)[r]{\strut{}1$\textsf{x}$}}%
      \csname LTb\endcsname%
      \put(770,2449){\makebox(0,0)[r]{\strut{}1.5$\textsf{x}$}}%
      \csname LTb\endcsname%
      \put(770,1990){\makebox(0,0)[r]{\strut{}2$\textsf{x}$}}%
      \csname LTb\endcsname%
      \put(770,1532){\makebox(0,0)[r]{\strut{}2.5$\textsf{x}$}}%
      \csname LTb\endcsname%
      \put(770,1073){\makebox(0,0)[r]{\strut{}3$\textsf{x}$}}%
      \csname LTb\endcsname%
      \put(770,614){\makebox(0,0)[r]{\strut{}3.5$\textsf{x}$}}%
      \csname LTb\endcsname%
      \put(1106,298){\makebox(0,0){\strut{}3}}%
      \csname LTb\endcsname%
      \put(2313,298){\makebox(0,0){\strut{}5}}%
      \csname LTb\endcsname%
      \put(3521,298){\makebox(0,0){\strut{}7}}%
      \csname LTb\endcsname%
      \put(4728,298){\makebox(0,0){\strut{}9}}%
      \csname LTb\endcsname%
      \put(5936,298){\makebox(0,0){\strut{}11}}%
      \csname LTb\endcsname%
      \put(7143,298){\makebox(0,0){\strut{}13}}%
    }%
    \gplgaddtomacro\gplfronttext{%
      \csname LTb\endcsname%
      \put(108,1761){\rotatebox{-270}{\makebox(0,0){\strut{}\textbf{latency penalty}}}}%
      \csname LTb\endcsname%
      \put(4124,60){\makebox(0,0){\strut{}\textbf{\#sites}}}%
      \csname LTb\endcsname%
      \put(3284,3714){\makebox(0,0)[r]{\strut{}\textbf{FPaxos $\mathbf{f=1}$}}}%
      \csname LTb\endcsname%
      \put(3284,3476){\makebox(0,0)[r]{\strut{}\textbf{FPaxos $\mathbf{f=2}$}}}%
      \csname LTb\endcsname%
      \put(5207,3714){\makebox(0,0)[r]{\strut{}\textbf{Mencius}}}%
      \csname LTb\endcsname%
      \put(5207,3476){\makebox(0,0)[r]{\strut{}\textbf{EPaxos}}}%
      \csname LTb\endcsname%
      \put(7130,3714){\makebox(0,0)[r]{\strut{}\textbf{\SYS $\mathbf{f=1}$}}}%
      \csname LTb\endcsname%
      \put(7130,3476){\makebox(0,0)[r]{\strut{}\textbf{\SYS $\mathbf{f=2}$}}}%
    }%
    \gplbacktext
    \put(0,0){\includegraphics{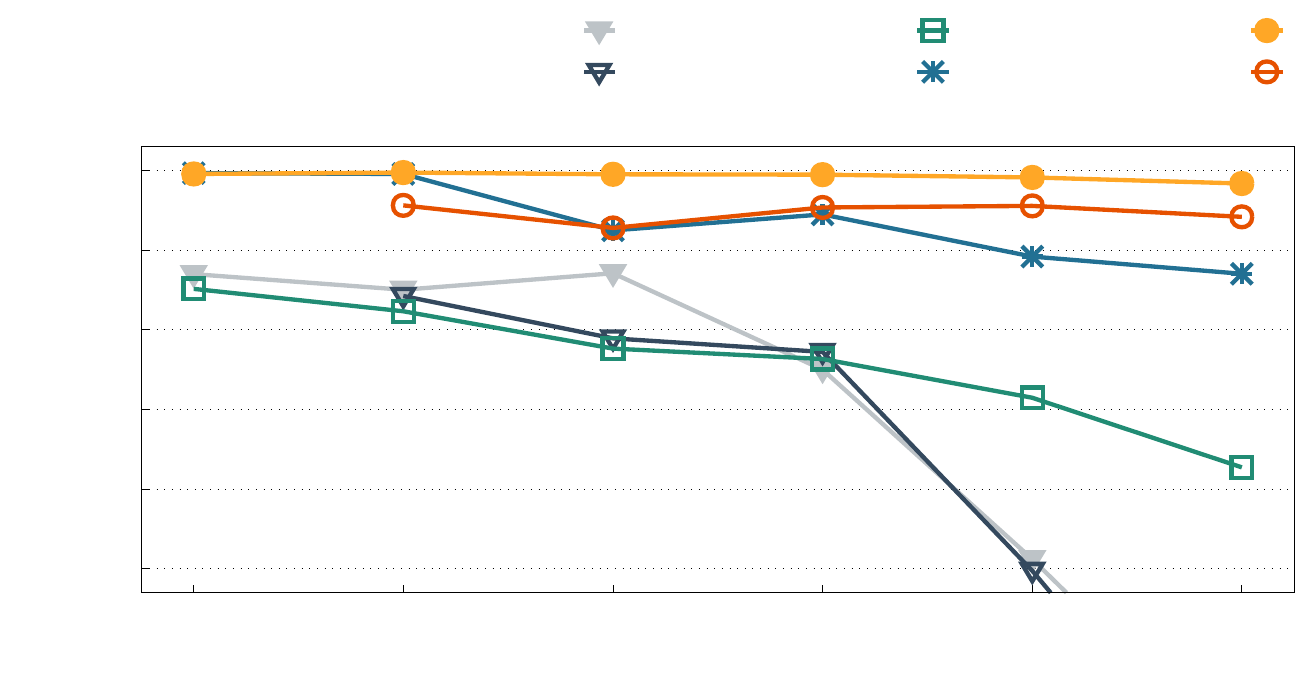}}%
    \gplfronttext
  \end{picture}%
\endgroup

%% file: new_tput_latency_plot_source.tex
\begingroup
  \fontfamily{ptm}%
  \selectfont
  \makeatletter
  \providecommand\color[2][]{%
    \GenericError{(gnuplot) \space\space\space\@spaces}{%
      Package color not loaded in conjunction with
      terminal option `colourtext'%
    }{See the gnuplot documentation for explanation.%
    }{Either use 'blacktext' in gnuplot or load the package
      color.sty in LaTeX.}%
    \renewcommand\color[2][]{}%
  }%
  \providecommand\includegraphics[2][]{%
    \GenericError{(gnuplot) \space\space\space\@spaces}{%
      Package graphicx or graphics not loaded%
    }{See the gnuplot documentation for explanation.%
    }{The gnuplot epslatex terminal needs graphicx.sty or graphics.sty.}%
    \renewcommand\includegraphics[2][]{}%
  }%
  \providecommand\rotatebox[2]{#2}%
  \@ifundefined{ifGPcolor}{%
    \newif\ifGPcolor
    \GPcolortrue
  }{}%
  \@ifundefined{ifGPblacktext}{%
    \newif\ifGPblacktext
    \GPblacktexttrue
  }{}%
  \let\gplgaddtomacro\g@addto@macro
  \gdef\gplbacktext{}%
  \gdef\gplfronttext{}%
  \makeatother
  \ifGPblacktext
    \def\colorrgb#1{}%
    \def\colorgray#1{}%
  \else
    \ifGPcolor
      \def\colorrgb#1{\color[rgb]{#1}}%
      \def\colorgray#1{\color[gray]{#1}}%
      \expandafter\def\csname LTw\endcsname{\color{white}}%
      \expandafter\def\csname LTb\endcsname{\color{black}}%
      \expandafter\def\csname LTa\endcsname{\color{black}}%
      \expandafter\def\csname LT0\endcsname{\color[rgb]{1,0,0}}%
      \expandafter\def\csname LT1\endcsname{\color[rgb]{0,1,0}}%
      \expandafter\def\csname LT2\endcsname{\color[rgb]{0,0,1}}%
      \expandafter\def\csname LT3\endcsname{\color[rgb]{1,0,1}}%
      \expandafter\def\csname LT4\endcsname{\color[rgb]{0,1,1}}%
      \expandafter\def\csname LT5\endcsname{\color[rgb]{1,1,0}}%
      \expandafter\def\csname LT6\endcsname{\color[rgb]{0,0,0}}%
      \expandafter\def\csname LT7\endcsname{\color[rgb]{1,0.3,0}}%
      \expandafter\def\csname LT8\endcsname{\color[rgb]{0.5,0.5,0.5}}%
    \else
      \def\colorrgb#1{\color{black}}%
      \def\colorgray#1{\color[gray]{#1}}%
      \expandafter\def\csname LTw\endcsname{\color{white}}%
      \expandafter\def\csname LTb\endcsname{\color{black}}%
      \expandafter\def\csname LTa\endcsname{\color{black}}%
      \expandafter\def\csname LT0\endcsname{\color{black}}%
      \expandafter\def\csname LT1\endcsname{\color{black}}%
      \expandafter\def\csname LT2\endcsname{\color{black}}%
      \expandafter\def\csname LT3\endcsname{\color{black}}%
      \expandafter\def\csname LT4\endcsname{\color{black}}%
      \expandafter\def\csname LT5\endcsname{\color{black}}%
      \expandafter\def\csname LT6\endcsname{\color{black}}%
      \expandafter\def\csname LT7\endcsname{\color{black}}%
      \expandafter\def\csname LT8\endcsname{\color{black}}%
    \fi
  \fi
    \setlength{\unitlength}{0.0500bp}%
    \ifx\gptboxheight\undefined%
      \newlength{\gptboxheight}%
      \newlength{\gptboxwidth}%
      \newsavebox{\gptboxtext}%
    \fi%
    \setlength{\fboxrule}{0.5pt}%
    \setlength{\fboxsep}{1pt}%
\begin{picture}(7200.00,3880.00)%
    \gplgaddtomacro\gplbacktext{%
      \csname LTb\endcsname%
      \put(636,714){\makebox(0,0)[r]{\strut{}$200$}}%
      \csname LTb\endcsname%
      \put(636,1050){\makebox(0,0)[r]{\strut{}$300$}}%
      \csname LTb\endcsname%
      \put(636,1386){\makebox(0,0)[r]{\strut{}$400$}}%
      \csname LTb\endcsname%
      \put(636,1722){\makebox(0,0)[r]{\strut{}$500$}}%
      \csname LTb\endcsname%
      \put(636,2059){\makebox(0,0)[r]{\strut{}$600$}}%
      \csname LTb\endcsname%
      \put(636,2395){\makebox(0,0)[r]{\strut{}$700$}}%
      \csname LTb\endcsname%
      \put(636,2731){\makebox(0,0)[r]{\strut{}$800$}}%
      \csname LTb\endcsname%
      \put(636,3067){\makebox(0,0)[r]{\strut{}$900$}}%
      \csname LTb\endcsname%
      \put(636,3403){\makebox(0,0)[r]{\strut{}$1000$}}%
      \csname LTb\endcsname%
      \put(1349,536){\makebox(0,0){\strut{}$100$}}%
      \csname LTb\endcsname%
      \put(2274,536){\makebox(0,0){\strut{}$1000$}}%
      \csname LTb\endcsname%
      \put(3198,536){\makebox(0,0){\strut{}$10000$}}%
    }%
    \gplgaddtomacro\gplfronttext{%
      \csname LTb\endcsname%
      \put(7,2058){\rotatebox{-270}{\makebox(0,0){\strut{}\textbf{latency (ms)}}}}%
      \csname LTb\endcsname%
      \put(2134,238){\makebox(0,0){\strut{}\textbf{throughput (ops/s)}}}%
    }%
    \gplgaddtomacro\gplbacktext{%
      \csname LTb\endcsname%
      \put(4236,714){\makebox(0,0)[r]{\strut{}$200$}}%
      \csname LTb\endcsname%
      \put(4236,1050){\makebox(0,0)[r]{\strut{}$300$}}%
      \csname LTb\endcsname%
      \put(4236,1386){\makebox(0,0)[r]{\strut{}$400$}}%
      \csname LTb\endcsname%
      \put(4236,1722){\makebox(0,0)[r]{\strut{}$500$}}%
      \csname LTb\endcsname%
      \put(4236,2059){\makebox(0,0)[r]{\strut{}$600$}}%
      \csname LTb\endcsname%
      \put(4236,2395){\makebox(0,0)[r]{\strut{}$700$}}%
      \csname LTb\endcsname%
      \put(4236,2731){\makebox(0,0)[r]{\strut{}$800$}}%
      \csname LTb\endcsname%
      \put(4236,3067){\makebox(0,0)[r]{\strut{}$900$}}%
      \csname LTb\endcsname%
      \put(4236,3403){\makebox(0,0)[r]{\strut{}$1000$}}%
      \csname LTb\endcsname%
      \put(4949,536){\makebox(0,0){\strut{}$100$}}%
      \csname LTb\endcsname%
      \put(5873,536){\makebox(0,0){\strut{}$1000$}}%
      \csname LTb\endcsname%
      \put(6797,536){\makebox(0,0){\strut{}$10000$}}%
    }%
    \gplgaddtomacro\gplfronttext{%
      \csname LTb\endcsname%
      \put(5734,238){\makebox(0,0){\strut{}\textbf{throughput (ops/s)}}}%
      \csname LTb\endcsname%
      \put(1081,3714){\makebox(0,0)[r]{\strut{}\textbf{FPaxos $\mathbf{f=1}$}}}%
      \csname LTb\endcsname%
      \put(3004,3714){\makebox(0,0)[r]{\strut{}\textbf{EPaxos}}}%
      \csname LTb\endcsname%
      \put(4927,3714){\makebox(0,0)[r]{\strut{}\textbf{\SYS $\mathbf{f=1}$}}}%
      \csname LTb\endcsname%
      \put(4987,1146){\rotatebox{-45}{\makebox(0,0){\small \textsf{8}}}}%
      \csname LTb\endcsname%
      \put(5264,1149){\rotatebox{-45}{\makebox(0,0){\small \textsf{16}}}}%
      \csname LTb\endcsname%
      \put(5538,1163){\rotatebox{-45}{\makebox(0,0){\small \textsf{32}}}}%
      \csname LTb\endcsname%
      \put(5812,1177){\rotatebox{-45}{\makebox(0,0){\small \textsf{64}}}}%
      \csname LTb\endcsname%
      \put(6091,1175){\rotatebox{-45}{\makebox(0,0){\small \textsf{128}}}}%
      \csname LTb\endcsname%
      \put(6363,1196){\rotatebox{-45}{\makebox(0,0){\small \textsf{256}}}}%
      \csname LTb\endcsname%
      \put(6521,1665){\rotatebox{-45}{\makebox(0,0){\small \textsf{512}}}}%
      \csname LTb\endcsname%
      \put(6850,3714){\makebox(0,0)[r]{\strut{}\textbf{\SYS $\mathbf{f=2}$}}}%
    }%
    \gplbacktext
    \put(0,0){\includegraphics{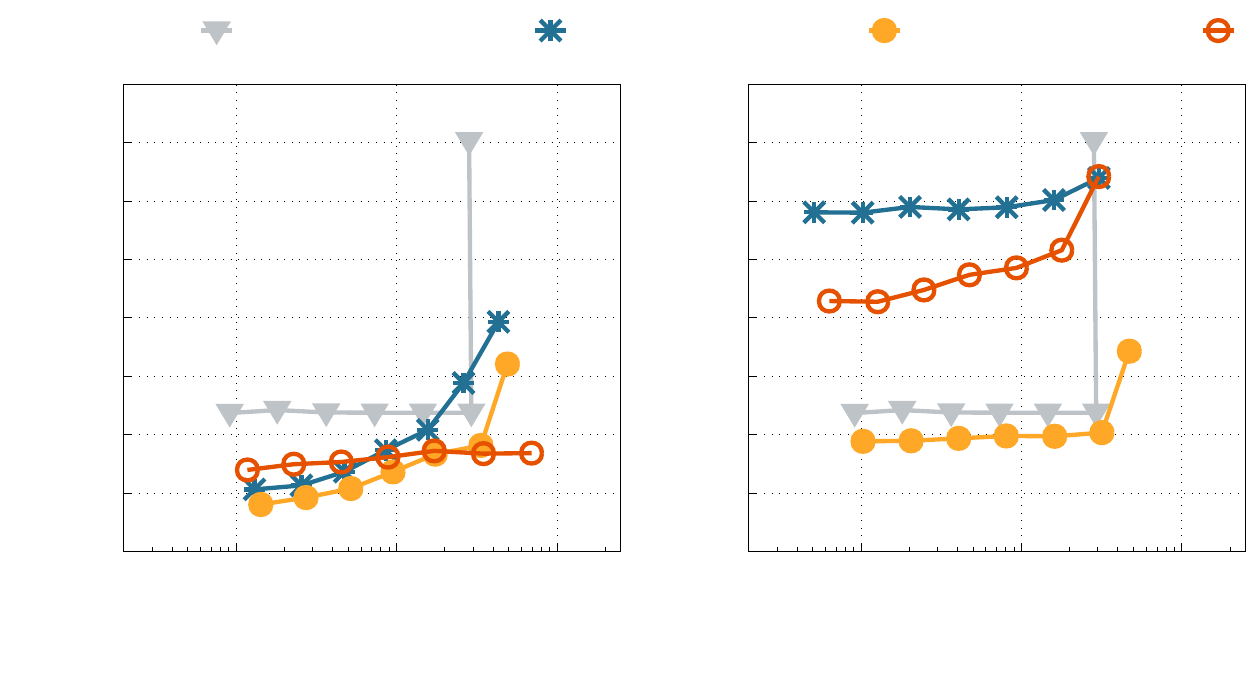}}%
    \gplfronttext
  \end{picture}%
\endgroup

%% file: recovery_multi_plot_source.tex
\begingroup
  \fontfamily{ptm}%
  \selectfont
  \makeatletter
  \providecommand\color[2][]{%
    \GenericError{(gnuplot) \space\space\space\@spaces}{%
      Package color not loaded in conjunction with
      terminal option `colourtext'%
    }{See the gnuplot documentation for explanation.%
    }{Either use 'blacktext' in gnuplot or load the package
      color.sty in LaTeX.}%
    \renewcommand\color[2][]{}%
  }%
  \providecommand\includegraphics[2][]{%
    \GenericError{(gnuplot) \space\space\space\@spaces}{%
      Package graphicx or graphics not loaded%
    }{See the gnuplot documentation for explanation.%
    }{The gnuplot epslatex terminal needs graphicx.sty or graphics.sty.}%
    \renewcommand\includegraphics[2][]{}%
  }%
  \providecommand\rotatebox[2]{#2}%
  \@ifundefined{ifGPcolor}{%
    \newif\ifGPcolor
    \GPcolortrue
  }{}%
  \@ifundefined{ifGPblacktext}{%
    \newif\ifGPblacktext
    \GPblacktexttrue
  }{}%
  \let\gplgaddtomacro\g@addto@macro
  \gdef\gplbacktext{}%
  \gdef\gplfronttext{}%
  \makeatother
  \ifGPblacktext
    \def\colorrgb#1{}%
    \def\colorgray#1{}%
  \else
    \ifGPcolor
      \def\colorrgb#1{\color[rgb]{#1}}%
      \def\colorgray#1{\color[gray]{#1}}%
      \expandafter\def\csname LTw\endcsname{\color{white}}%
      \expandafter\def\csname LTb\endcsname{\color{black}}%
      \expandafter\def\csname LTa\endcsname{\color{black}}%
      \expandafter\def\csname LT0\endcsname{\color[rgb]{1,0,0}}%
      \expandafter\def\csname LT1\endcsname{\color[rgb]{0,1,0}}%
      \expandafter\def\csname LT2\endcsname{\color[rgb]{0,0,1}}%
      \expandafter\def\csname LT3\endcsname{\color[rgb]{1,0,1}}%
      \expandafter\def\csname LT4\endcsname{\color[rgb]{0,1,1}}%
      \expandafter\def\csname LT5\endcsname{\color[rgb]{1,1,0}}%
      \expandafter\def\csname LT6\endcsname{\color[rgb]{0,0,0}}%
      \expandafter\def\csname LT7\endcsname{\color[rgb]{1,0.3,0}}%
      \expandafter\def\csname LT8\endcsname{\color[rgb]{0.5,0.5,0.5}}%
    \else
      \def\colorrgb#1{\color{black}}%
      \def\colorgray#1{\color[gray]{#1}}%
      \expandafter\def\csname LTw\endcsname{\color{white}}%
      \expandafter\def\csname LTb\endcsname{\color{black}}%
      \expandafter\def\csname LTa\endcsname{\color{black}}%
      \expandafter\def\csname LT0\endcsname{\color{black}}%
      \expandafter\def\csname LT1\endcsname{\color{black}}%
      \expandafter\def\csname LT2\endcsname{\color{black}}%
      \expandafter\def\csname LT3\endcsname{\color{black}}%
      \expandafter\def\csname LT4\endcsname{\color{black}}%
      \expandafter\def\csname LT5\endcsname{\color{black}}%
      \expandafter\def\csname LT6\endcsname{\color{black}}%
      \expandafter\def\csname LT7\endcsname{\color{black}}%
      \expandafter\def\csname LT8\endcsname{\color{black}}%
    \fi
  \fi
    \setlength{\unitlength}{0.0500bp}%
    \ifx\gptboxheight\undefined%
      \newlength{\gptboxheight}%
      \newlength{\gptboxwidth}%
      \newsavebox{\gptboxtext}%
    \fi%
    \setlength{\fboxrule}{0.5pt}%
    \setlength{\fboxsep}{1pt}%
\begin{picture}(7200.00,5320.00)%
    \gplgaddtomacro\gplbacktext{%
      \csname LTb\endcsname%
      \put(737,3493){\makebox(0,0)[r]{\strut{}$0$}}%
      \csname LTb\endcsname%
      \put(737,3758){\makebox(0,0)[r]{\strut{}$200$}}%
      \csname LTb\endcsname%
      \put(737,4022){\makebox(0,0)[r]{\strut{}$400$}}%
      \csname LTb\endcsname%
      \put(737,4287){\makebox(0,0)[r]{\strut{}$600$}}%
      \csname LTb\endcsname%
      \put(737,4552){\makebox(0,0)[r]{\strut{}$800$}}%
      \csname LTb\endcsname%
      \put(737,4816){\makebox(0,0)[r]{\strut{}$1000$}}%
      \csname LTb\endcsname%
      \put(737,5081){\makebox(0,0)[r]{\strut{}$1200$}}%
      \csname LTb\endcsname%
      \put(804,3315){\makebox(0,0){\strut{}$10$}}%
      \csname LTb\endcsname%
      \put(1264,3315){\makebox(0,0){\strut{}$20$}}%
      \csname LTb\endcsname%
      \put(1725,3315){\makebox(0,0){\strut{}$30$}}%
      \csname LTb\endcsname%
      \put(2185,3315){\makebox(0,0){\strut{}$40$}}%
      \csname LTb\endcsname%
      \put(2645,3315){\makebox(0,0){\strut{}$50$}}%
      \csname LTb\endcsname%
      \put(3106,3315){\makebox(0,0){\strut{}$60$}}%
      \csname LTb\endcsname%
      \put(3566,3315){\makebox(0,0){\strut{}$70$}}%
    }%
    \gplgaddtomacro\gplfronttext{%
      \csname LTb\endcsname%
      \put(7,4287){\rotatebox{-270}{\makebox(0,0){\strut{}\textbf{throughput (ops/s)}}}}%
      \csname LTb\endcsname%
      \put(2185,5260){\makebox(0,0){\strut{}\textbf{TW}}}%
    }%
    \gplgaddtomacro\gplbacktext{%
      \csname LTb\endcsname%
      \put(4337,3493){\makebox(0,0)[r]{\strut{}$0$}}%
      \csname LTb\endcsname%
      \put(4337,3758){\makebox(0,0)[r]{\strut{}$200$}}%
      \csname LTb\endcsname%
      \put(4337,4022){\makebox(0,0)[r]{\strut{}$400$}}%
      \csname LTb\endcsname%
      \put(4337,4287){\makebox(0,0)[r]{\strut{}$600$}}%
      \csname LTb\endcsname%
      \put(4337,4552){\makebox(0,0)[r]{\strut{}$800$}}%
      \csname LTb\endcsname%
      \put(4337,4816){\makebox(0,0)[r]{\strut{}$1000$}}%
      \csname LTb\endcsname%
      \put(4337,5081){\makebox(0,0)[r]{\strut{}$1200$}}%
      \csname LTb\endcsname%
      \put(4404,3315){\makebox(0,0){\strut{}$10$}}%
      \csname LTb\endcsname%
      \put(4864,3315){\makebox(0,0){\strut{}$20$}}%
      \csname LTb\endcsname%
      \put(5324,3315){\makebox(0,0){\strut{}$30$}}%
      \csname LTb\endcsname%
      \put(5785,3315){\makebox(0,0){\strut{}$40$}}%
      \csname LTb\endcsname%
      \put(6245,3315){\makebox(0,0){\strut{}$50$}}%
      \csname LTb\endcsname%
      \put(6705,3315){\makebox(0,0){\strut{}$60$}}%
      \csname LTb\endcsname%
      \put(7165,3315){\makebox(0,0){\strut{}$70$}}%
    }%
    \gplgaddtomacro\gplfronttext{%
      \csname LTb\endcsname%
      \put(5784,5260){\makebox(0,0){\strut{}\textbf{FI}}}%
    }%
    \gplgaddtomacro\gplbacktext{%
      \csname LTb\endcsname%
      \put(737,833){\makebox(0,0)[r]{\strut{}$0$}}%
      \csname LTb\endcsname%
      \put(737,1098){\makebox(0,0)[r]{\strut{}$200$}}%
      \csname LTb\endcsname%
      \put(737,1363){\makebox(0,0)[r]{\strut{}$400$}}%
      \csname LTb\endcsname%
      \put(737,1628){\makebox(0,0)[r]{\strut{}$600$}}%
      \csname LTb\endcsname%
      \put(737,1892){\makebox(0,0)[r]{\strut{}$800$}}%
      \csname LTb\endcsname%
      \put(737,2157){\makebox(0,0)[r]{\strut{}$1000$}}%
      \csname LTb\endcsname%
      \put(737,2422){\makebox(0,0)[r]{\strut{}$1200$}}%
      \csname LTb\endcsname%
      \put(804,655){\makebox(0,0){\strut{}$10$}}%
      \csname LTb\endcsname%
      \put(1264,655){\makebox(0,0){\strut{}$20$}}%
      \csname LTb\endcsname%
      \put(1725,655){\makebox(0,0){\strut{}$30$}}%
      \csname LTb\endcsname%
      \put(2185,655){\makebox(0,0){\strut{}$40$}}%
      \csname LTb\endcsname%
      \put(2645,655){\makebox(0,0){\strut{}$50$}}%
      \csname LTb\endcsname%
      \put(3106,655){\makebox(0,0){\strut{}$60$}}%
      \csname LTb\endcsname%
      \put(3566,655){\makebox(0,0){\strut{}$70$}}%
    }%
    \gplgaddtomacro\gplfronttext{%
      \csname LTb\endcsname%
      \put(7,1627){\rotatebox{-270}{\makebox(0,0){\strut{}\textbf{throughput (ops/s)}}}}%
      \csname LTb\endcsname%
      \put(2185,417){\makebox(0,0){\strut{}\textbf{time (s)}}}%
      \csname LTb\endcsname%
      \put(2185,2601){\makebox(0,0){\strut{}\textbf{SC}}}%
    }%
    \gplgaddtomacro\gplbacktext{%
      \csname LTb\endcsname%
      \put(4337,833){\makebox(0,0)[r]{\strut{}$0$}}%
      \csname LTb\endcsname%
      \put(4337,1098){\makebox(0,0)[r]{\strut{}$500$}}%
      \csname LTb\endcsname%
      \put(4337,1363){\makebox(0,0)[r]{\strut{}$1000$}}%
      \csname LTb\endcsname%
      \put(4337,1628){\makebox(0,0)[r]{\strut{}$1500$}}%
      \csname LTb\endcsname%
      \put(4337,1892){\makebox(0,0)[r]{\strut{}$2000$}}%
      \csname LTb\endcsname%
      \put(4337,2157){\makebox(0,0)[r]{\strut{}$2500$}}%
      \csname LTb\endcsname%
      \put(4337,2422){\makebox(0,0)[r]{\strut{}$3000$}}%
      \csname LTb\endcsname%
      \put(4404,655){\makebox(0,0){\strut{}$10$}}%
      \csname LTb\endcsname%
      \put(4864,655){\makebox(0,0){\strut{}$20$}}%
      \csname LTb\endcsname%
      \put(5324,655){\makebox(0,0){\strut{}$30$}}%
      \csname LTb\endcsname%
      \put(5785,655){\makebox(0,0){\strut{}$40$}}%
      \csname LTb\endcsname%
      \put(6245,655){\makebox(0,0){\strut{}$50$}}%
      \csname LTb\endcsname%
      \put(6705,655){\makebox(0,0){\strut{}$60$}}%
      \csname LTb\endcsname%
      \put(7165,655){\makebox(0,0){\strut{}$70$}}%
    }%
    \gplgaddtomacro\gplfronttext{%
      \csname LTb\endcsname%
      \put(5784,417){\makebox(0,0){\strut{}\textbf{time (s)}}}%
      \csname LTb\endcsname%
      \put(5784,2601){\makebox(0,0){\strut{}\textbf{all sites}}}%
      \csname LTb\endcsname%
      \put(3388,131){\makebox(0,0)[r]{\strut{}\textbf{Paxos}}}%
      \csname LTb\endcsname%
      \put(4641,131){\makebox(0,0)[r]{\strut{}\textbf{\SYS}}}%
    }%
    \gplbacktext
    \put(0,0){\includegraphics{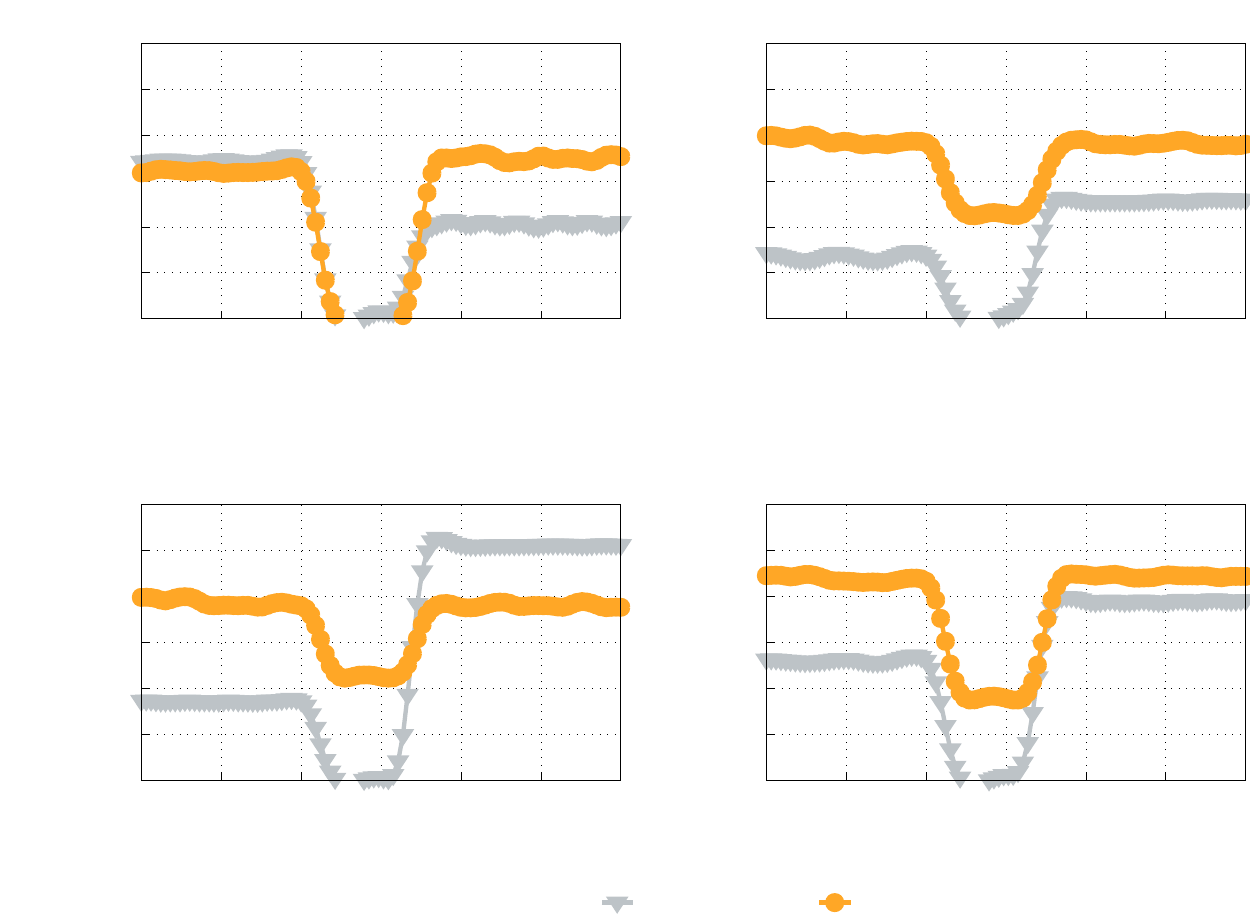}}%
    \gplfronttext
  \end{picture}%
\endgroup

%% file: new_ycsb_reads_commute_plot_source.tex
\begingroup
  \fontfamily{ptm}%
  \selectfont
  \makeatletter
  \providecommand\color[2][]{%
    \GenericError{(gnuplot) \space\space\space\@spaces}{%
      Package color not loaded in conjunction with
      terminal option `colourtext'%
    }{See the gnuplot documentation for explanation.%
    }{Either use 'blacktext' in gnuplot or load the package
      color.sty in LaTeX.}%
    \renewcommand\color[2][]{}%
  }%
  \providecommand\includegraphics[2][]{%
    \GenericError{(gnuplot) \space\space\space\@spaces}{%
      Package graphicx or graphics not loaded%
    }{See the gnuplot documentation for explanation.%
    }{The gnuplot epslatex terminal needs graphicx.sty or graphics.sty.}%
    \renewcommand\includegraphics[2][]{}%
  }%
  \providecommand\rotatebox[2]{#2}%
  \@ifundefined{ifGPcolor}{%
    \newif\ifGPcolor
    \GPcolortrue
  }{}%
  \@ifundefined{ifGPblacktext}{%
    \newif\ifGPblacktext
    \GPblacktexttrue
  }{}%
  \let\gplgaddtomacro\g@addto@macro
  \gdef\gplbacktext{}%
  \gdef\gplfronttext{}%
  \makeatother
  \ifGPblacktext
    \def\colorrgb#1{}%
    \def\colorgray#1{}%
  \else
    \ifGPcolor
      \def\colorrgb#1{\color[rgb]{#1}}%
      \def\colorgray#1{\color[gray]{#1}}%
      \expandafter\def\csname LTw\endcsname{\color{white}}%
      \expandafter\def\csname LTb\endcsname{\color{black}}%
      \expandafter\def\csname LTa\endcsname{\color{black}}%
      \expandafter\def\csname LT0\endcsname{\color[rgb]{1,0,0}}%
      \expandafter\def\csname LT1\endcsname{\color[rgb]{0,1,0}}%
      \expandafter\def\csname LT2\endcsname{\color[rgb]{0,0,1}}%
      \expandafter\def\csname LT3\endcsname{\color[rgb]{1,0,1}}%
      \expandafter\def\csname LT4\endcsname{\color[rgb]{0,1,1}}%
      \expandafter\def\csname LT5\endcsname{\color[rgb]{1,1,0}}%
      \expandafter\def\csname LT6\endcsname{\color[rgb]{0,0,0}}%
      \expandafter\def\csname LT7\endcsname{\color[rgb]{1,0.3,0}}%
      \expandafter\def\csname LT8\endcsname{\color[rgb]{0.5,0.5,0.5}}%
    \else
      \def\colorrgb#1{\color{black}}%
      \def\colorgray#1{\color[gray]{#1}}%
      \expandafter\def\csname LTw\endcsname{\color{white}}%
      \expandafter\def\csname LTb\endcsname{\color{black}}%
      \expandafter\def\csname LTa\endcsname{\color{black}}%
      \expandafter\def\csname LT0\endcsname{\color{black}}%
      \expandafter\def\csname LT1\endcsname{\color{black}}%
      \expandafter\def\csname LT2\endcsname{\color{black}}%
      \expandafter\def\csname LT3\endcsname{\color{black}}%
      \expandafter\def\csname LT4\endcsname{\color{black}}%
      \expandafter\def\csname LT5\endcsname{\color{black}}%
      \expandafter\def\csname LT6\endcsname{\color{black}}%
      \expandafter\def\csname LT7\endcsname{\color{black}}%
      \expandafter\def\csname LT8\endcsname{\color{black}}%
    \fi
  \fi
    \setlength{\unitlength}{0.0500bp}%
    \ifx\gptboxheight\undefined%
      \newlength{\gptboxheight}%
      \newlength{\gptboxwidth}%
      \newsavebox{\gptboxtext}%
    \fi%
    \setlength{\fboxrule}{0.5pt}%
    \setlength{\fboxsep}{1pt}%
\begin{picture}(7340.00,7200.00)%
    \gplgaddtomacro\gplbacktext{%
      \csname LTb\endcsname%
      \put(569,4076){\makebox(0,0)[r]{\strut{}0}}%
      \csname LTb\endcsname%
      \put(569,4605){\makebox(0,0)[r]{\strut{}1}}%
      \csname LTb\endcsname%
      \put(569,5135){\makebox(0,0)[r]{\strut{}2}}%
      \csname LTb\endcsname%
      \put(569,5664){\makebox(0,0)[r]{\strut{}3}}%
      \csname LTb\endcsname%
      \put(569,6194){\makebox(0,0)[r]{\strut{}4}}%
      \csname LTb\endcsname%
      \put(1970,3864){\makebox(0,0){\strut{}20{\small\%}-80{\small\%}}}%
      \csname LTb\endcsname%
      \put(3304,3864){\makebox(0,0){\strut{}50{\small\%}-50{\small\%}}}%
      \csname LTb\endcsname%
      \put(4637,3864){\makebox(0,0){\strut{}80{\small\%}-20{\small\%}}}%
      \csname LTb\endcsname%
      \put(5971,3864){\makebox(0,0){\strut{}100{\small\%}-0{\small\%}}}%
    }%
    \gplgaddtomacro\gplfronttext{%
      \csname LTb\endcsname%
      \put(141,5399){\rotatebox{-270}{\makebox(0,0){\strut{}\textbf{throughput (Kops/s)}}}}%
      \csname LTb\endcsname%
      \put(3970,3565){\makebox(0,0){\strut{}\textbf{read-write percentages}}}%
      \csname LTb\endcsname%
      \put(2474,7034){\makebox(0,0)[r]{\strut{}\textbf{\textbf{EPaxos}}}}%
      \csname LTb\endcsname%
      \put(2474,6796){\makebox(0,0)[r]{\strut{}\textbf{*\textbf{EPaxos}}}}%
      \csname LTb\endcsname%
      \put(1703,5224){\makebox(0,0){\strut{}\scriptsize{1.1\tiny{\textsf{x}}}}}%
      \csname LTb\endcsname%
      \put(3037,5427){\makebox(0,0){\strut{}\scriptsize{1.2\tiny{\textsf{x}}}}}%
      \csname LTb\endcsname%
      \put(4371,5747){\makebox(0,0){\strut{}\scriptsize{1.3\tiny{\textsf{x}}}}}%
      \csname LTb\endcsname%
      \put(5704,6393){\makebox(0,0){\strut{}\scriptsize{1.3\tiny{\textsf{x}}}}}%
      \csname LTb\endcsname%
      \put(4732,7034){\makebox(0,0)[r]{\strut{}\textbf{\SYS $\mathbf{f=1}$}}}%
      \csname LTb\endcsname%
      \put(1890,5924){\makebox(0,0){\strut{}\scriptsize{1.7\tiny{\textsf{x}}}}}%
      \csname LTb\endcsname%
      \put(3224,6035){\makebox(0,0){\strut{}\scriptsize{1.7\tiny{\textsf{x}}}}}%
      \csname LTb\endcsname%
      \put(4557,6209){\makebox(0,0){\strut{}\scriptsize{1.6\tiny{\textsf{x}}}}}%
      \csname LTb\endcsname%
      \put(5891,6472){\makebox(0,0){\strut{}\scriptsize{1.3\tiny{\textsf{x}}}}}%
      \csname LTb\endcsname%
      \put(4732,6796){\makebox(0,0)[r]{\strut{}\textbf{*\SYS $\mathbf{f=1}$}}}%
      \csname LTb\endcsname%
      \put(2103,5980){\makebox(0,0){\strut{}\scriptsize{1.9\tiny{\textsf{x}}}}}%
      \csname LTb\endcsname%
      \put(3437,6148){\makebox(0,0){\strut{}\scriptsize{1.9\tiny{\textsf{x}}}}}%
      \csname LTb\endcsname%
      \put(4771,6292){\makebox(0,0){\strut{}\scriptsize{1.8\tiny{\textsf{x}}}}}%
      \csname LTb\endcsname%
      \put(6105,6393){\makebox(0,0){\strut{}\scriptsize{1.3\tiny{\textsf{x}}}}}%
      \csname LTb\endcsname%
      \put(6990,7034){\makebox(0,0)[r]{\strut{}\textbf{\SYS $\mathbf{f=2}$}}}%
      \csname LTb\endcsname%
      \put(2277,5642){\makebox(0,0){\strut{}\scriptsize{1.5\tiny{\textsf{x}}}}}%
      \csname LTb\endcsname%
      \put(3610,5652){\makebox(0,0){\strut{}\scriptsize{1.4\tiny{\textsf{x}}}}}%
      \csname LTb\endcsname%
      \put(4944,5706){\makebox(0,0){\strut{}\scriptsize{1.3\tiny{\textsf{x}}}}}%
      \csname LTb\endcsname%
      \put(6278,5849){\makebox(0,0){\strut{}\scriptsize{1\tiny{\textsf{x}}}}}%
      \csname LTb\endcsname%
      \put(6990,6796){\makebox(0,0)[r]{\strut{}\textbf{*\SYS $\mathbf{f=2}$}}}%
      \csname LTb\endcsname%
      \put(2477,5716){\makebox(0,0){\strut{}\scriptsize{1.6\tiny{\textsf{x}}}}}%
      \csname LTb\endcsname%
      \put(3810,5862){\makebox(0,0){\strut{}\scriptsize{1.6\tiny{\textsf{x}}}}}%
      \csname LTb\endcsname%
      \put(5144,6076){\makebox(0,0){\strut{}\scriptsize{1.6\tiny{\textsf{x}}}}}%
      \csname LTb\endcsname%
      \put(6478,6393){\makebox(0,0){\strut{}\scriptsize{1.3\tiny{\textsf{x}}}}}%
    }%
    \gplgaddtomacro\gplbacktext{%
      \csname LTb\endcsname%
      \put(569,3124){\makebox(0,0)[r]{\strut{}0}}%
      \csname LTb\endcsname%
      \put(569,2683){\makebox(0,0)[r]{\strut{}2}}%
      \csname LTb\endcsname%
      \put(569,2241){\makebox(0,0)[r]{\strut{}4}}%
      \csname LTb\endcsname%
      \put(569,1800){\makebox(0,0)[r]{\strut{}6}}%
      \csname LTb\endcsname%
      \put(569,1359){\makebox(0,0)[r]{\strut{}8}}%
      \csname LTb\endcsname%
      \put(569,917){\makebox(0,0)[r]{\strut{}10}}%
      \csname LTb\endcsname%
      \put(569,476){\makebox(0,0)[r]{\strut{}12}}%
      \csname LTb\endcsname%
      \put(1970,3256){\makebox(0,0){\strut{}20{\small\%}-80{\small\%}}}%
      \csname LTb\endcsname%
      \put(3304,3256){\makebox(0,0){\strut{}50{\small\%}-50{\small\%}}}%
      \csname LTb\endcsname%
      \put(4637,3256){\makebox(0,0){\strut{}80{\small\%}-20{\small\%}}}%
      \csname LTb\endcsname%
      \put(5971,3256){\makebox(0,0){\strut{}100{\small\%}-0{\small\%}}}%
    }%
    \gplgaddtomacro\gplfronttext{%
      \csname LTb\endcsname%
      \put(74,1800){\rotatebox{-270}{\makebox(0,0){\strut{}\textbf{throughput (Kops/s)}}}}%
      \csname LTb\endcsname%
      \put(1703,2238){\makebox(0,0){\strut{}\scriptsize{1.1\tiny{\textsf{x}}}}}%
      \csname LTb\endcsname%
      \put(3037,1997){\makebox(0,0){\strut{}\scriptsize{1.4\tiny{\textsf{x}}}}}%
      \csname LTb\endcsname%
      \put(4371,1596){\makebox(0,0){\strut{}\scriptsize{1.5\tiny{\textsf{x}}}}}%
      \csname LTb\endcsname%
      \put(5704,708){\makebox(0,0){\strut{}\scriptsize{1.5\tiny{\textsf{x}}}}}%
      \csname LTb\endcsname%
      \put(1890,1770){\makebox(0,0){\strut{}\scriptsize{1.9\tiny{\textsf{x}}}}}%
      \csname LTb\endcsname%
      \put(3224,1710){\makebox(0,0){\strut{}\scriptsize{1.8\tiny{\textsf{x}}}}}%
      \csname LTb\endcsname%
      \put(4557,1389){\makebox(0,0){\strut{}\scriptsize{1.8\tiny{\textsf{x}}}}}%
      \csname LTb\endcsname%
      \put(5891,797){\makebox(0,0){\strut{}\scriptsize{1.4\tiny{\textsf{x}}}}}%
      \csname LTb\endcsname%
      \put(2103,1633){\makebox(0,0){\strut{}\scriptsize{2\tiny{\textsf{x}}}}}%
      \csname LTb\endcsname%
      \put(3437,1363){\makebox(0,0){\strut{}\scriptsize{2.2\tiny{\textsf{x}}}}}%
      \csname LTb\endcsname%
      \put(4771,963){\makebox(0,0){\strut{}\scriptsize{2.2\tiny{\textsf{x}}}}}%
      \csname LTb\endcsname%
      \put(6105,708){\makebox(0,0){\strut{}\scriptsize{1.5\tiny{\textsf{x}}}}}%
      \csname LTb\endcsname%
      \put(2277,1720){\makebox(0,0){\strut{}\scriptsize{1.9\tiny{\textsf{x}}}}}%
      \csname LTb\endcsname%
      \put(3610,1623){\makebox(0,0){\strut{}\scriptsize{1.9\tiny{\textsf{x}}}}}%
      \csname LTb\endcsname%
      \put(4944,1460){\makebox(0,0){\strut{}\scriptsize{1.7\tiny{\textsf{x}}}}}%
      \csname LTb\endcsname%
      \put(6278,1181){\makebox(0,0){\strut{}\scriptsize{1.2\tiny{\textsf{x}}}}}%
      \csname LTb\endcsname%
      \put(2477,1596){\makebox(0,0){\strut{}\scriptsize{2.1\tiny{\textsf{x}}}}}%
      \csname LTb\endcsname%
      \put(3810,1359){\makebox(0,0){\strut{}\scriptsize{2.3\tiny{\textsf{x}}}}}%
      \csname LTb\endcsname%
      \put(5144,1093){\makebox(0,0){\strut{}\scriptsize{2.1\tiny{\textsf{x}}}}}%
      \csname LTb\endcsname%
      \put(6478,620){\makebox(0,0){\strut{}\scriptsize{1.5\tiny{\textsf{x}}}}}%
    }%
    \gplbacktext
    \put(0,0){\includegraphics{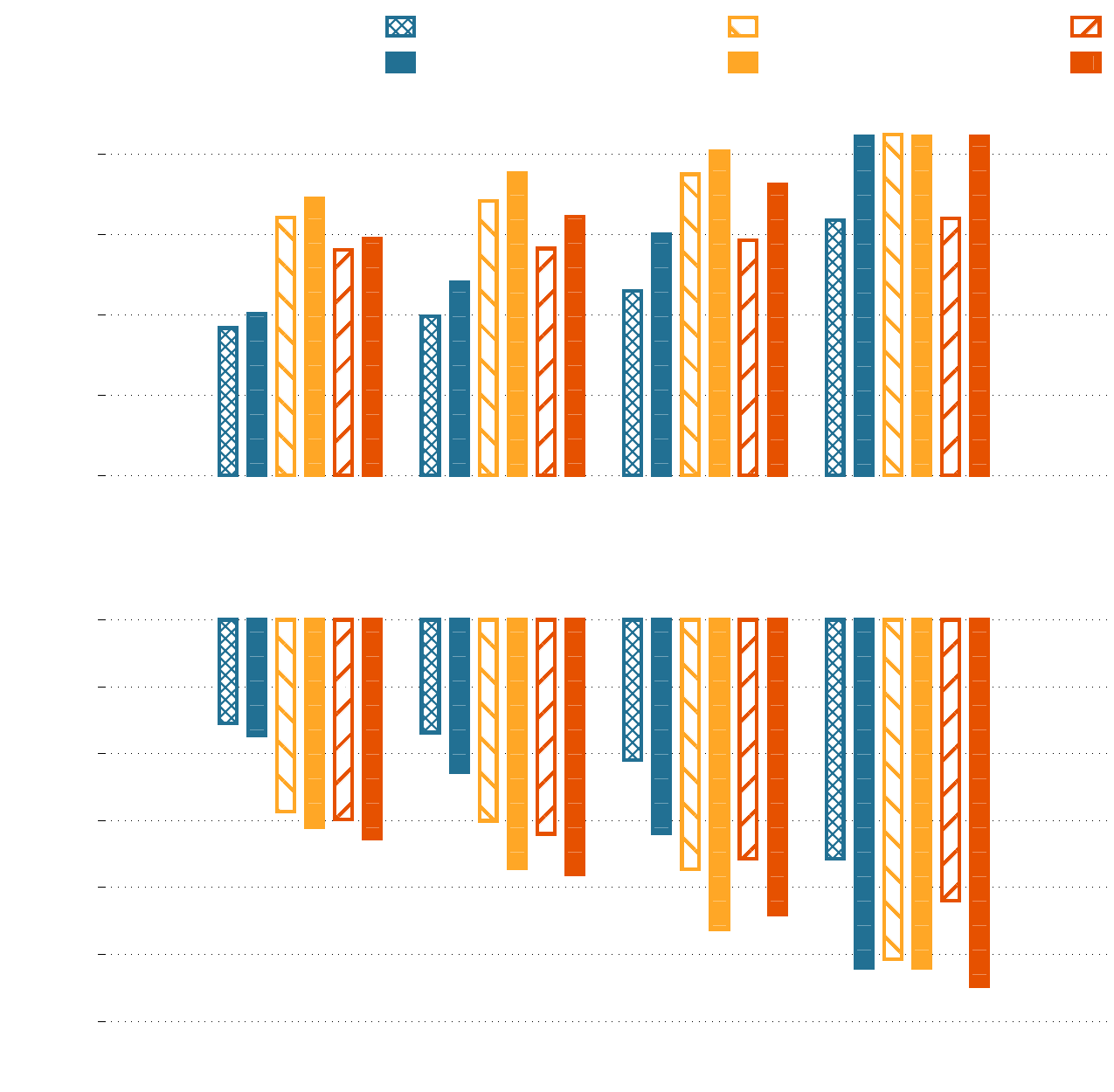}}%
    \gplfronttext
  \end{picture}%
\endgroup

%% file: related.tex
\section{Related Work}
\labsection{related}

The classical way of implementing SMR is by funneling all commands through a
single leader replica~\cite{paxos,zab,raft,vr}, which impairs scalability. A way
to mitigate this problem is to distribute the leader responsibilities
round-robin among replicas, as done in Mencius~\cite{mencius}. However, this
makes the system run at the speed of the slowest replica.

Exploiting commutativity to improve the scalability of SMR was first proposed in
Generalized Paxos~\cite{gb} and Generic Broadcast~\cite{gpaxos}. These protocols
still rely on a leader to order concurrent non-commuting commands, which creates
a bottleneck.

The closest SMR protocol to ours is EPaxos~\cite{epaxos}, which is also
leaderless and exploits commutativity. We compared \SYS with EPaxos in detail in
\refsection{properties}. There have been two follow-up protocols to EPaxos,
Alvin~\cite{alvin} and Caesar~\cite{caesar}. \SYS compares to these protocols
similarly to EPaxos; in particular, both follow-ups have large fast quorums that
depend on the overall number of processes only.

Flexible Paxos~\cite{flexible-paxos} reduces the size of Paxos Phase 2 quorums
to $f+1$, a technique we also use on the slow path of \SYS.
However, this technique is not directly applicable to computing dependencies via
fast path, as required by leaderless SMR. To the best of our knowledge, \SYS is
the first protocol to reduce the size of fast quorums to
$\floor{\frac{n}{2}}+f$.

An approach to scaling SMR is to shard the state of the application being
replicated and add cross-shard coordination to preserve
consistency~\cite{fernando-smr}. Such approaches build on a non-sharded SMR
protocol and are hence orthogonal to our proposal: \SYS can be combined with
them to scale SMR even further. %
Protocols such as M2Paxos~\cite{m2paxos}, WPaxos~\cite{wpaxos} and DPaxos~\cite{dpaxos} scale up SMR using a variation of the sharding approach.
These protocols exploit access locality by optimizing for workloads where commands do not frequently access objects in multiple locations.

There have been recent proposals of SMR protocols that improve scalability using
special hardware capabilities, such as low-latency switches or
RDMA~\cite{Li:2016,Dang:2015,apus}. However, currently these protocols 
work within a single data center only.

%% file: conclusion.tex
\section{Conclusion}
\labsection{conclusion}

This paper presented \SYS, the first leaderless SMR protocol parameterized with the number of allowed failures.
\SYS is designed for planet-scale systems where concurrent site failures are rare.
It uses tight quorums, executes a high percentage of the operations within a single round trip and executes quick linearizable reads.
As demonstrated empirically with large-scale experiments in Google Cloud Platform, all these innovations pay off in practice:
adding new nearby replicas improves client-perceived latency, and expanding to new locations maintains the system performance.
Compared to the state of the art, \SYS consistently outperforms existing protocols:
it is up to two times faster than Flexible Paxos with identical failure assumptions,
and more than doubles the performance of EPaxos in mixed read-write workloads.

%% file: thanks.tex
\section*{Acknowledgments}

We thank Lennart Oldenburg for his valuable feedback on early versions of this paper. We also thank our shepherd, Liuba Shrira, and the anonymous reviewers for their comments and suggestions. Vitor Enes was supported by an FCT
PhD Fellowship (PD/BD/142927/2018).
Tuanir Fran\c{c}a Rezende and Pierre Sutra were supported by EU H2020 grant No 825184 and ANR grant 16-CE25-0013-04.
Alexey Gotsman was supported by an ERC Starting Grant RACCOON.
This work was partially supported by the Google Cloud Platform research credits program.

%% file: atlas_full_spec.tex
\section{The \SYS Protocol and its Correctness}
\label{app:On}

\subsection{Protocol}

\begin{figure*}[!h]
  \begin{algorithm}[H]
  \begin{multicols}{2}
    \algosetup
    \setcounter{AlgoLine}{0}
    \SubAlgo{\Fun $\KwPCD(c)$\label{algo:full:request}}{
      $\id \leftarrow \tup{i, \mathsf{min}\{l \mid \tup{i, l} \in \KwPhAset \}}$ \\
      $\past \leftarrow \conflicts(c)$ \label{algo:full:past}\;
      $Q \leftarrow \fastq(i)$ \label{algo:full:fast-quorum}\;
      \Send $\KwMCol(\id, c, \past, Q)$ \To $Q$ \label{algo:full:collect-past} \;
    }

    \SubAlgo{\RecA $\KwMCol(\id, c, \past, Q)$ \RecB $j$}{
      \Pre $\id \in \KwPhAset$ \label{algo:full:collect-pre}\;
      $\KwMsgCcl[\id] \leftarrow \conflicts(c) \union \past$ \label{algo:full:conflicts}\;
      $\KwMsgM[\id] \leftarrow c$;
      $\KwMsgQ[\id] \leftarrow Q$\;
      $\KwMsgPh[\id] \leftarrow \KwPhB$\;
      \Send $\KwMColA(\id, \KwMsgCcl[\id])$ \To $j$ \label{algo:full:collect-ack}
    }

    \SubAlgo{\RecA $\KwMColA(\id, \dep_j)$ \RecBAll $j \in Q$\label{algo:full:collect-ack-receive}}{
      \Pre $\id \in \KwPhBset \land Q = \KwMsgQ[\id]$ \label{algo:full:collect-ack-pre}\;
      $D = \bigcup_{Q} \dep$ \label{algo:full:fp-proposal} \;
      \If{$\bigcup_{Q} \dep = \quorumcup{f}{Q} \dep$\label{algo:full:fp-condition}}{
        \Send $\KwMCom(\id, \KwMsgM[\id], D)$ \ToAll \label{algo:full:fp}
      }
      \Else{
        $Q' \leftarrow \slowq(i)$ \label{algo:full:slow-quorum}\;
        \Send $\KwMCons(\id, \KwMsgM[\id], D, i)$ \To $Q'$ \label{algo:full:sp}
      }
    }

    \SubAlgo{\RecA $\KwMCons(\id, c, D, b)$ \RecB $j$} {
      \Pre $\KwMsgR[\id] \leq b$ \label{algo:full:consensus-pre} \;
      $\KwMsgM[\id] \leftarrow c$;
      $\KwMsgCcl[\id] \leftarrow D$ \;
      $\KwMsgR[\id] \leftarrow b$;
      $\KwMsgW[\id] \leftarrow b$ \;
      \Send $\KwMConsA(\id, b)$ \To $j$ \;
    }

    \SubAlgo{\RecA $\KwMConsA(\id, b)$ \RecB $Q$} {
      \Pre $\KwMsgR[\id] = b \land |Q| = f + 1$ \label{algo:full:consensus-ack-pre} \;
      \Send $\KwMCom(\id, \KwMsgM[\id], \KwMsgCcl[\id])$ \ToAll \label{algo:full:consensus-end}\;
    }  

    \SubAlgo{\RecA $\KwMCom(\id, c, D)$} {
      \Pre $\id \not \in \KwPhDset \union \KwPhEset$\;
      $\KwMsgM[\id] \leftarrow c$;
      $\KwMsgCcl[\id] \leftarrow D$;
      $\KwMsgPh[\id] \leftarrow \KwPhD$ \;
    }

    \SubAlgo{\Fun $\KwRec(\id)$\label{algo:full:rec}}{
      $b \leftarrow i + n (\floor{\frac{\KwMsgR[\id]}{n}} + 1)$ \label{algo:full:my-ballot} \;
      \Send $\KwMR (\id, \KwMsgM[\id], b)$ \ToAll \;
    }

    \SubAlgo{\RecA $\KwMR(\id, \_, \_)$ \RecB $j$}{
      \Pre $\id \in \KwPhDset \union \KwPhEset$\label{algo:full:rec-committed} \;
      \Send $\KwMCom(\id, \KwMsgM[\id], \KwMsgCcl[\id])$ \To $j$ \label{algo:full:rec-send-commit}
    }

    \SubAlgo{\RecA $\KwMR(\id, c, b)$ \RecB $j$}{
      \Pre $\KwMsgR[\id] < b \land \id \not \in \KwPhDset \union \KwPhEset$ \label{algo:full:rec-pre} \;
      \If{$\KwMsgR[\id] = 0 \land \id \in \KwPhAset$\label{algo:full:rec-first-see}}{
        $\KwMsgCcl[\id] \leftarrow \conflicts(c)$ \label{algo:full:conflicts-rec} \;
        $\KwMsgM[\id] \leftarrow c$
      }
      $\KwMsgR[\id] \leftarrow b$ \label{algo:full:recovery-ballot} \;
      $\KwMsgPh[\id] \leftarrow \KwPhF$ \label{algo:full:recovery-phase}\;
      \Send $\KwMRA(\id, \KwMsgM[\id], \KwMsgCcl[\id], \KwMsgQ[\id],$
      \nonl\hspace*{10.1em}
      $\KwMsgW[\id], b)$ \To $j$ \label{algo:full:rec-ack}
  }

    \SubAlgo{\mbox{\RecA $\KwMRA(\id, \msg_j, \dep_j, Q^0_j, ab_j, b)$ \RecBAll
      $j \in Q$}\label{algo:full:mrec-ack-pre}}{ 
      \Pre $\KwMsgR[\id] = b \land |Q| = n - f$\;
      \If{$\exists k \in Q.\, ab_k \not = 0$\label{algo:full:previous-proposal}}{
        \Let $k$ be such that $ab_k$ is maximal \label{algo:full:highest-proposal}\;
        \Send $\KwMCons(\id, \msg_k, \dep_k, b)$ \ToAll \label{algo:full:highest-proposal-cons}
      }
    \ElseIf{$\exists k \in Q.\, Q^0_k \not= \emptyset$\label{algo:full:qk-check}}{
      $Q' \leftarrow$ \leIf{$id.1 \in Q$\label{algo:full:coordinator-check}}{$Q$}{$Q \cap Q^0_k$}
      \Send $\KwMCons(\id, \msg_k, \bigcup_{Q'} \dep, b)$ \ToAll \label{algo:full:rec-union}
    }
    \lElse{\Send $\KwMCons(\id, \KwNoop, \emptyset, b)$ \ToAll \label{algo:full:noop}}
    }

    \SubAlgo{\Loop}{
      \Let $S$ be the smallest subset of $\KwPhDset$ such that\;
      {
        \nonl
        \hspace{0.65em}
          $\forall \id \in S.\, (\KwMsgCcl[\id] \subseteq S \cup \KwPhEset)$
          \label{algo:full:exec-start} \;
      }
      \For{$\id \in S$ \OrderedBy $<$\label{algo:full:exec-sort}}{
        $\pcdeliver(\KwMsgM[\id])$ \;
        $\KwMsgPh[\id] \leftarrow \KwPhE$
        \label{algo:full:exec-end}
      }
    }

    \algrule[0.8pt]
    \vspace{0.3em}
    \begin{center}
        \begin{tabular}{r@{}c@{}l@{}ll}
          $\KwMsgM[\id]$ & ${} \assign {}$ & $\KwNoop$ & ${} \in \CMD$ & Command
          \\
          $\KwMsgPh[\id]$ & ${} \assign {}$ & $\KwPhA$ & & Phase
          \\
          $\KwMsgCcl[\id]$ & ${} \assign {}$ & $\emptySet$ & ${} \subseteq \ID$ & Dependency set
          \\
          $\KwMsgQ[\id]$ & ${} \assign {}$ & $\emptySet$ & ${} \subseteq \Proc$ & Fast quorum
          \\
          $\KwMsgR[\id]$ & ${} \assign {}$ & $0$ & ${} \in \mathds{N}$ & Current ballot
          \\
          $\KwMsgW[\id]$ & ${} \assign {}$ & $0$ & ${} \in \mathds{N}$ & Last accepted ballot
          \\
        \end{tabular}
    \end{center}

\end{multicols}
  \caption{Full \SYS protocol at process $i$.}
  \label{algo:full}
  \end{algorithm}
\end{figure*}

%% file: app-proof.tex
\subsection{Protocol Correctness}
\label{app:correctness}

  In what follows, we prove Invariants~\ref{inv:cons} and \ref{inv:conf-a}.
  To this end, we use the auxiliary invariants below:
\begin{enumerate}
  \setcounter{enumi}{4}
\item \label{inv:validity} At any process, if $\KwMsgM[\id] \not= \KwNoop$, then
  $\KwMsgM[\id]$ has been previously submitted by a client.

\item \label{inv:c-or-noop} Assume $\KwMCol(\id, c, \_, \_)$ has been sent.
  Then for any $\KwMCons(\id, c', \_, \_)$,
    $\KwMCom(\id, c', \_)$
    and $\KwMR(\id, c', \_)$, we have
    $c' = c$ or $c' = \KwNoop$.

\item \label{inv:consensus-ballot}
  Assume $\KwMCons(\id, \_, \_, b)$ has been sent.  Then $b = \id.1$ or $b > n$.

\item \label{inv:unique-consensus}
  Assume $\KwMCons(\id, c, D, b)$ and $\KwMCons(\id, c', D', b')$ have been sent.
  If $b = b'$, then $c = c'$ and $D = D'$.

\item \label{inv:unique-initial-quorum}
  Assume $\KwMRA(\id, c, \_, Q, ab, \_)$ and $\KwMRA(\id, c', \_, Q', ab', \_)$ have been sent.
    If $Q \not= \varnothing$ and $Q' \not= \varnothing$, then $Q = Q'$.
    If additionally $ab = ab' = 0$, then $c = c'$.

\item \label{inv:mrecack-ballots}
  Assume $\KwMRA(\_, \_, \_, \_,  ab, b)$ has been sent by some process.
  Then $ab < b$.

\item \label{inv:after-leader}
  Assume $\KwMConsA(\id, b)$ and $\KwMRA(\id, \_, \_, \_,  ab, b')$ have been sent by some process.
  If $b' > b$, then $b \le ab < b'$ and $ab \not = 0$.

\item \label{inv:main-consensus}
  Assume a slow quorum has received $\KwMCons(\id, c, D, b)$ and responded to it with $\KwMConsA(\id, b)$.
  For any $\KwMCons(\id, c', D', b')$ sent, if $b' > b$, then $c' = c$ and $D' = D$.

\item \label{inv:main-fastpath}
  Assume $\KwMCom(\id, c, D)$ has been sent at line~\ref{algo:full:fp}.
    Then for any $\KwMCons(\id, c', D', \_)$ sent, $c' = c$ and $D' = D$.

\item \label{inv:conflict-ordering-single-process}
  Assume $\KwMCom(\id, c, \_)$ and $\KwMCom(\id', c', \_)$ have been sent,
  $c \not= \KwNoop$, $c' \not= \KwNoop$ and $\conflict(c, c')$.
  Assume further that some process sends two messages:
  either $\KwMColA(\id, \dep)$ or $\KwMRA(\id, \_, \dep, \_, 0, \_)$
    and either $\KwMColA(\id', \dep')$ or $\KwMRA(\id', \_, \dep', \_, 0, \_)$.
    Then $\id' \in \dep$ or $\id \in \dep'$.

\end{enumerate}

Invariants~\ref{inv:validity}-\ref{inv:after-leader} easily follow from the
structure of the protocol.
We now prove the rest of the invariants.

\paragraph{Proof of Invariant~\ref{inv:main-consensus}.}
Assume that at some point
\begin{quote}
  (*) a slow quorum has received $\KwMCons(\id, c, D, b)$ and responded to it with $\KwMConsA(\id, b)$.
\end{quote}
We prove by induction on $b'$ that,
if a process $i$ sends $\KwMCons(\id, c', D', b')$ with $b' > b$,
then $c' = c$ and $D' = D$.
Given some $b^*$, assume this property holds for all $b' < b^*$.
We now show that it holds for $b' = b^*$.
We make a case split depending on the transition of process $i$ that sends the
$\KwMCons$ message.

First, assume that process $i$ sends $\KwMCons$ at line~\ref{algo:full:sp}.
In this case, $b' = i$.
Since $b' > b$, we have $b < i$.
But this contradicts Invariant~\ref{inv:consensus-ballot}.
Hence, this case is impossible.

The remaining case is when process $i$ sends $\KwMCons$ during the transition at line~\ref{algo:full:mrec-ack-pre}.
In this case, $i$ has received
$$
\KwMRA(\id, \msg_j, \conf_j, \_, ab_j, b')
$$
from all processes $j$ in a recovery quorum $Q^R$.
Let $\abmax = \mathsf{max}\{ab_j \mid j \in Q^R\}$; then by Invariant~\ref{inv:mrecack-ballots} we have $\abmax < b'$.

Since the recovery quorum $Q^R$ has size $n - f$ and the slow quorum from (*) has size $f + 1$,
we get that at least one process in $Q^R$ must have received the $\KwMCons(\id, c, D, b)$ message and responded to it with $\KwMConsA(\id, b)$.
Let one of these processes be $p$.
Since $b' > b$,
by Invariant~\ref{inv:after-leader} we have $ab_p \not = 0$, and thus
process $i$ executes line~\ref{algo:full:highest-proposal-cons}.
By Invariant~\ref{inv:after-leader} we also have $b \le ab_p$
and thus $b \le \abmax$.

Consider an arbitrary process $k \in Q^R$, selected at line~\ref{algo:full:highest-proposal},
such that $ab_k = \abmax$. We now prove that $\msg_k = c$ and $\conf_k = D$.
If $\abmax > b$, then since $\abmax < b'$,
by induction hypothesis we have $\msg_k = c$ and $\conf_k = D$, as required.
If $\abmax = b$, then
since $\abmax \not = 0$, process $k$ has received some
$\KwMCons(\id, \_, \_, \abmax)$ message.
By Invariant~\ref{inv:unique-consensus},
process $k$ must have received the same $\KwMCons(\id, c, D, \abmax)$
received by process $p$.
Upon receiving this message, process $k$ stores $c$ in $\KwMsgM$ and $D$ in $\KwMsgCcl$
and does not change these values at line~\ref{algo:full:rec-first-see}:
$\abmax \not = 0$ and thus $\bal[id]$ cannot be $0$ when the process executes this line.
Then process $k$ must have sent $\KwMRA(\id, \msg_k, \conf_k, \_, \abmax, b')$
with $\msg_k = c$ and $\conf_k = D$, which concludes the proof.\qed

\paragraph{Proof of Invariant~\ref{inv:main-fastpath}.}
Assume $\KwMCom(\id, c, D)$ has been sent at line~\ref{algo:full:fp}.
Then, the process that sent this $\KwMCom$ message must be process $\id.1$.
Moreover, we have that for some fast quorum $Q^F$ such that $\id.1 \in Q^F$:
\begin{quote}
  (*) every process $j \in Q^F$ has received $\KwMCol(\id, c, Q^F, \past)$ and responded with $\KwMColA(\id, \conf_j)$ such that $D = \quorumcup{f}{Q^F} \conf = \bigcup_{Q^F} \conf$.
\end{quote}
  We prove by induction on $b$ that, if a process $i$ sends $\KwMCons(\id, c', D', b)$,
  then $c' = c$ and $D' = D$.
Given some $b^*$, assume this property holds for all $b < b^*$.
We now show that it holds for $b = b^*$.

First note that process $i$ cannot send $\KwMCons$ at line~\ref{algo:full:sp},
since in this case we would have $i = \id.1$,
and $\id.1$ took the fast path at line~\ref{algo:full:fp}.
Hence, process $i$ must have sent $\KwMCons$ during the transition at line~\ref{algo:full:mrec-ack-pre}.
In this case, $i$ has received
$$
\KwMRA(\id, \msg_j, \conf_j, Q^0_j, ab_j, b)
$$
from all processes $j$ in a recovery quorum $Q^R$.

If $\KwMCons$ is sent at line~\ref{algo:full:highest-proposal-cons},
then we have $ab_k > 0$ for the process $k \in Q^R$ selected at line~\ref{algo:full:highest-proposal}.
In this case, before sending $\KwMRA$, process $k$ must have received
$$
\KwMCons(\id, \msg_k, \conf_k, ab_k)
$$
with $ab_k < b$.
Then by induction hypothesis we have $c' = \msg_k = c$ and $D' = \conf_k = D$.
This establishes the required.

If $\KwMCons$ is not sent in line~\ref{algo:full:highest-proposal-cons},
then we have $ab_k = 0$ for all processes $k \in Q^R$.
In this case, process $i$ sends $\KwMCons$ in either
line~\ref{algo:full:rec-union} or line~\ref{algo:full:noop}.
Since the recovery quorum $Q^R$ has size $n - f$ and the fast quorum $Q^F$ from (*) has size $\floor{\frac{n}{2}} + f$, we have that
\begin{quote}
  (**) at least $\floor{\frac{n}{2}}$ processes in $Q^R$ are part of $Q^F$ and thus must have received $\KwMCol(\id, c, Q^F, \past)$ and responded to it with $\KwMColA$.
\end{quote}

Let process $p$ be one these processes.
Due to the assignment at line~\ref{algo:full:recovery-phase} and the check at line~\ref{algo:full:collect-pre},
process $p$ must have received $\KwMCol$ before sending $\KwMRA$.
Then, since $ab_p = 0$,
process $p$ reports the initial fast quorum $Q^F$ and command $c$,
i.e., process $p$ sends $\KwMRA(\id, cmd_p, \_, Q^0_p, ab_p, \_)$ with
$Q^0_p = Q^F$ and $cmd_p = c$.
Then $Q^0_p \not= \emptyset$, so that
process $i$ must send $\KwMCons$ at line~\ref{algo:full:rec-union}.

By Invariant~\ref{inv:unique-initial-quorum},
and since process $p$ has sent $\KwMRA(\id, c, \_, Q^F, \_, \_)$,
any process $k$ selected in line~\ref{algo:full:qk-check} has $Q^0_k = Q^F$ and $cmd_k = c$.
For this reason, $c' = cmd_k = c$, as required.
We now show that $D' = D$.
By our assumption, process $\id.1$ sent an $\KwMCom(\id, c, D)$ at line~\ref{algo:full:fp}.
Then due to line~\ref{algo:full:rec-committed},
this process would reply to $\KwMR$ with 
$\KwMCom$ instead of $\KwMRA$.
Hence, $\id.1$ is not part of the recovery quorum, i.e., $\id.1 \not \in Q^R$,
and with that,
quorum $Q^R \cap Q^F$ is selected in line~\ref{algo:full:coordinator-check}.
Let this quorum be $Q^U$.
By Property~\ref{property:set-union}, the fast path proposal
$D = \bigcup_{Q^F} \conf$
can be recovered by the set union of the dependencies initially reported by any
$\floor{\frac{n}{2}}$ fast quorum members (excluding the initial coordinator).
By (**), and since all processes $k \in Q^U$ have $ab_k = 0$,
then all processes in $Q^U$ replied with the dependencies that were reported to the initial coordinator.
Thus, by Property~\ref{property:set-union} we have
$D = \bigcup_{Q^F} \conf = \bigcup_{Q^U} \conf = D'$,
which concludes the proof.\qed

\paragraph{Proof of Invariant~\ref{inv:cons}.}
Consider that $\KwMCom(\id, c, D)$ and $\KwMCom(\id, c', D')$ have been sent.
We prove that $c = c'$ and $D = D'$.

Note that, if an $\KwMCom(\id, c, D)$ was sent at line~\ref{algo:full:rec-send-commit},
then some process sent an $\KwMCom(\id, c, D)$ at line~\ref{algo:full:fp} or line~\ref{algo:full:consensus-end}.
Hence, without loss of generality,
we can assume that the two $\KwMCom$ under consideration were sent at line~\ref{algo:full:fp} or at line~\ref{algo:full:consensus-end}.
We can also assume that the two $\KwMCom$ have been sent by different processes.
Only one process can send an $\KwMCom$ at line~\ref{algo:full:fp} and only once.
Hence, it is sufficient to only consider the following two cases.

Assume first that both $\KwMCom$ messages are sent at line~\ref{algo:full:consensus-end}.
Then for some $b$, a slow quorum has received $\KwMCons(\id, c, D, b)$ and responded to it with $\KwMConsA(\id, b)$.
Likewise, for some $b'$, a slow quorum has received $\KwMCons(\id, c', D', b')$ and responded to it with $\KwMConsA(\id, b')$.

Assume without loss of generality that $b \leq b'$.
If $b < b'$, then $c' = c$ and $D' = D$ by Invariant~\ref{inv:main-consensus}.
If $b = b'$, then $c' = c$ and $D' = D$ by Invariant~\ref{inv:unique-consensus}.
Hence, in this case $c' = c$ and $D' = D$, as required.

Assume now that $\KwMCom(\id, c, D)$ was sent at line~\ref{algo:full:fp} and $\KwMCom(\id, c', D')$ at line~\ref{algo:full:consensus-end}.
Then for some $b$, a slow quorum has received $\KwMCons(\id, c', D', b)$ and responded to it with $\KwMConsA(\id, b)$.
Then by Invariant~\ref{inv:main-fastpath}, we must have $c' = c$ and $D' = D$,
as required.\qed

\paragraph{Proof of Invariant~\ref{inv:conflict-ordering-single-process}.}
Assume $\KwMCom(\id, c, \_)$ and $\KwMCom(\id', c', \_)$ have been sent,
$c \not= \KwNoop$, $c' \not= \KwNoop$ and $\conflict(c, c')$.
Assume further that process $j$ sends two messages:
either $\KwMColA(\id, \dep)$ or $\KwMRA(\id, \_, \dep, \_, 0, \_)$
  and either $\KwMColA(\id', \dep')$ or $\KwMRA(\id', \_, \dep', \_, 0, \_)$.
If $\KwMColA(\id, \dep)$ is sent, it must be in response to $\KwMCol(\id, d, \_, \_)$, and by Invariant~\ref{inv:c-or-noop} we have $d = c$.
Similarly, if $\KwMColA(\id', \dep')$ is sent, it must be in response to $\KwMCol(\id', d', \_, \_)$, and by Invariant~\ref{inv:c-or-noop} we have $d' = c'$.
If $\KwMRA(\id, \_, \dep, \_, 0, \_)$ is sent, it must be in response to $\KwMR(\id, d, \_)$, and by Invariant~\ref{inv:c-or-noop} we have $d \in \{c, \KwNoop\}$.
If $\KwMRA(\id', \_, \dep', \_, 0, \_)$ is sent, it must be in response to $\KwMR(\id', d', \_)$, and by Invariant~\ref{inv:c-or-noop} we have $d' \in \{c', \KwNoop\}$.

Without loss of generality, assume that process $j$ sends the message about $\id$ before the message about $\id'$.
We prove that $\id \in \dep'$.
We have four cases depending on which message ($\KwMColA$ or $\KwMRA$)
is sent for each identifier:

\emph{1) Process $j$ sends $\KwMColA(\id, \dep)$ and then $\KwMColA(\id', \dep')$.}
When handling $\KwMCol(\id, c, \_, \_)$,
process $j$ stores $c$ in $\KwMsgM[\id]$.
By Invariant~\ref{inv:c-or-noop}, $\KwMsgM[\id]$ can only change to $\KwNoop$.
When handling $\KwMCol(\id', c', \_, \_)$,
since $\KwMsgM[\id] \in \{c, \KwNoop\}$ and $\KwNoop$ conflicts with all commands,
we have
$\id \in \conflicts(c')$ in line~\ref{algo:full:conflicts},
and thus $\id \in \dep'$ in $\KwMColA(\id', \dep')$, as required.

\emph{2) Process $j$ sends $\KwMColA(\id, \dep)$ and then $\KwMRA(\id', \_, \dep', \_, 0, \_)$.}
When handling $\KwMCol(\id, c, \_, \_)$,
process $j$ stores $c$ in $\KwMsgM[\id]$.
By Invariant~\ref{inv:c-or-noop}, $\KwMsgM[\id]$ can only change to $\KwNoop$.
When handling $\KwMR(\id', d', \_)$ with $d' \in \{c', \KwNoop\}$
we have two cases depending on $\KwMsgPh[\id']$.
If $\id' \in \KwPhAset$, then 
since $\KwMsgM[\id] \in \{c, \KwNoop\}$ and $\KwNoop$ conflicts with all commands,
we have $\id \in \conflicts(d')$ in line~\ref{algo:full:conflicts-rec}.
If $\id' \not \in \KwPhAset$, then process $j$ is a member of the original fast
quorum for $\id'$ and thus included $\id$ into $\KwMsgCl[\id']$ when it processed
$\KwMCol(\id', c', \_, \_)$.
Thus, in both cases $\id \in \dep'$ in $\KwMRA(\id', \_, \dep', \_, 0, \_)$, as required.

\emph{3) Process $j$ sends $\KwMRA(\id, \_, \dep, \_, 0, \_)$ and then $\KwMColA(\id', \dep')$.}
Analogous to the above.

\emph{4) Process $j$ sends $\KwMRA(\id, \_, \dep, \_, 0, \_)$ and then $\KwMRA(\id', \_, \dep', \_, 0, \_)$.}
Analogous to the above.
\qed

\paragraph{Proof of Invariant~\ref{inv:conf-a}.}
Assume that $\KwMCom(\id, c, D)$ and $\KwMCom(\id', c', D')$ have been sent
with $\id \not= \id'$, $c \not= \KwNoop$, $c' \not= \KwNoop$ and $\conflict(c, c')$.
The protocol structure ensures that $D = \bigcup_{Q} \dep$ for $Q$ and $\dep$
given as parameters of handlers at lines~\ref{algo:full:collect-ack-receive}
or~\ref{algo:full:mrec-ack-pre}, and the computation of $D$ occurs at
lines~\ref{algo:full:fp-proposal} or~\ref{algo:full:rec-union}.
We start by proving that there exists a quorum $\widehat{Q}$
with $|\widehat{Q}| \geq \floor{\frac{n}{2}} + 1$
and $\widehat{\dep}$ such that
$\bigcup_Q \dep = \bigcup_{\widehat{Q}} \widehat{\dep}$,
where each process $j \in \widehat{Q}$ computes its $\widehat{\dep}_j$ in either
line~\ref{algo:full:conflicts} or line~\ref{algo:full:conflicts-rec}
and sends it in either
$\KwMColA(\id, \widehat{\dep}_j)$ or $\KwMRA(\id, \_, \widehat{\dep}_j, \_, 0, \_)$.

The computation of $D$ occurs either in the transition at
line~\ref{algo:full:collect-ack-receive} or at line~\ref{algo:full:rec-union}.
If the computation of $D$ occurs in the transition at
line~\ref{algo:full:collect-ack-receive},
then $Q$ is a fast quorum with size $\floor{\frac{n}{2}} + f$.
In this case, we let $\widehat{Q} = Q$ and $\widehat{\dep} = dep$.
Since $f \geq 1$, we have $|\widehat{Q}| \geq \floor{\frac{n}{2}} + 1$, as required.
If the computation of $D$ occurs at line~\ref{algo:full:rec-union},
we have two situations depending on whether $\id.1 \in Q$
(line~\ref{algo:full:coordinator-check}).
If $\id.1 \in Q$, then $Q$ is a recovery quorum of size $n - f$.
In this case, we let $\widehat{Q} = Q$ and $\widehat{\dep} = dep$.
Since $f \leq \floor{\frac{n - 1}{2}}$, 
we have $|\widehat{Q}| \geq \floor{\frac{n}{2}} + 1$, as required.
If $\id.1 \not\in Q$,
then $Q$ consists of the fast quorum members that are part of the recovery quorum
(line~\ref{algo:full:coordinator-check}).
Given that fast quorum size is $\floor{\frac{n}{2}} + f$ and the recovery quorum
size is $n - f$, in this case $Q$ contains at least $\floor{\frac{n}{2}} + f - f
= \floor{\frac{n}{2}}$ fast quorum processes, 
and thus $|Q| \geq \floor{\frac{n}{2}}$.
Since $D$ is computed in the branch where 
the initial fast quorum is known (line~\ref{algo:full:qk-check}),
at least one of the fast quorum members in $Q$ must have computed its
set of dependencies at line~\ref{algo:full:conflicts}, 
including in its dependencies those reported by the original coordinator.
In this case, we let $\widehat{Q} = Q \union \{\id.1\}$,
$\forall j \in Q.\, \widehat{\dep}_j = \dep_j$ and
$\widehat{\dep}_{\id.1}$ be the set of dependencies sent by $\id.1$ in its
$\KwMColA(\id, \widehat{\dep}_{\id.1})$ message.
Since $|Q| \geq \floor{\frac{n}{2}}$ and $\id.1 \not \in Q$,
we have $|\widehat{Q}| \geq \floor{\frac{n}{2}} + 1$, as required.

Similarly to the above, we can also prove that there exists a quorum $\widehat{Q}'$
with $|\widehat{Q}'| \geq \floor{\frac{n}{2}} + 1$
and $\widehat{\dep'}$ such that
$\bigcup_{Q'} \dep' = \bigcup_{\widehat{Q}'} \widehat{\dep'}$,
where each process $j \in \widehat{Q}'$ computes its $\widehat{\dep'_j}$ in either
line~\ref{algo:full:conflicts} or line~\ref{algo:full:conflicts-rec} and sends
its $\widehat{\dep'_j}$ in either $\KwMColA(\id', \widehat{\dep'_j})$ or
$\KwMRA(\id', \_, \widehat{\dep'_j}, \_, 0, \_)$.

We now prove that $id' \in D$ or $id \in D'$. By contradiction,
assume that $\id' \not \in D$ and $\id \not \in D'$.  Since $id' \not \in D$, we
have $\forall j \in \widehat{Q}.\, \id' \not \in \widehat{\dep}_j$.  Similarly, since
$id \not \in D'$, we have $\forall j \in \widehat{Q}'.\, \id \not \in \widehat{\dep'_j}$.
Given that $|\widehat{Q}| \geq \floor{\frac{n}{2}} + 1$ and
$|\widehat{Q}'| \geq \floor{\frac{n}{2}} + 1$, $\widehat{Q}$ and $\widehat{Q}'$ must
intersect.  For this reason, there must exist a process
$p \in \widehat{Q} \cap \widehat{Q}'$ such that $id' \not \in \widehat{\dep}_p$ and
$id \not \in \widehat{\dep'_p}$.  But this contradicts
Invariant~\ref{inv:conflict-ordering-single-process}.
\qed

\subsubsection{Slow-path optimization}
\label{app:slow-path-opt}

Section~\refsection{execution:improvements} describes an optimization that reduces the number of the dependencies in the slow path by proposing to consensus $\quorumcup{f}{Q} \dep$ instead of $\bigcup_{Q} \dep$ at line~\ref{algo:full:sp}. The previous proofs are not affected by this optimization with the exception of the proof of Invariant~\ref{inv:conf-a}.  We now prove this invariant when the optimization is enabled.

\paragraph{Proof of Invariant~\ref{inv:conf-a}.}
Assume that $\KwMCom(\id, c, D)$ and $\KwMCom(\id', c', D')$ have been sent
with $\id \not= \id'$, $c \not= \KwNoop$, $c' \not= \KwNoop$ and $\conflict(c, c')$.
The protocol structure ensures that either $D = \bigcup_{Q} \dep$ or $D = \quorumcup{f}{Q} \dep$ for $Q$ and $\dep$
given as parameters of handlers at lines~\ref{algo:full:collect-ack-receive}
or~\ref{algo:full:mrec-ack-pre}.
Similarly, the protocol structure ensures that either $D' = \bigcup_{Q'} \dep'$ or $D' = \quorumcup{f}{Q'} \dep'$ for $Q'$ and $\dep'$
given as parameters of handlers at lines~\ref{algo:full:collect-ack-receive}
or~\ref{algo:full:mrec-ack-pre}.
The computation of $D = \bigcup_{Q} \dep$ and $D' = \bigcup_{Q'} \dep'$ occurs at
lines~\ref{algo:full:fp-proposal} or~\ref{algo:full:rec-union},
while the computation of $D = \quorumcup{f}{Q} \dep$ and $D' = \quorumcup{f}{Q'} \dep'$ occurs at line~\ref{algo:full:sp}.

Similarly to the previous proof of Invariant~\ref{inv:conf-a},
we can prove that if $D = \bigcup_Q \dep$, then there exists a quorum $\widehat{Q}$
with $|\widehat{Q}| \geq \floor{\frac{n}{2}} + 1$
and $\widehat{\dep}$ such that
$\bigcup_Q \dep = \bigcup_{\widehat{Q}} \widehat{\dep}$,
where each process $j \in \widehat{Q}$ computes its $\widehat{\dep}_j$ in either
line~\ref{algo:full:conflicts} or line~\ref{algo:full:conflicts-rec}
and sends it in either
$\KwMColA(\id, \widehat{\dep}_j)$ or $\KwMRA(\id, \_, \widehat{\dep}_j, \_, 0, \_)$.
Likewise,
we can prove prove that if $D' = \bigcup_{Q'} \dep'$, then there exists a quorum $\widehat{Q}'$
with $|\widehat{Q}'| \geq \floor{\frac{n}{2}} + 1$
and $\widehat{\dep'}$ such that
$\bigcup_{Q'} \dep' = \bigcup_{\widehat{Q}'} \widehat{\dep'}$,
where each process $j \in \widehat{Q}'$ computes its $\widehat{\dep'_j}$ in either
line~\ref{algo:full:conflicts} or line~\ref{algo:full:conflicts-rec} and sends
its $\widehat{\dep'_j}$ in either $\KwMColA(\id', \widehat{\dep'_j})$ or
$\KwMRA(\id', \_, \widehat{\dep'_j}, \_, 0, \_)$.

We now prove that $id' \in D$ or $id \in D'$. By contradiction,
assume that $\id' \not \in D$ and $\id \not \in D'$. 
We have four cases depending on the mechanism ($\bigcup$ or $\quorumcup{f}{}$\hspace{-0.25em}) used to compute $D$ and $D'$:

\emph{1) $D = \bigcup_Q \dep$ and $D' = \bigcup_{Q'} \dep'$}.
Analogous to previous proof of Invariant~\ref{inv:conf-a}.

\newcommand\plus[1]{{#1}^*}
\emph{2) $D = \bigcup_Q \dep$ and $D' = \quorumcup{f}{Q'} \dep'$}.
Since $id' \not \in D$, we
have $\forall j \in \widehat{Q}.\, \id' \not \in \widehat{\dep}_j$. 
Since $D'$ is computed in line~\ref{algo:full:sp}, we have that $|Q'| = \floor{\frac{n}{2}} + f$.
Moreover, since $id \not \in D'$, by the definition of $D' = \quorumcup{f}{Q'} \dep'$, we have:
\begin{align*}
  &\ |\{ j \in Q' \mid \id \in \dep'_j \}| < f & \\
  \Leftrightarrow &\ |\{ j \in Q' \mid \id \not \in \dep'_j \}| \geq \floor{\tfrac{n}{2}} + f - (f - 1)
  & (\text{since } |Q'| = \floor{\tfrac{n}{2}} + f) \\
  \Leftrightarrow &\ |\{ j \in Q' \mid \id \not \in \dep'_j \}| \geq \floor{\tfrac{n}{2}} + 1 & \\
  \Leftrightarrow &\ \exists \plus{Q} \subseteq Q'.\, |\plus{Q}| \geq \floor{\tfrac{n}{2}} + 1 \land \forall j \in \plus{Q}.\, \id \not \in \dep'_j &
\end{align*}

Given that
$|\widehat{Q}| \geq \floor{\frac{n}{2}} + 1$ and
$|\plus{Q}| \geq \floor{\frac{n}{2}} + 1$,
$\widehat{Q}$ and $\plus{Q}$ must intersect.
For this reason, there must exist a process
$p \in \widehat{Q} \cap \plus{Q}$ such that
$id' \not \in \widehat{\dep}_p$ and $id \not \in \dep'_p$.
But this contradicts
Invariant~\ref{inv:conflict-ordering-single-process}.

\emph{3) $D = \quorumcup{f}{Q} \dep$ and $D' = \bigcup_{Q'} \dep'$}.
Analogous to the above.

\emph{4) $D = \quorumcup{f}{Q} \dep$ and $D' = \quorumcup{f}{Q'} \dep'$}.
Analogous to the above.
\qed

%% file: app_smr.tex
\section{State-Machine Replication with \SYS}
\labappendix{smr}

State-machine replication (SMR) implements what is called in literature a \emph{universal construction}\footnote{Maurice Herlihy. Wait-free synchronization. \emph{ACM Trans. Program.  Lang. Syst}., 1991.}, that is a general mechanism to obtain a linearizable shared object from a sequential one.
In Appendix~\ref{app:correctness}, we proved that \SYS correctly implements the SMR protocol specification given in \refsection{smr}.
This section explains how to build a universal construction from this protocol.
To achieve this, we first introduce some preliminary notions.
Then, we explain how to implement any linearizable data type on top of the \SYS protocol.
The bottom of this section covers the \readopt optimization proposed in \refsection{execution:improvements}.

\subsection{Preliminaries}
\labsection{smr:preliminaries}

We base our reasoning and algorithms upon the notion of trace\footnote{\label{Rozenberg:1995}Volker Diekert and Grzegorz Rozenberg, editors. \emph{The Book of Traces}.  World Scientific, 1995}, that is a class of equivalent command words.
Two words in a class contain the same commands and sort non-commuting ones in the same order.
A trace can be seen as as special case of the notion of c-struct used to define the generalized consensus problem \cite{gpaxos}.

\paragraph{State machine.}
We assume a sequential object specified by the following components:
\begin{inparaenum}
\item a set of states $\sttSet$;
\item an initial state $\sttInit \in \sttSet$;
\item a set of commands $\cmdSet$ that can be performed on the object;
\item a set of their response values $\valSet$; and
\item a transition function $\tau :  \sttSet \times \cmdSet \rightarrow \sttSet \times \valSet$.
\end{inparaenum}
In the following, we use special symbols $\bot$ and $\top$ that do not belong to $\valSet$.
When applying a command, we use $.\stt$ and $.\val$ selectors to respectively extract the state and the response value, i.e., given a state $s$ and a command $c$, we let $\tau(s,c)=(\tau(s,c).\stt,\tau(s,c).\val)$.
Without lack of generality, we consider that commands are applicable to every state.
A command $c$ is a \emph{read} if it does not change the object state: $\forall s.\, \tau(s,c).\stt = s$; otherwise, $c$ is a \emph{write}.
We denote by $\Read$ and $\Write$ the set of read and write commands.

\paragraph{Command words.}
A \emph{command word} $x$ is a sequence of commands.
The empty word is denoted $1$ and $\cmdSet^*$ is the set of all command words.
We use the following notations for a word $x$:
$|x|$ is the length of $x$;
$x[i \geq 1]$ is the $i$-th element in $x$;
$|x|_{c}$ is the number of occurrences of command $c$ in $x$.
We write $c^i \in x$ when $c$ occurs at least $i>0$ times in $x$.
$\pos(c^i,x)$ is the position of the $i$-th occurrence of command $c$ in $x$, with $\pos(c^i,x)=0$ when $c^i \notin x$.
The shorthand $c^i <_x d^j$ stands for $\pos(c^i,x) < \pos(d^j,x)$.
The set $\cmdOf{x}$ is defined as $\{(c,i) : c^i \in x\}$.
The operator $x \setminus c$ deletes the last occurrence of $c$ in $x$ (if such an occurrence exists).
By extension, for some word $y$, $x \setminus y$ applies $x \setminus c$ for every $(c,i) \in \cmdOf{y}$.
We let $\prefix$ be the prefix relation induced by the append operator over $\cmdSet^*$.
The prefix of $x$ up to some occurrence $c^i$ is the command word $x\upto{c^i}$.
If $c^i \notin x$, then by convention $x\upto{c^i}$ equals $1$.
In case $c$ appears once in $x$, $x\upto{c}$ is a shorthand for $x\upto{c^1}$.

\begin{lemma}
  \lablem{prel:1}
  Consider a command $c$ and two words $x$ and $y$.
  Then, $|xy|_{c}$ equals $|x|_{c} + |y|_{c}$.
  Moreover, if $c^k \in xy$ then $\pos(c^k, xy)$ equals $\pos(c^k, x)$, if $k \leq |x|_c$ and $|x| + pos(c^{k - |x|_c}, y)$ otherwise.
\end{lemma}

\begin{proof}
  Follows from the definitions.
\end{proof}

\paragraph{Equivalence of command words.}
We define function $\tau^{*}$ by the repeated application of $\tau$.
In detail, for a state $s$ we define $\tau^*(s,1) = (s,\nil)$, for some symbol $\nil \in \valSet$, and if $x$ is non-empty then we have:
\begin{displaymath}
  \tau^{*}(s,x)
  =
  \left\lbrace
  \begin{array}{ll}
    \tau(s,x[1]), & \text{if $|x|=1$}; \\
    \tau^{*}(\tau(s,x[1]).\stt, x[2] \ldots x[n]), & \text{otherwise}.
  \end{array}
  \right.
\end{displaymath}
Two commands $c$ and $d$ \emph{commute}, written $c \commuting d$, if in every state $s$ we have:
\begin{displaymath}
    \begin{array}{l}
      \tau^{*}(s,cd).\stt = \tau^{*}(s,dc).\stt; \\
      \tau^{*}(s,dc).\val = \tau^{*}(s,c).\val; \\
      \tau^{*}(s,cd).\val  = \tau^{*}(s,d).\val.
    \end{array}
\end{displaymath}
Relation $\commuting$ is an equivalence relation over $\cmdSet$.
We write $c \nonCommuting d$ the fact that $c$ and $d$ do no commute.
Two words $x,y \in \cmdSet^*$ are \emph{equivalent}, written $x \equiv y$, when there exist words $z_1, \ldots, z_{k \geq 1}$ such that $z_1=x$, $z_k=y$ and for all $i$, $1 \leq i < k$, there exist words $z'$, $z''$ and commands $c \commuting d$ satisfying $z_i = z'cdz'', z_{i+1} = z'dcz''$.
This means that a word can be obtained from another by successive transpositions of neighboring commuting commands.
One may show that $u \equiv v$ holds when $u$ and $v$ contain the same commands and order non-commuting ones the same way.
In such a case, commands have the same effects.

\begin{lemma}
  \lablem{prel:2}
  Relation $x \equiv y$ holds iff $\cmdOf{x}=\cmdOf{y}$ and for any $c \nonCommuting d$, $c^i <_x d^j \iff c^i <_y d^j$. \footnote{Volker Diekert and Yves M{\'{e}}tivier. Partial Commutation and Traces. In \emph{Handbook of Formal Languages, Volume 3: Beyond Words}. 1997.}
\end{lemma}

\begin{lemma}
  \lablem{prel:3}
  If $x \equiv y$ then for every command $c$, $\tau^*(\sttInit, x\upto{c^i}).\val =  \tau^*(\sttInit, y\upto{c^i}).\val$.
\end{lemma}

\begin{proof}
  We show that the proposition holds if $x=z'abz''$ and $y=z'baz''$, for $a \commuting b$ and words $z'$ and $z''$.
  Obviously, this is true for any command $c$ in $z'$.
  Now, if $a=c^i$, then the proposition holds by definition of relation $\commuting$.
  A symmetric argument holds for $b=c^i$.
  Then, because $a$ and $b$ are commuting, we may observe that $\tau^{*}(\sttInit,z'ab).\stt = \tau^{*}(\sttInit,z'ba).\stt$.
  From which, we deduce that the result also holds if $c^i \in z''$.
  Now, applying the above claim to the definition of $x \equiv y$, we deduce that the proposition holds in the general case.
\end{proof}   

\paragraph{Command traces.}
The equivalence class of $x$ for the relation $\equiv$ is denoted $[x]$.
This is the set of words that order non-commuting commands in the same way as $x$.
Hereafter, we note $\traceSet$ the quotient set of $\cmdSet^*$ by relation $\equiv$.
An element in $\traceSet$ is named a \emph{command trace}.
For any $x,y,z \in \cmdSet^*$, it is easy to observe that if $x \equiv y$ holds, then both $(zx \equiv zy)$ and $(xz \equiv yz)$ are true.
As a consequence, $\equiv$ is a congruence relation over $\cmdSet^*$.
It follows that $\traceSet$ together with the append operator defined as $[x][y]=[xy]$ forms a monoid\footnote{Gerard Lallement. \emph{Semigroups and Combinatorial Applications}. John Wiley \& Sons, Inc., 1979.}.
Now, consider the natural ordering induced by the append operator on $\traceSet$.
In other words, $[x] \prefix [y]$ holds iff $[x][z]=[y]$ for some $[z]$.
One can show that relation $\prefix$ is a partial order over $\traceSet$\textsuperscript{\ref{Rozenberg:1995}}.

\begin{lemma}
  \lablem{prel:4}
  If $\classOf{x} \prefix \classOf{y}$, then $\classOf{x} \classOf{y \setminus x} = \classOf{y}$.
\end{lemma}

\begin{proof}
  From $\classOf{x} \prefix \classOf{y}$, there exists some $z$ such that $\classOf{x} \classOf{z} = \classOf{y}$.
  We show that $\classOf{y \setminus x} = \classOf{z}$.
  If $c^i \in y$ and $c^i \notin x$, by \reflem{prel:2}, $c^i \in z$.
  Conversely, if $c^i \in z$ then $c^i \notin x$ and by \reflem{prel:2}, $c^i \in y$.
  Then, by applying again \reflem{prel:2}, we deduce that $c^i <_z d^j \iff  c^i <_{y \setminus x} d^j$.
\end{proof}

\begin{lemma}
  \lablem{prel:5}
  If $\cmdOf{x} \subseteq \cmdOf{y}$ and for any $c \nonCommuting d$, $c^i <_{y} d^j \land d^j \in x \implies c^i <_{x} d^j$, then $\classOf{x} \prefix \classOf{y}$.
\end{lemma}

\begin{proof}
  By \reflem{prel:1}, $\cmdOf{x(y \setminus x)} = \cmdOf{y}$. 
  Then, choose $c, d \in \cmdSet$ with $c \nonCommuting d$ and $c^i <_y d^j$.
  We show that $c^i <_{x(y \setminus x)} d^j$.
  Let $k=|x|_{c}$ and $l=|x|_d$.
  (Case $l=j$)
  By assumption.
  (Otherwise)
  If $k=i$ then $c^i \in x$ and $d^{j-l} \in (y \setminus x)$.
  In the converse case, $c^{i-k}$ and $d^{j-l}$ are both in $(y \setminus x)$.
  We then conclude by applying \reflem{prel:1}.
\end{proof}

\begin{lemma}
  \lablem{prel:6}
  If $\classOf{x} \prefix \classOf{y}$, then for every command $c$ with $c^i \in x$, $\tau^*(\sttInit, x\upto{c^i}).\val =  \tau^*(\sttInit, y\upto{c^i}).\val$.
\end{lemma}

\begin{proof}
  From \reflem{prel:4}, $x(y \setminus x) \equiv y$.
  Choose $c^i \in x$.
  By \reflem{prel:3}, $\tau^*(\sttInit, x(y \setminus x)\upto{c^i}).\val =  \tau^*(\sttInit, y\upto{c^i}).\val$.
  Since $c^i \in x$, $c^i \notin (y \setminus x)$ and $x(y \setminus x)\upto{c^i} = x\upto{c^i}$.
\end{proof}

\paragraph{Histories.}
A \emph{history} is a sequence of \emph{events} of the form $\invocation{i}{c}$ or $\response{i}{c}{v}$, where $i \in \procSet$, $c \in \cmdSet$ and $v \in \valSet$.
The two kinds of events denote respectively an \emph{invocation} of command $c$ by process $i$, and a \emph{response} to this command returning some value $v$.
We write $c \hb_h d$ the fact that the response of $c$ precedes the invocation of command $d$ in history $h$.
For a process $i$, we let $h|i$ be the projection of history $h$ onto the events by $i$.
The following classes of histories are of particular interest:
\begin{compactitem}
\item[--] %
  A history $h$ is \emph{sequential} if it is a  non-interleaved sequence of invocations and matching responses, possibly terminated by a non-returning invocation.
\item[--] %
  A history $h$ is \emph{well-formed} if
  \begin{inparaenum}
  \item $h|i$ is sequential for every $i \in \procSet$;
  \item each command $c$ is invoked at most once in $h$; and
  \item for every response $\response{i}{c}{v}$, an invocation $\invocation{i}{c}$ occurs before in $h$.
  \end{inparaenum}
\item[--] %
  A well-formed history $h$ is \emph{complete} if every invocation has a matching response.
  We shall write $\completeOf{h}$ the largest complete prefix of $h$.
\item[--] %
  A well-formed history $h$ is \emph{legal} if $h$ is complete and sequential and for any command $c$, if a response value appears in $h$, then it equals $\tau^{*}(\sttInit, h\upto{c}).\val$.
\end{compactitem}

\paragraph{Linearizability.}
Two histories $h$ and $h'$ are \emph{equivalent}, written $h \equiv h'$, if they contain the same set of events.
History $h$ is \emph{linearizable} \cite{linearizability} when it can be extended (by appending zero or more responses) into some history $h'$ such that $\completeOf{h'}$ is equivalent to a legal and sequential history $l$ preserving the real-time order in $h$, i.e., $\hb_{h} \subseteq \hb_{l}$.

\subsection{Algorithm}
\labsection{smr:alg}

\begin{algorithm*}[t]
  \setcounter{AlgoLine}{0}
  
  \caption{
    \labalg{smr}
    SMR with \SYS -- code at process $i$
  }

  \nonl\SubAlgo{\text{\bf Variables:}}{
    $B$  \tcp*[f]{An instance of \SYS with $c \nonCommuting d \implies \conflict(c,d)$} \labline{smr:var:1} \\
    $S \assign \sttInit$ \tcp*[f]{A local copy of the sequential object} \labline{smr:var:2} \\
    $\lambda \assign 1$ \tcp*[f]{A command word} \labline{smr:var:3} \\
    $\pending(c) \assign \bot, \forall c \in \cmdSet$ \tcp*[f]{A map of the response values} \labline{smr:var:4} \\
  }

  \nonl\SubAlgo{\Fun $\invoke(c)$}{
    $\pending(c) \assign \top$ \labline{smr:inv:1} \\
    $B.\pcdbroadcast(c)$ \labline{smr:inv:2} \\ 
    \textbf{wait until} $\pending(c) \neq \top$ \labline{smr:inv:3} \\
    \Return $\pending(c)$ \labline{smr:inv:4}
  }
 
  \When{$B.\pcdeliver(c)$}{ \labline{smr:exec:0}
    $\lambda \assign \lambda \append c$; $(S,v) \assign \tau(S,c)$ \labline{smr:exec:1} \\
    \If{$\caller(c) = i$}{ \labline{smr:exec:2}
      $\pending(c) \assign v$ \labline{smr:exec:3}
    }
  }
  
\end{algorithm*}

\refalg{smr} presents the pseudo-code of our universal construction on top of \SYS.
Each line of this algorithm is atomic.
To execute a command $c$ on the shared object, a process executes $\invoke(c)$.
As usual, we shall assume that no two process invoke the same command.
\refalg{smr} employs the following four variables:
\begin{compactitem}
\item[--] %
  $B$ is an instance of \SYS with $\conflict$ set to the non-commutativity relation among commands ($\nonCommuting$).
\item[--] %
  $S$ is a local copy of the state of the sequential object under concern.
  Initially, it equals $\sttInit$.
\item[--] %
  Variable $\lambda$ stores the log of commands applied to the local copy.
\item[--] %
  Variable $\pending$ stores the response value of each command call.
  Initially, $\pending(c) = \bot$ holds for every command $c$.
\end{compactitem}

\paragraph{Internals.}
When a process $i$ invokes some command $c$, it sets $\pending(c)$ to $\top$ to signal that the command is invoked  (\refline{smr:inv:1}).
Then, $i$ submits $c$ to \SYS and awaits that $c$ is applied locally (\refline{smr:inv:3}) before returning its response value (\refline{smr:inv:4}).
Upon the delivery of a command $c$, $i$ appends $c$ to $\lambda$ then executes it.
In case $i$ is the caller of $c$, $\pending(c)$ is set to the response value (\refline{smr:exec:3}).

\subsection{Correctness}
\labsection{smr:correctness}

In what follows, $\run$ is a run of \refalg{smr} and $h$ the corresponding history.
For some variable $\mathit{var}$, we denote by $\mathit{var}_i$ the value of $\mathit{var}$ at process $i$.
The notation $\mathit{var}_i^{\run}$ refers to the value of $\mathit{var}_i$ at the end of the execution $\run$.
For starters, we prove that at any point in time a single occurrence of a command may appear in $\lambda_i$.

\begin{proposition}
  \labprop{smr:1}
  $\forall i \in \procSet.~\tlGlobally (\forall c \in \cmdSet.~|\lambda_i|_{c} \leq 1)$.
\end{proposition}

\begin{proof}
  (by induction)
  $\lambda_i$ is initially empty.
  Then, assume that process $i$ appends $c$ to $\lambda_i$ at \refline{smr:exec:1}.
  Command $c$ is thus executed by \SYS at \refline{smr:exec:0}.
  By the Integrity property of \SYS, this happens at most once.
  Hence, ($|\lambda_i|_{c} = 1$) is true from that point in time.
\end{proof}

The execution mechanism at \reflines{smr:exec:1}{smr:exec:3} applies in order the commands of $\lambda_i$ to update $S_i$.
Such an approach maintains the following two invariants:

\begin{proposition}
  \labprop{smr:2}
  $\forall i \in \procSet.~ \tlGlobally(S_i = \tau^*(\sttInit,{\lambda_i}).\stt)$.
\end{proposition}

\begin{proof}
  (by induction.)
  Initially $\lambda_i=1$, leading to $\tau^*(\sttInit,\lambda_i).\stt = \sttInit$.
  This coincides with the value of $S_i$  at start time.
  At \refline{smr:exec:1}, variable $S_i$ is changed to ${S_i}'=\tau^*(S_i,\lambda_i').\stt$, with $\lambda_i'=\lambda_i \append c$.
  By induction, $S_i=\tau^*(\sttInit,\lambda_i).\stt$.
  It follows that:
  \begin{align*}
    {S_i}' &= \tau^*(S_i,\lambda_i \append c).\stt \\
    &= \tau^*(\tau^*(\sttInit,\lambda_i).\stt , c).\stt \\
    &= \tau^*(\sttInit,\lambda_i').\stt
  \end{align*}  
\end{proof}

\begin{proposition}
  \labprop{smr:3}
  $\forall \response{i}{c}{v} \in h.~ v = \tau^*(\sttInit, \lambda_i^{\run}\upto{c}).\val$.
\end{proposition}

\begin{proof}
  From \refline{smr:exec:3}, we have $v = \pending(c)$.
  The map $\pending(c)$ is set to $\top$ at \refline{smr:inv:1}.
  Process $i$ then awaits that $(\pending(c) \neq \top)$ holds at \refline{smr:inv:3}.
  As a consequence, $v$ is the result of the computation at \reflines{smr:exec:1}{smr:exec:3}.
  Let $\lambda$ be the value of $\lambda_i$ before this execution.
  Applying \refprop{smr:2} leads to $S_i = \tau^*(\sttInit,\lambda).\stt$.
  Thus, we have $v = \tau^*(\sttInit, \lambda \append c).\val$.
  By \refprop{smr:1}, $\lambda_i^{\run}\upto{c} = \lambda \append c$.
  Thus, the claim holds.
\end{proof}

The above proposition explains how the response values of $h$ are computed.
We now construct a linearization of the commands submitted to the replicated state-machine that is consistent with these return values.
This linearization is denoted $\delta$ and built as follows:
\begin{construction}
  \label{construction:1}
  Initially, $\delta$ is set to $1$.
  Let $E$ be the set $\bigunion_{j \in \procSet} \cmdOf{\lambda_j^{\run}}$.
  By the Ordering property of \SYS, the transitive closure of $\mapsto$ forms an order over $\cmdSet$.
  We append each command $c \in E$ to $\delta$ following some linear extension of this relation over $E$.
\end{construction}

\begin{proposition}
  \labprop{smr:4}
  $\forall i \in \procSet.~ \classOf{\lambda_i^{\run}} \prefix \classOf{\delta}$.
\end{proposition}

\begin{proof}
  For any $\lambda_i^{\run}$, we have $\cmdOf{\lambda_i^{\run}} \subseteq \cmdOf{\delta}$.
  Now consider a pair of non-commuting command $(c,d)$ in $\delta$, with $c <_{\delta} d$ and $d \in \lambda_i^{\run}$.
  Observe that if $c \notin \lambda_i^{\run}$ or $d <_{\lambda_i^{\run}} c$, then $d \mapsto_i c$ holds;
  thus, we have necessarily $c <_{\lambda_i^{\run}} d$.
  Applying \reflem{prel:5}, $\classOf{\lambda_i^{\run}} \prefix \classOf{\delta}$.
\end{proof}

Consider the complete, sequential and legal history $l$ produced by applying the commands in $\delta$ to $\sttInit$ following the order $<_{\delta}$.
For every pending command $c$ in $h$, if $c$ has no response $v$ in $h$, we append $\response{i}{c}{v}$ to $h$, where $i$ is the caller of $c$ and $v$ the response of $c$ in $l$.
Name $h'$ the resulting history that by construction completes $h$.

\begin{proposition}
  \labprop{smr:5}
  $l \equiv h'$
\end{proposition}

\begin{proof}
  By applying \refprop{smr:3}, \refprop{smr:4} and \reflem{prel:6}.
\end{proof}

\begin{proposition}
  \labprop{smr:6}
  $\hb_{h} \subseteq \hb_{l}$
\end{proposition}

\begin{proof}
  By construction of $l$ and the fact that $\hb_h \subseteq <_\lambda$.
\end{proof}

At the light of the last two propositions, we may conclude the result that follows.

\begin{theorem}
  \labtheo{smr:1}
  For every run $\run$ of \refalg{smr}, the history $h$ induced by $\run$ is linearizable.
\end{theorem}

\subsection{Non-fault-Tolerant Reads}
\labsection{smr:reads}

Command $\get()$ in the KVS use case of \refsection{evaluation:ycsb} belongs to a class of commands whose conflicts are \emph{transitive} \cite{epaxos}.
This means that for any such read $c$ and any two writes $d$ and $d'$, if $c \nonCommuting d \land c \nonCommuting d'$ then $d \nonCommuting d'$.
We denote by $\ReadStar$ the set of such reads.

In \refsection{execution:improvements}, we introduce the \readopt optimization.
This optimization skips the commands in $\ReadStar$ when computing conflicts and it allows the coordinator to use a fast quorum which consists of a plain majority.
As shown experimentally in \refsection{evaluation:ycsb}, this mechanism reduces dependencies for write commands and improves overall performance.

When the \readopt optimization is enabled, line~\ref{algo:full:fast-quorum} in Algorithm~\ref{algo:full} assigns a majority quorum if command $c$ belongs to $\ReadStar$.
Additionally, function $\conflicts$ is redefined to become: $\conflicts(c) = \{ \id \not \in \KwPhAset \mid \conflict(c, \KwMsgM[\id]) \land \KwMsgM[\id] \notin \ReadStar \}$.

\paragraph{Sketch of proof.}
\SYS with the \readopt optimization implements the SMR specification given in \refsection{smr} when the Ordering property is restricted to $\cmdSet \setminus \ReadStar$.
Moreover, if $c \hb_h d$, $c$ and $d$ do not commute, and $c$ is a write, then $d \mapsto_i c$ cannot hold at a process $i$.

\refprops{smr:1}{smr:3} do not change when \readopt is applied.
The word $\delta$ is built by first applying Construction~\ref{construction:1} to $E \setminus \ReadStar$.
Then, for each $c \in E \inter \ReadStar$, following some linearization of $\hb_h$ over $E \inter \ReadStar$, $\delta$ is extended as follows:
Name $i$ the caller of $c$.
We add $c$ after the last command in $\delta$ that either happens-before $c$ in $h$ or does not commute with $c$ and precedes it in $\lambda_i^{\run}$.

Then, \refprop{smr:4} is established for $\hat{\lambda_i^{\run}}$ and $\hat{\delta}$, where $\hat{x}$ the projection of $x$ over $\cmdSet \setminus \ReadStar$.
This implies that response value of some command in $c \in \cmdSet \setminus \ReadStar$ is the same in $l$ and $h'$.
To obtain the same result when $c \in \ReadStar$, we observe that the commands in $\lambda$ which do not non-commute with $c$ form a total order in $\mapsto$.
Thus, \refprop{smr:5} holds.
As previously, \refprop{smr:6} follows from the fact that $\hb_h \subseteq <_{\lambda}$.